\documentclass[aps,prxquantum,reprint,superscriptaddress,amsmath,amssymb]{revtex4-2}
\usepackage{blindtext}
\usepackage{mathtools}
\mathtoolsset{showonlyrefs}
\usepackage{amsmath}
\allowdisplaybreaks
\usepackage[colorlinks=true, linkcolor=black, urlcolor=black, citecolor=black]{hyperref}
\usepackage{amssymb}
\usepackage{amsthm}
\usepackage{dsfont}
\usepackage{comment}
\usepackage{bm}        
\usepackage{mathrsfs}  
\usepackage[dvipsnames]{xcolor}
\usepackage{titlesec}
\usepackage{tikz}
\usetikzlibrary{arrows.meta}
\usetikzlibrary{quantikz}
\usepackage{graphicx}
\usepackage{array}      
\usepackage{float}
\usepackage{ulem}
\usepackage{braket}
\usetikzlibrary{positioning,calc}
\usepackage{tcolorbox}
\usepackage{tasks}

\usepackage{appendix}

\allowdisplaybreaks

\newcommand{\MM}{\mathsf{M}}
\newcommand{\pp}{\mathsf{P}}
\newcommand{\e}{\mathsf{E}}
\newcommand{\TT}{\mathsf{T}}

\newcommand{\rank} {\operatorname{rank}}

\newcommand{\M}{{\cal M}}
\newcommand{\T}{{\cal T}}
\newcommand{\PP}{{\cal P}}
\newcommand{\EE}{{\cal E}}

\newcommand{\St}{{\cal S}}
\newcommand{\V}{{\cal V}}
\newcommand{\inprod}[2] {\langle #1 , #2 \rangle}

\newcommand{\poly}{\operatorname{poly}}

\newtheorem{theorem}{Theorem}
\newtheorem{lemma}[theorem]{Lemma} 
\newtheorem{corollary}[theorem]{Corollary} 
\theoremstyle{definition} 

\newtheorem*{ass*}{Assumption}
\newtheorem{ass}{Assumption}
\usepackage{algcompatible}
\usepackage[ruled,vlined,linesnumbered]{algorithm2e}
\usepackage{stfloats}

\floatname{algorithm}{Example}

\theoremstyle{definition}
\newtheorem{example}{Example}
\newtcolorbox{exbox}{
    colback=black!5, 
    colframe=white, 
    breakable 
}
\tcbuselibrary{breakable}

\titlespacing*{\subsection}{0pt}{\baselineskip}{\baselineskip}

\begin{document}
\title{Linear Algebra of Generalized Contextuality in All Prepare-Transform-Measure Scenarios}
\author{Theodoros Yianni}
\affiliation{Department of Computer Science, Royal Holloway, University of London, United Kingdom}
\author{Nyan Raess}
\affiliation{Department of Computer Science, Royal Holloway, University of London, United Kingdom}
\author{Farid Shahandeh}
\affiliation{Department of Computer Science, Royal Holloway, University of London, United Kingdom}

\begin{abstract}
    Generalized contextuality is a canonical distinguishing property of 
    nonclassical generalized probabilistic theories, in particular quantum mechanics.
    Methods for certification and characterization of generalized contextuality of a given generalized probabilistic theory are well developed for prepare-measure and
    single-stage prepare-transform-measure scenarios.
    In a recent work \href{https://arxiv.org/abs/2512.10000}{[arXiv:2512.10000]}, a bottom-up, statistics-first linear-algebraic framework for contextuality in prepare-measure scenarios was introduced.
    We extend 
    this approach to operational scenarios with sequential transformations with an arbitrary number of stages. 
    We give a full decision procedure for contextuality of such scenarios within operational theories and analyze its computational complexity.  In particular, our decision procedure has a complexity linearly exponential in the minimum generalized probabilistic theory (GPT) dimension, and polynomial in the number of procedures.
    We demonstrate our framework and approach through multiple examples, including Spekkens' toy theory and the 8-state single-qubit stabilizer theory.
    In particular, we construct
    an operational theory in which contextuality manifests itself only in
    the sequential structure of the transformations.
    Our findings thus shed new light on the significant role of compositional structures in the phenomenon of generalized contextuality.
\end{abstract}
\maketitle

\section{Introduction}
A central problem of quantum foundations is explaining the separation between classical and nonclassical theories. A canonical distinguishing property of the latter is contextuality. There are several widely used notions and frameworks of contextuality, including Kochen-Specker contextuality~\cite{Koc}, the sheaf-theoretic framework~\cite{Abramsky_2011}, contextuality-by-default~\cite{DzhafarovKujala2016}, state-independent contextuality~\cite{Cabello2008}, and graph-theoretic formulations~\cite{CabelloSeveriniWinter2014}. Interestingly, Bell’s work on nonlocality~\cite{Bell1964,Bell1966}
can also be viewed through the lens of contextuality~\cite{Abramsky_2011,Budroni2022}.

In the current work, we adopt the notion of \textit{generalized contextuality}~\cite{Spekkens_2005}.
It denotes the impossibility of assigning an ontological representation to the experimental procedures in an operational theory, such that the statistical predictions of the theory are reproduced, and operationally indistinguishable procedures are represented identically~\cite{Spekkens_2005,Kunjwal2015}.
This property is known to underlie quantum advantage in various information-processing tasks, encompassing quantum computations~\cite{computation,shahandeh2021advantage,Schmid2022Stabilizer} and quantum communications~\cite{Hameedi2017, Saha_2019,Roy2024,Ambainis2016}.

Many information-processing tasks of interest are naturally described as sequential compositions of operations. In such settings, classical simulability and classicality are often not properties of any single step in isolation, but rather of how the steps compose. 
In such protocols, it may be that each individual stage admits a classical explanation, even though the overall process does not.
Conversely, one can also devise protocols in which the compositional process does not admit a classical explanation, even though such an explanation exists when all the transformations are treated as a single, combined operation.
This motivates studying contextuality in scenarios with multiple transformation stages, where one can ask how contextuality behaves under concatenation and coarse-graining. 
In particular, computation is implemented by circuits whose power is determined by the composition of stages, making multi-stage contextuality a natural candidate property for explaining circuit-level quantum advantage~\cite{shahandeh2021advantage, Howard_2014, PhysRevLett.121.230401}.

Despite its relevance to sequential scenarios, contextuality has been most extensively studied in the simplest possible setting, namely, the prepare-measure (PM) scenario, in which the theory contains a finite number of ways to prepare and measure the system. Many tools already exist to probe the contextuality of PM scenarios. Typically, these rely
on a pre-defined \textit{generalized probabilistic theory} (GPT) model of the operational theory.
One can then use linear programming to decide whether a noncontextual ontological model (NCOM) of the GPT exists~\cite{Selby_2024,Gitton_2022}. 
A contrasting approach employed in Ref.~\cite{NCOM-OPTs} also discussed noncontextuality of operational probabilistic theories.
However, in Refs.~\cite{ourpaper,shahandeh2025unifiedlinearalgebraicframework}, a bottom-up, statistics-first framework was developed to characterize contextuality of PM scenarios within operational theories.
In this framework, the statistical structure of
a theory is represented in a matrix $C$ of \textit{conditional outcome-probability of events} (COPE matrix). Then it was shown
that contextuality of a PM scenario relies on a single rank-based criterion of $C$, without reference to an underlying GPT.

In this work, we develop a rank-based characterization of PTM scenarios with an arbitrary number of transformation stages by generalizing the COPE formalism of Refs.~\cite{ourpaper,shahandeh2025unifiedlinearalgebraicframework}.
In doing so, we replace the COPE matrix with a COPE tensor of appropriate order.
We demonstrate, among other things, that ontological models of PTM scenarios specify factorizations of the COPE tensor. 
We then show that an ontological model is noncontextual if and only if the matrix factors satisfy a set of rank constraints.

We provide an algorithm for checking whether the criteria are satisfiable and perform a complexity analysis that precisely characterizes the efficiency of checking the sequential criteria.
More specifically, we show that noncontextuality of any PTM scenario with an arbitrary number of transformations in an operational theory can be decided in polynomial time in the number of states, transformations, and measurement outcomes, for fixed minimal GPT dimension.
However, the complexity is exponential in the ranks of the COPE tensor mode flattenings.

A recent study by Schmid \textit{et al.}~\cite{Schmid_2024} characterized the contextuality of GPTs in the prepare-transform-measure scenario (PTM) with a single transformation stage. Similar to typical PM scenario methods, this relied on an underlying GPT model and linear programming.  
Moreover, the composition of multi-stage transformations was not addressed.
More broadly, Ref.~\cite{Schmid2024} developed a process-theoretic framework for generalized contextuality in arbitrary compositional operational theories and proved a structure theorem for it.
Our approach differs from these works
in three key respects:
we are able to determine the contextuality of scenarios with an arbitrary number of transformations; we only require the statistics of the operational theory, together with minimal assumptions, to do so; and we give a precise account of the problem's complexity.

We illustrate our methods on a range of worked examples.
First, we recover the expected behaviour for Spekkens' toy theory, as well as an 8-state theory for single-qubit stabilizers with a single stage of transformations. For the latter, the rank criterion detects only the presence of transformation contextuality, corroborating the result of Ref.~\cite{Lillystone2019}.
We also contribute a 4-state toy theory
in which consecutive transformations induce contextuality which is not present in the single-transformation-stage setting, demonstrating that sequential structure is a source of contextual behavior. Finally, we revisit the 8-state theory to examine how the contextuality of a scenario responds to restrictions on the available transformation procedures.

Our results provide a general framework for detecting contextuality in operational scenarios that go beyond the prepare-measure setting, including multi-stage transformation structures.  We view this as a step toward a circuit-level understanding of nonclassicality, and 
connecting operational notions of contextuality with known requirements for quantum advantage in information processing.

The paper is organized as follows. In Sec.~\ref{sec:OP_formalism} we give an overview of the preliminary concepts and definitions, including operational theories and the COPE tensor formalism. In Sec.~\ref{sec:models}, we introduce the different types of models of operational theories, including GPTs and NCOMs. We then give an explicit construction of a GPT for a single-stage $\pp\TT^1\MM$, and show that no smaller GPT model exists. We subsequently devise a method of deciding if an NCOM exists in this scenario using our formalism, and constructing one if it does. In Sec.~\ref{sec:sequential}, we consider scenarios with sequential  transformations, and once again construct the smallest possible GPT. Moreover, we generalize our method for deciding contextuality to sequential PTM scenarios.
Finally, in Sec.~\ref{sec:application}, we use our rank criteria to characterize the contextuality of specific toy models for scenarios including transformations.
Discussions and conclusions are presented in Sec.~\ref{sec:conclusion}.

\section{Formalism} \label{sec:OP_formalism}
The groundwork for the fundamental concept in this paper, namely representation-free \textit{operational theories}, has been laid out in Refs.~\cite{ourpaper,shahandeh2025unifiedlinearalgebraicframework}.
In this picture, each
operational theory has three structures: an \textit{operational} one, a \textit{causal} one (including sequential and parallel compositions), and a \textit{probabilistic} one. 
In this section, we recap these structures, focusing primarily on the causal structure and, above all, on the sequential composition of transformations.

\subsection{Operational Structure}

The primitive elements of operational theories are lists of laboratory instructions.
For example, ``turn the laser on AND adjust its intensity as such'' is an instruction for preparing a photonic system.
Similarly, propositions such as ``turn the detector on AND align its aperture as such'', associated with measurements of the system, are measurement instructions.
Intermediate propositions such as ``add a polarizer at this point in the photons' path'' are instructions for transforming such a system. The composition of valid propositions in an operational theory produces a valid sentence in the theory.

The scenarios we consider thus refer to the compositional structure of the sentences in the operational theory. In particular, the PTMs consist of a preparation, a single transformation, and a measurement. These are ternary propositional schemas in the theory.
The operational structure thus consists of the prescriptions (propositions) for preparing, transforming, and measuring the system, denoted by the sets $\PP:=\{\pp_i\}$, $\T^1:=\{\TT^1_i\}$, and $\M:=\{\MM_j\}$, respectively.
Later, we will extend this to an arbitrary number of transformation stages ($n$-ary propositional schemas)
where the sets of transformation instructions are given by $\T^1:=\{\TT^1_i\}$, $\T^2:=\{\TT^2_i\}$, ..., $\T^m:=\{\TT^m_i\}$.
In all scenarios, a measurement $\MM_j$ results in an outcome event $\e^j_k$ from a list of possible outcomes indexed by $\EE = \{\e^j_k\}$. In our approach, the sets $\PP$, $\T^i$ and $\M$ are the \textit{atomic} propositions from which all sentences and probabilistic combinations are constructed.

\subsection{Probabilistic Structure}
The probabilistic structure of PTM scenarios in an operational theory is provided by a probability measure $p$ that takes every possible PTM sentence to a real number in the unit interval.
The image of $p$ is thus the collection of conditional outcome probabilities associated with each sentence in the theory.
We can organize these probabilities into a tensor $C$, the entries of which are the probabilities $p(k|i,j,l):=p(\e^l_k|\pp_i,\TT^1_j, \MM_l)$ of each outcome $\e^l_k$, given that the system was prepared, transformed and measured according to $\pp_i,\TT^1_j$ and $\MM_l$ respectively. In a $\pp\TT^1\MM$ scenario, $C$ is therefore an $n_\e^\MM n_\MM \times n_{\TT^1} \times n_\pp$ 3-tensor of \textit{conditional outcome probabilities of events} (COPE), where we assume $n_\pp$ preparations, $n_{\TT^1}$ transformations, $n_\MM$ measurements, and without loss of generality that each of the $n_\MM$ measurements has $n_\e^\MM$ outcome events. For simplicity, we denote by $n_\e$ the total number of outcome events.

A COPE 3-tensor is thus the natural extension of the COPE \textit{matrix} of a prepare-measure scenario introduced in Refs.~\cite{ourpaper,shahandeh2025unifiedlinearalgebraicframework}.
We have depicted the schematic of a COPE 3-tensor in Fig. \ref{fig:placeholder}. 
The tensor is divided into column stochastic blocks corresponding to each measurement setting, which are stacked on top of each other.
Even though each probability depends on four indices, we have introduced a joint index which runs over measurements and outcome events, which allows us to represent the COPE tensor as a 3-tensor. Accordingly, we will often label effects simply as $\e_k$, where $k$ is the joint index; we will list them as $\e_k^j$ only when we wish to clarify that an effect is associated with the $j$th measurement $\MM_j$.

Each dimension of the tensor is referred to as a \textit{mode}, and fixing any index produces a slice of the tensor in that mode. For instance, 
each choice of transformation $\TT^1_j$ corresponds to an $n_\e \times n_\pp$ COPE matrix, which is a mode-$\TT^1$ slice given by
\begin{gather}\label{eq:COPE}
C_{:j:} := 
\begin{pmatrix}
\begin{array}{c}
   \begin{matrix}
p(1|1,j,1) & \cdots & p(1|I,j,1) \\
        \vdots &   & \vdots  \\
        p(n_1|1,j,1) & \cdots & p(n_1|I,j,1)
   \end{matrix}\\ 
   \hline
    \vdots \\
   \hline
   \begin{matrix}
        p(1|1,j,K) & \cdots & p(1|I,j,K)\\
        \vdots & & \vdots \\
        p(n_K|1,j,K) & \cdots & p(n_K|I,j,K)
   \end{matrix} \\
\end{array}    
\end{pmatrix}.
\end{gather}
We can see that each of these slices can be thought of as representing a prepare-measure experiment, where the fixed transformation has been absorbed into either the preparations or measurements. Similarly, each mode-$\pp$ slice represents a fixed choice of preparation $\pp$, and each mode-$\e$ slice represents a fixed choice of outcome event $\e$. 

\begin{figure}[h]
    \centering
    \includegraphics[width=0.8\linewidth]{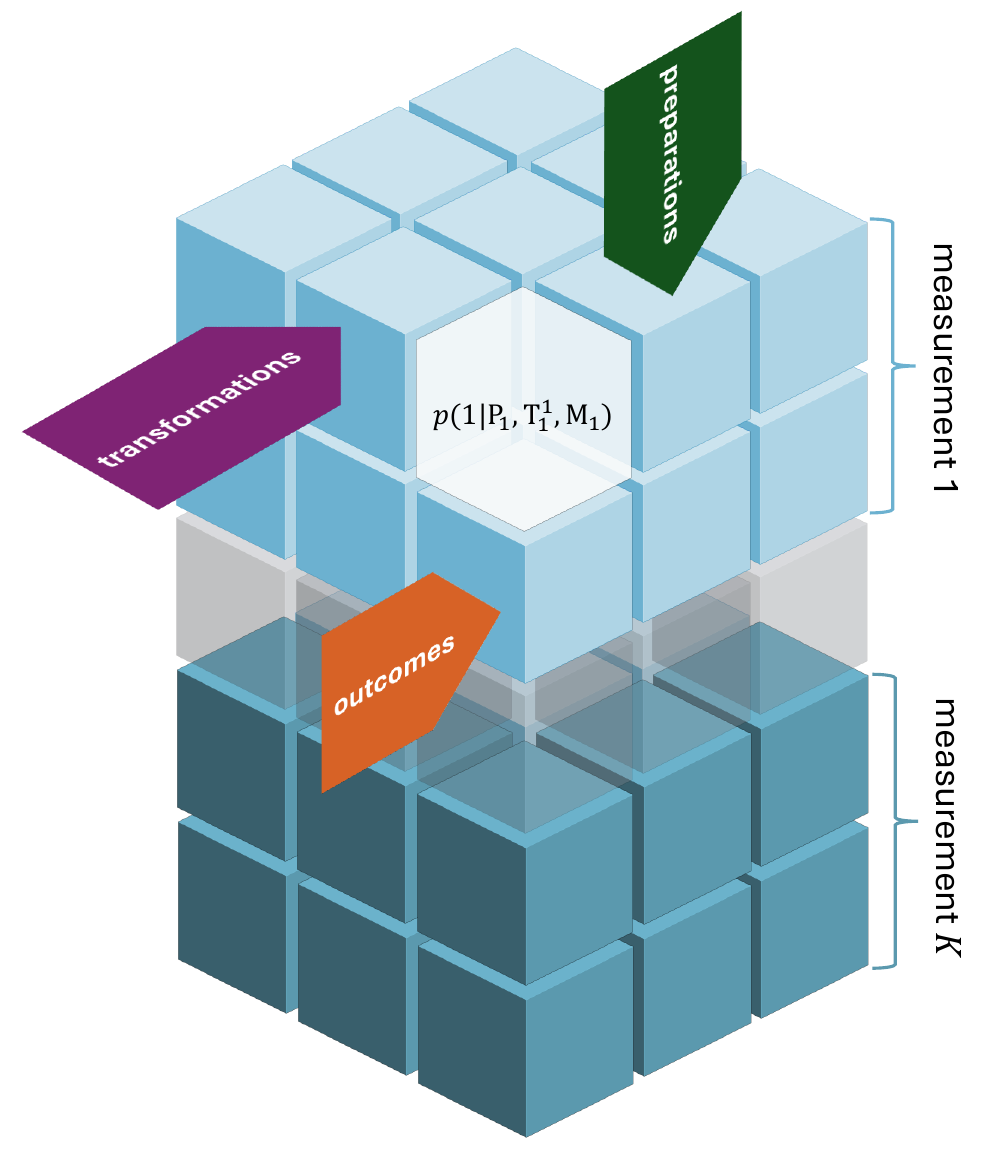}
    \caption{Diagram of a COPE 3-tensor. Each axis corresponds to a different set of atomic propositions. Each entry is the conditional-outcome probability associated with a particular choice of atomic propositions, i.e. sentence, in the operational theory. }
    \label{fig:placeholder}
\end{figure}

The assumptions used to define the probabilistic structure of the $\pp\TT^1\MM$ scenario are the natural extensions of 
those used to define the probabilistic structure of the PM scenario in Ref.~\cite{shahandeh2025unifiedlinearalgebraicframework}.
In particular, we have the following:
\begin{ass}\label{ass:COPE_comp}
    The COPE tensor contains the complete probabilistic structure of the operational theory.
\end{ass}
\noindent In the following, we employ a 2D toy theory as a running example to clarify the concepts and approach.

\begin{exbox}
\begin{example}\label{ex:Spekkens_COPE}                    
\textbf{2D toy theory:
Probabilistic structure.}                                    
Imagine a world which accommodates an \textit{elementary} system:
There are
four propositions describing how to prepare the system, four propositions describing how to transform the system after preparation, and one proposition for measuring it, which yields four outcomes.
We thus have $\PP=\{\pp_1,\pp_2,\pp_3,\pp_4\}$, $\mathcal{T}^1=\{\TT^1_1,\TT^1_2,\TT^1_3,\TT^1_4\}$ and $\mathcal{E}=\{\e_1,\e_2,\e_3,\e_4\}$.
The COPE tensor has mode-$\TT^1$ slices given by
\begin{equation}\label{eq:STT_COPE}
\begin{split}
 C^{2D}_{:1:}=\frac{1}{8}
\begin{pmatrix}
1 & 3 & 3 & 1\\
1 & 1 & 3 & 3\\
3 & 1 & 1 & 3\\
3 & 3 & 1 & 1
\end{pmatrix},
\end{split}
\end{equation}
\begin{equation}
\begin{split}
C^{2D}_{:2:}=\frac{1}{8}
\begin{pmatrix}
3 & 3 & 1 & 1\\
1 & 3 & 3 & 1\\
1 & 1 & 3 & 3\\
3 & 1 & 1 & 3
\end{pmatrix},
\end{split}
\end{equation}
\begin{equation}
\begin{split}
C^{2D}_{:3:}=\frac{1}{8}
\begin{pmatrix}
3 & 1 & 1 & 3\\
3 & 3 & 1 & 1\\
1 & 3 & 3 & 1\\
1 & 1 & 3 & 3
\end{pmatrix},
\end{split}
\end{equation}
\begin{equation}
\begin{split}
C^{2D}_{:4:}=\frac{1}{8}
\begin{pmatrix}
1 & 1 & 3 & 3\\
3 & 1 & 1 & 3\\
3 & 3 & 1 & 1\\
1 & 3 & 3 & 1
\end{pmatrix}.
\end{split}
\end{equation}
Each of these matrices encodes the statistics associated with all preparations and outcome events for a given transformation.
We note that each slice is given by cyclic shifts of the rows of the previous one.
Note also that all rows and columns of each slice are convexly independent as required by a COPE matrix~\cite{shahandeh2025unifiedlinearalgebraicframework}.
\end{example}
\end{exbox}

\subsection{Operational equivalences}

From the probabilistic structure, we define a notion of operational equivalence of procedures.
The operational equivalence $(\simeq)$ of preparations (measurement events) means that there are no combinations of transformations and measurement events (preparations) in the operational theory that can distinguish them~\cite{Spekkens_2005}, i.e.,
\begin{equation} \label{eq:preparationmeasurementeqv}
\begin{split}
\pp_i \simeq \pp_j \; \text{iff} \; p(\e| \pp_i, \TT^1, \MM) \!&=\! p(\e|\pp_j,\TT^1, \MM) \quad \forall \  \e, \MM, \TT^1,\\
\e_k \simeq \e_l \; \text{iff} \; p(\e_k| \pp,\TT^1, \MM) \!&=\! p(\e_l|\pp,\TT^1, \MM) \quad \forall \  \pp,\TT^1.
\end{split}
\end{equation}
Similarly, two transformations are operationally equivalent if and only if there are no combinations of preparations and measurements that can distinguish them, i.e.,
\begin{equation} \label{eq:transformationeqv}
\begin{split}
\TT^1_i \simeq \TT^1_j \; \text{iff} \; p(\e| \pp, \TT^1_i, \MM) = p(\e|\pp, \TT^1_j, \MM) \quad \forall \  \pp, \e, \MM.
\end{split}
\end{equation}
It follows that two operationally equivalent procedures give rise to two identical slices in the theory's COPE tensor.
Then, Assumption~\ref{ass:COPE_comp} leads to the following observation.
\begin{corollary}\label{cor:IDRC}
    For two identical slices of the COPE tensor of an operational theory, corresponding to two choices of the same procedure type, there exist no combinations of the other procedure types that can separate them.
\end{corollary}

\begin{exbox}
\begin{example}\label{ex:equive_ops}
\textbf{2D toy theory:
Operational equivalences.}          
    Consider the slices in Example~\ref{ex:Spekkens_COPE}, assuming that the set of allowed transformations is closed under convex combinations. 
    We have,
\begin{equation}
    \frac{1}{2}( C^{2D}_{:1:}+ C^{2D}_{:3:})=\frac{1}{2}( C^{2D}_{:2:}+ C^{2D}_{:4:}),
\end{equation}
implying that,
\begin{equation}
    \frac{1}{2}(\TT^1_1+\TT^1_3)\simeq\frac{1}{2}(\TT^1_2+\TT^1_4).
    \label{eq:ex2equiv}
\end{equation}
The procedures on the LHS and RHS of~\eqref{eq:ex2equiv} amount to randomly selecting and implementing $\TT^1_1$ or $\TT^1_3$, and $\TT^1_2$ or $\TT^1_4$ respectively.
\end{example}
\end{exbox}

\section{Models of COPE 3-Tensors} \label{sec:models}

Throughout this work, following the approach introduced in Ref.~\cite{shahandeh2025unifiedlinearalgebraicframework}, we use the term model to denote any mathematical framework that reproduces the statistics of an operational theory. Accordingly, we assign mathematical objects to propositions of the operational theory, together with composition rules mirroring those of the \textit{scenarios}--theory’s schemas.
In this section, we introduce three different model types and analyze their relationships. These are then used to study the nonclassicality of the operational theory.

\subsection{preGPT and GPT Models of $\pp\TT^1\MM$ Scenario}

\label{sec:preGPTPTM}

The first two models we consider are preGPTs and GPTs, which we extend from Ref.~\cite{ourpaper}. 
A preGPT for a $\pp\TT^1\MM$ scenario is the tuple $(\EE_{\rm pGPT},\T^1_{\rm pGPT},\St_{\rm pGPT},\V,\inprod{\cdot}{\cdot})$ where $\V$ is an ordered inner-product vector space with inner product $\inprod{\cdot}{\cdot}$, and $\EE_{\rm pGPT}$, $\T^1_{\rm pGPT}$ and $\St_{\rm pGPT}$ are its sets of effects, transformations, and states, respectively, constructed as follows.

We define the function $\mathcal{K}_\e$ which maps each event $\e_k$ to a vector $\mathcal{K}_\e(\e_k)$ in $\mathcal{V}$. A collection of these vectors belonging to a fixed measurement form a {\it probability vector-valued measure} (PVVM) satisfying
\begin{equation} \label{eq:PVVM}
\begin{split}
    & \forall \  \e^j_k\in\mathcal{E} \quad \mathcal{K}_{\e}(\e^j_k){\geqslant} 0,\\
    & \exists ! u \in \V \text{~s.t.~} \sum_{k} \mathcal{K}_{\e}(\e^j_k){=}u,\\
    &  \forall \  j \quad \mathcal{K}_{\e}({\cup}_k \e^j_k){=}\sum_k \mathcal{K}_{\e}(\e^j_k).
\end{split}    
\end{equation}
Each vector $\mathcal{K}_{\e}(\e^j_k)$ is called an effect, $u$ is called the unit effect of the preGPT, and the collection of all operationally legitimate effects is denoted by $\EE_{\rm pGPT}$.
Further, by the convexity of the set of all events $\mathcal{E}$ and the linearity of the map $\mathcal{K}_{\e}$, it is assumed that $\EE_{\rm pGPT}$ is convex. 
The set of extremal points of $\EE_{\rm pGPT}$ is given by ${\rm ex}\EE_{\rm pGPT}$, and similarly for $\T^1_{\rm pGPT}$ and $\St_{\rm pGPT}$.

We correspondingly define $\mathcal{K}_{\pp}$ as the function that maps preparations of the operational theory to linear functionals in $\V^*$, the dual of $\V$.
These vectors are called the preGPT states.
It is assumed that the states are normalized in the sense that $\inprod{u}{\mathcal{K}_{\pp}(\pp)}=1$.
Finally, denote by $\mathcal{K}_{\TT^1}$ the function that maps transformations of the operational theory to linear transformations over $\V^*$.
We assume that
${\mathcal{K}_{\TT^1}(\TT^1_j)[\mathcal{K}_{\pp}(\pp_i)]}$ is a valid state for every transformation $\TT^1_j$ and every preparation $\pp_i$.

The preGPT states, effects, and transformations are denoted by,
\begin{equation}
    \begin{split}
        &\mathcal{K}_{\pp}(\pp_i)=s^i,\\
        &\mathcal{K}_{\e}(\e^l_k)=e^{lk},\\
        &\mathcal{K}_{\TT^1}(\TT^1_j)=T^{1j},
    \end{split}
\end{equation}
respectively.
To clarify, we use upper indices to label elements of the operational theory, while lower indices act as matrix indices. For instance, $T^{1j}_{ik}$ is the element at position $(i,k)$ in the matrix $T^{1j}$, which represents the $j$th transformation in the first stage.

The probability of an outcome event is given by the inner product
\begin{equation} \label{eq:GPT_pob_rule}
\begin{split}
    p(k|\pp_i,\TT^1_j,\MM_l) &= \inprod{\mathcal{K}_{\e}(\e^l_k)}{\mathcal{K}_{\TT^1}(\TT^1_j)[\mathcal{K}_{\pp}(\pp_i)]}\\
    &:=\inprod{e^{lk}}{T^{1j}[s^i]}.
\end{split}
\end{equation}
By choosing an orthonormal basis for $\mathcal{V}$
we write the probability rule as,
\begin{equation}
    p(k|\pp_i,\TT^1_j,\MM_l) = e^{lk\top} T^{1j} s^{i}.
\end{equation}
Furthermore, by equivalently considering transformations as acting on effects rather than states, we require that transformations preserve the unit effect:
\begin{equation}
    u^\top T^{1j}=u^\top.
\end{equation}
Importantly, the probability is linear in the state vectors, transformations, and effect vectors.
The maps from procedures to preGPT elements are necessarily linear, although they can be set-valued~\cite{shahandeh2025unifiedlinearalgebraicframework}. That is, rather than mapping procedures to unique preGPT elements, they may also map them to subsets of these elements.  For a procedure $\mathsf{O} \in\{\mathcal{P},\mathcal{T},\mathcal{E}\}$ which is probabilistically dependent on other procedures of the same type, for instance, $\mathsf{O}=a\mathsf{O}_j+(1-a)\mathsf{O}_k$, it is always possible to construct a preGPT model that preserves convexity as $\mathcal{K}_{\mathsf{O}}(\mathsf{O}_i)=a\mathcal{K}_{\mathsf{O}}(\mathsf{O}_j)+(1-a)\mathcal{K}_{\mathsf{O}}(\mathsf{O}_k)$. 

We are now in a position to show how a preGPT model of a PTM scenario is constructed from the theory's COPE tensor. For simplicity denote the effects by $e^i$, where it is understood that the single index $i$ runs over outcomes of all measurements. Defining the matrices of effects, $A_{i:}=e^{i\top}$, and states, $B_{:i}=s^i$, a preGPT of the scenario is given by a decomposition of the COPE tensor $C$ as
\begin{equation}\label{eq:C_dec_PTM}
    C_{ijl}=A_{i:}T^{1j}B_{:l}.
\end{equation}

\begin{exbox}
\begin{example}\label{ex:rebit_preGPT}    
\textbf{2D toy theory: PreGPT}          
Consider the COPE tensor $C^{2D}$ of Example~\ref{ex:Spekkens_COPE}.
A possible preGPT model is given by,
\begin{equation*}
    A =
\begin{pmatrix}
1 & 0 & 0 & 0 & \tfrac{1}{2}\\[2pt]
0 & 1 & 0 & 0 & -\tfrac{1}{2}\\[2pt]
0 & 0 & 1 & 0 & \tfrac{1}{2}\\[2pt]
0 & 0 & 0 & 1 & -\tfrac{1}{2}
\end{pmatrix},
\end{equation*}
\begin{equation*}
    B =
\begin{pmatrix}
1 & 0 & 0 & 0\\[2pt]
0 & 1 & 0 & 0\\[2pt]
0 & 0 & 1 & 0\\[2pt]
0 & 0 & 0 & 1\\[2pt]
\tfrac{1}{2} & \tfrac{1}{2} & -\tfrac{1}{2} & -\tfrac{1}{2}
\end{pmatrix},
\end{equation*}
\begin{equation*}
    T^{11} =
\begin{pmatrix}
\frac{5}{32} & \frac{13}{32} & \frac{11}{32} & \frac{3}{32} & -\frac{1}{16}\\[2pt]
\frac{3}{32} & \frac{3}{32} & \frac{13}{32} & \frac{13}{32} & \frac{1}{16}\\[2pt]
\frac{13}{32} & \frac{5}{32} & \frac{3}{32} & \frac{11}{32} & -\frac{1}{16}\\[2pt]
\frac{11}{32} & \frac{11}{32} & \frac{5}{32} & \frac{5}{32} & \frac{1}{16}\\[2pt]
-\frac{1}{16} & -\frac{1}{16} & \frac{1}{16} & \frac{1}{16} & \frac{1}{8}
\end{pmatrix},
\end{equation*}
\begin{equation*}
    T^{12} =
\begin{pmatrix}
\frac{5}{12} & \frac{5}{12} & \frac{1}{12} & \frac{1}{12} & -\frac{1}{12}\\[2pt]
\frac{1}{12} & \frac{1}{3} & \frac{5}{12} & \frac{1}{6} & \frac{1}{12}\\[2pt]
\frac{1}{6} & \frac{1}{6} & \frac{1}{3} & \frac{1}{3} & -\frac{1}{12}\\[2pt]
\frac{1}{3} & \frac{1}{12} & \frac{1}{6} & \frac{5}{12} & \frac{1}{12}\\[2pt]
-\frac{1}{12} & -\frac{1}{12} & \frac{1}{12} & \frac{1}{12} & \frac{1}{6}
\end{pmatrix},
\end{equation*}
\begin{equation*}
    T^{13} =
\begin{pmatrix}
\frac{7}{16} & \frac{3}{16} & \frac{1}{16} & \frac{5}{16} & -\frac{1}{8}\\[2pt]
\frac{5}{16} & \frac{5}{16} & \frac{3}{16} & \frac{3}{16} & \frac{1}{8}\\[2pt]
\frac{3}{16} & \frac{7}{16} & \frac{5}{16} & \frac{1}{16} & -\frac{1}{8}\\[2pt]
\frac{1}{16} & \frac{1}{16} & \frac{7}{16} & \frac{7}{16} & \frac{1}{8}\\[2pt]
-\frac{1}{8} & -\frac{1}{8} & \frac{1}{8} & \frac{1}{8} & \frac{1}{4}
\end{pmatrix},
\end{equation*}
\begin{equation}
    T^{14} =
\begin{pmatrix}
\frac{1}{4} & \frac{1}{4} & \frac{1}{4} & \frac{1}{4} & -\frac{1}{4}\\[2pt]
\frac{1}{4} & 0 & \frac{1}{4} & \frac{1}{2} & \frac{1}{4}\\[2pt]
\frac{1}{2} & \frac{1}{2} & 0 & 0 & -\frac{1}{4}\\[2pt]
0 & \frac{1}{4} & \frac{1}{2} & \frac{1}{4} & \frac{1}{4}\\[2pt]
-\frac{1}{4} & -\frac{1}{4} & \frac{1}{4} & \frac{1}{4} & \frac{1}{2}
\end{pmatrix}.
\end{equation}

The unit effect of this preGPT is given by the sum of the rows of $A$, which is $u=\begin{pmatrix}
    1&1&1&1&0
\end{pmatrix}$.
\end{example}
\end{exbox}

A common formalism in the multilinear algebra and tensor decomposition literature is to factorize tensors in Tucker form.
A Tucker decomposition of a 3-tensor $C \in \mathbb{R}^{n_\e \times n_{\TT^1} \times n_\pp}$ is
\begin{equation}
C = \mathcal{G}\times_1 U \times_2 V \times_3 W,
\end{equation}
where $\mathcal{G}\in\mathbb{R}^{d_1\times d_2\times d_3}$ is called the \textit{core tensor} and
$U\in\mathbb{R}^{n_\e\times d_1}$, $V\in\mathbb{R}^{n_{\TT^1}\times d_2}$, and $W\in\mathbb{R}^{n_\pp\times d_3}$ are the factor matrices.
Equivalently,
\begin{equation} \label{eq:tucker}
C_{ijl}=\sum_{\alpha=1}^{d_1}\sum_{\beta=1}^{d_2}\sum_{\gamma=1}^{d_3}
g_{\alpha\beta\gamma}\,U_{i\alpha}\,V_{j\beta}\,W_{l\gamma}.
\end{equation}
The preGPT in Eq.~\eqref{eq:C_dec_PTM} is thus 
a Tucker decomposition of the form of Eq.~\eqref{eq:tucker} wherein
\begin{equation}
\begin{split}
  A&=U, \quad  T^{1j}_{\alpha\gamma} := \sum_\beta  g_{\alpha\beta\gamma}V_{j\beta}, \quad  B=W^\top.
  \label{eq:preGPTCOPEfactor}
  \end{split}
\end{equation}
The rows of $A$ are the effects, each $T^{1j}$ is a transformation matrix, and the columns of $B$ are the states. Conversely, any Tucker decomposition of the form given in Eq.~\eqref{eq:tucker} can be converted into a preGPT model as demonstrated below.
\begin{lemma}
\label{lemma:preGPT_Tucker}
Consider a COPE $3$-tensor $C$ with a Tucker decomposition of the form given in Eq.~\eqref{eq:tucker}. This decomposition can always be converted into a preGPT model for the $\pp\TT^1\MM$ scenario.
\end{lemma}
\begin{proof}
We begin by defining the matrices $A$, $B$, and $T^{1j}$ as in Eq.~\eqref{eq:preGPTCOPEfactor} and let the outcomes be indexed by
$(l,l_k)$, where $l$ denotes a measurement and $l_k$ its $k$-th outcome.
We denote the set of outcomes of measurement $l$ as $O_l$.
Following Eq.~\eqref{eq:preGPTCOPEfactor}, the COPE tensor's Tucker decomposition can be written as
\begin{equation}
    \begin{split}
        C_{(l,l_k)ji}&=A_{(l,l_k):}T^{1j}B_{:i},\\
        &=A_{(l,l_k):}B'_{:(j,i)},
    \end{split}
\end{equation}
wherein,
\begin{equation} \label{eq:B'}
    B'_{:(j,i)}:=T^{1j}B_{:i}.
\end{equation}

The sum of effects for measurement $l$ is given by
\begin{equation}
u^{l\top}:=\sum_{l_k\in O_l}A_{(l,l_k):}.
\end{equation}
For any preparation and transformation, the sum of all measurement outcomes associated with measurement $l$ produces a row of ones in the COPE, that is,
$u^{l \top} B'=\mathbf{1}^\top$ for all $l$.
We now choose a reference measurement $l_0$ and set $u:=u^{l_0}$ and translate the rows of $l$th block of $A$ as \begin{equation}A_{(l,l_k):}\mapsto A_{(l,l_k):}+\frac{1}{|O_l|}(u-u^l)^\top.\end{equation}
Since
$(u-u^l)^\top B'=0$
for all $l$, the product
$A  B'$
and as a result, the COPE tensor, remains unchanged.
This establishes $u$ as the common unit vector of all measurements.

We now modify the transformation matrices so that they all map the unit vector $u$ to a common output vector. To do so, we first consider the vector $v^\top=u^\top T^{1j_0}$, for some choice of reference transformation $T^{1j_0}$, which satisfies $v^\top B=\mathbf 1^{\top}$. We then choose a column vector $w$ such that $u^\top w=1$ and introduce shifted transformation matrices $\widetilde T^{1j}:=T^{1j}+w(v^\top-u^\top T^{1j})$.
Clearly, this has no effect on the COPE tensor since $(v^\top-u^\top T^{1j})B=0$.
The new transformation matrices $\widetilde{T}^j$ all act identically on the unit $u$, as clearly we have $u^\top \widetilde T^{1j}=v^\top$
for all $j$.

To ensure that the transformation matrices preserve the unit effect $u$ rather than mapping it to $v$, they must be made square. As the vectors $u$ and $v$ are of length $d_1$ and $d_3$ respectively, we define
$D:=\max\{d_1,d_3\}$, and subsequently duplicate columns of  $A$, and $T^1$s, as well as append zero rows to the $T^1$s and $B$, to get $\widehat A$ with $D$ columns, $\widehat B$ with $D$ rows, and $\widehat T^{1j}$s that are square $D \times D$ matrices.
Similarly, duplicating entries in the vectors $u$ and $v$ if necessary, we get $D$-dimensional vectors $\widehat u$ and $\widehat v$, respectively.
Now we have
\begin{equation}
\widehat u^\top\,\widehat T^{1j}=\widehat v^\top,
\qquad
\widehat A\,\widehat T^{1j}\widehat B=A\widetilde T^{1j}B.
\end{equation}

Finally, there exists $X\in GL(D)$ such that $\widehat v^\top X=\widehat u^\top$. We define
\begin{equation}
T^{\ast 1j}=\widehat T^{1j}X,\qquad
B^\ast=X^{-1}\widehat B.
\end{equation}
These result in,
\begin{equation}
    C_{(l,l_k)ji}=\widehat A_{(l,l_k):}T^{\ast 1j}B^\ast_{:i}
\end{equation}
with
\begin{equation}
    \begin{split}
        \sum_{l_k\in O_l}\widehat A_{(l,l_k):}&=\widehat u^\top \quad \forall \  \MM_l,\\
        \widehat u^\top T^{\ast 1j}&=\widehat u^\top \quad \forall \  j.
    \end{split}
\end{equation}
The latter COPE tensor decomposition is therefore a valid preGPT model for the $\pp\TT^{1}\MM$ scenario.
\end{proof}

A GPT model of an operational theory is a preGPT that is quotiented under the operational equivalence relations for preparations, transformations and outcomes~\cite{shahandeh2025unifiedlinearalgebraicframework}.
That is, a GPT does not allow for two indistinguishable procedures to have distinct representations, i.e.,
\begin{equation}\label{eq:quot_def}
\begin{split}
    & \pp_i\simeq\pp_j \Leftrightarrow s^{i} = s^{j},\\
    & \e_i \simeq \e_j \Leftrightarrow e^i = e^j,\\
    &\TT^{1}_i\simeq \TT^{1}_j \Leftrightarrow T^{1i}=T^{1j}.
\end{split}
\end{equation}
For a preGPT model given by the COPE factorizations described above, the following criteria must be met for each type of procedure to be quotiented. 

\begin{lemma} \label{lemma:quotiented1}
    Consider a preGPT model of the operational theory. Compose the transformations with the preparations (measurements) to obtain a PM scenario with a new set of preparations $\{\pp'_i\}$ (measurements $\{\MM'_i\}$), which corresponds to a new set of states $\{s'^i\}$ (effects $\{e'^i\}$). The original effects $\{e^i\}$ (states $\{s^i\}$) are quotiented if and only if the set $\{s'^i\}$ ($\{e'^i\}$) separates them.
\end{lemma}
\begin{proof}
    Assume for contradiction that the effects are quotiented, and yet there exist two distinct effects $e^k$ and $e^l$ which are not separated by elements of $\{s'^i\}$. That is, using Eq.~\eqref{eq:GPT_pob_rule},
    \begin{equation}
        \langle e^k, s'^i\rangle=\langle e^l, s'^i\rangle \; \forall \  i.
        \label{eq:sep_new_states}
    \end{equation}
    By definition, the set $\{s'^i\}$ is equivalent to the set of elements $T^{1m}s^n$ for all choices of $m$ and $n$. Therefore, Eq.~\eqref{eq:sep_new_states} implies that the effects $e^k$ and $e^l$ represent operationally equivalent measurement outcomes $\e_k$ and $\e_l$ as in Eq.~\eqref{eq:preparationmeasurementeqv}. This is a contradiction, since operationally equivalent outcomes should have the same representation in a model with quotiented effects.

    A similar argument can be applied to the states.
\end{proof}
\begin{lemma}\label{lemma:quotiented2}
    In a preGPT of the operational theory, the representation of the transformations in the model is quotiented if and only if the states and effects together separate the transformation matrices.
\end{lemma}
\begin{proof}
    Similarly to the proof of Lemma~\ref{lemma:quotiented1}, assume for contradiction that the transformations are quotiented in the model, and yet there exist two different transformations $T^{1p}$ and $T^{1q}$ which are not separated by the states and effects. Then, using Eq.~\eqref{eq:GPT_pob_rule},
    \begin{equation}
        \langle e^i, T^{1p}[s^j]\rangle=\langle e^i, T^{1q}[s^j]\rangle \; \forall \  i,j.
    \end{equation}
    It follows from Eq.~\eqref{eq:transformationeqv} that $T^{1p}$ and $T^{1q}$ represent operationally equivalent transformation procedures. This is a contradiction, since operationally equivalent transformation procedures should have the same representation in a model with quotiented transformations.
\end{proof}
\noindent Lemmas~\ref{lemma:quotiented1} and~\ref{lemma:quotiented2} amount to the following statement: the operational equivalence of two procedures is oblivious to whether all other procedures occur in multiple stages, or in a single stage. These lemmas yield a necessary and sufficient condition for the preGPT to be fully quotiented, i.e., to form a GPT, as follows. 
We first note that a decomposition of a COPE 3-tensor $C$ via Eq.~\eqref{eq:C_dec_PTM} gives rise to the following factorizations of each unfolding of $C$:
    \begin{enumerate}
    \item $C_{[\e]} = A B'$, where each possible composition of a preparation and transformation as a column of $C_{[\e]}$ is represented by a column of $B'$, i.e. $B'_{:(ij)}=T^{1j}B_{:i}$.
    \item $C_{[\pp]} = A'B$ where each possible composition of a transformation and a measurement outcome as a row of $C_{[\pp]}$ is represented by a row of $A'$, i.e. $A'_{(ij):}=A_{i:}T^{1j}$.
    \item $C_{[\TT^1]} = Y^1 \Theta^1$,
    where $\Theta^1=A^\top\otimes B$ is the Kronecker product of $A^\top$ and $B$, and $Y^1$ is such that $Y^1_{j:}$ is the row-flattening of $T^{1j}$.
\end{enumerate}
The GPT conditions of Eq.~\eqref{eq:quot_def} can be formulated as rank constraints on the matrix factors above, as follows.
\begin{theorem} \label{thm:GPT model}
Suppose a COPE 3-tensor is modelled by a preGPT specified by a decomposition as in Eq.~\eqref{eq:C_dec_PTM}. Let $r_\e=\rank(C_{[\e]})$, $r_\pp=\rank(C_{[\pp]})$, and $r_{\TT^1}=\rank(C_{[\TT^1]})$. Define  $Y^1$ such that the $j$th row of $Y^1$ is the row-flattening of $T^{1j}$. Then the preGPT is a GPT if and only if the following conditions are satisfied.
\begin{align}
\rank(A)&=r_\e \notag , \\
    \rank(B)&=r_\pp,\label{eq:proceduresquotiented} \\
    \rank(Y^1)&=r_{\TT^1}. \notag
\end{align}
\end{theorem}
\noindent  \begin{proof} Recalling the definition of the matrix $B'$, whose columns consist of representations of all transformed preparations, i.e. $B'_{:(ij)}=T^{1j}B_{:i}$, we note that $\rank(A)>r_\e$ implies that there are effects $e^k$ and $e^l$ in $\text{row}(A)$ which are not separated by the columns of $B'$. These columns are exactly the states $s^{i\prime}$ referred to in Lemma~\ref{lemma:quotiented1}. Therefore $\rank(A)>r_\e$ certifies that the effects are not quotiented via Lemma~\ref{lemma:quotiented1}; the same argument applies to the states, and to the transformations instead via Lemma~\ref{lemma:quotiented2}. Therefore the preGPT is a GPT only if all the conditions in Eq.~\eqref{eq:proceduresquotiented} are satisfied. For the if direction, suppose that $\rank(A)=r_\e$. Then the number of measurement outcomes with independent statistics $r_\e$ is equal to the number of independent effects i.e. rows of $A$. This implies that all outcomes are uniquely represented by effects. An analogous argument applies to states and transformations.
\end{proof}
The values $r_\e,r_{\TT^1}$ and $r_\pp$ can be interpreted as the capacity of each of the sets $\EE,\PP$ and $\T$ to distinguish pairs of elements composed from the other two sets. For instance, $r_\e$ is the number of preparation-transformation pairs whose statistics over the set of effects are linearly independent. The rank equalities in Eq.~\eqref{eq:proceduresquotiented} ensure that the GPT representation of the elements in any set uniquely separates any pair of elements drawn from the other two sets; equivalently, the GPT does not specify any more elements than the statistical structure demands.

Following the Tucker decomposition introduced in Eq.~\eqref{eq:tucker}, the Tucker rank of a $3$-tensor $C$ is the triple $(r_1,r_2,r_3)$, where $r_i$ ($i=1,2,3$) denotes the rank of the mode-$i$ unfolding of $C$.
A tensor $C$ can always be decomposed in Tucker form such that $(d_1,d_2,d_3)=(r_1,r_2,r_3)$, which is known as its \textit{minimal} Tucker decomposition.
It follows that the minimal Tucker decomposition of a COPE 3-tensor has the Tucker rank $(r_1,r_2,r_3)=(r_\e,r_{\TT^1},r_\pp)$.
It is then immediate that the preGPT constructed from the minimal Tucker decomposition via Lemma~\ref{lemma:preGPT_Tucker} is the minimal GPT of size $\max\{r_\e,r_\pp\}$. In fact, the construction is simplified, as in the minimal decomposition the left-kernel of the matrix $B'$ in Eq.~\eqref{eq:B'} would be empty, guaranteeing a unique unit effect.
Conversely, a minimal GPT provides a minimal Tucker decomposition of the COPE tensor after removing all-zero rows and columns from $A$, $T^1$, and $B$, reducing their sizes as necessary, followed by a further factorization of the $T^{1j}$ matrices according to Eq.~\eqref{eq:preGPTCOPEfactor}. 

Finally, we note that the most well-known algorithm for minimal Tucker decomposition is higher-order singular-value-decomposition~\cite{HOSVD1,HOSVD2}, which can easily be modified to form a GPT with transformations preserving the unit effect, following the prescription in Lemma~\ref{lemma:preGPT_Tucker}.

\begin{exbox}
\begin{example}
\textbf{2D toy theory: GPT}          
\label{ex:rebit_GPT}    
A possible GPT model for $C^{2D}$ is given by,
\begin{equation}\label{eq:A}
A =
\frac{1}{2}\begin{pmatrix}
1 & 0 & 1\\
1 & -1 & 0\\
1 & 0 & -1\\
1 & 1 & 0
\end{pmatrix},
\end{equation}
\begin{equation}\label{eq:T11}
T^{11} =
\begin{pmatrix}
1 & 0 & 0\\
0 & 0 & 1\\
0 & -1 & 0
\end{pmatrix},
\end{equation}
\begin{equation}\label{eq:T12}
T^{12} =
\begin{pmatrix}
1 & 0 & 0\\
0 & -1 & 0\\
0 & 0 & -1
\end{pmatrix},
\end{equation}
\begin{equation}\label{eq:T13}
T^{13} =
\begin{pmatrix}
1 & 0 & 0\\
0 & 0 & -1\\
0 & 1 & 0
\end{pmatrix},
\end{equation}
\begin{equation}\label{eq:T14}
T^{14} =
\begin{pmatrix}
1 & 0 & 0\\
0 & 1 & 0\\
0 & 0 & 1
\end{pmatrix},
\end{equation}
\begin{equation}\label{eq:B}
B =
\begin{pmatrix}
\frac{1}{2} & \frac{1}{2} & \frac{1}{2} & \frac{1}{2}\\[4pt]
-\frac{1}{4} & \;\;\frac{1}{4} & \;\;\frac{1}{4} & -\frac{1}{4}\\[4pt]
-\frac{1}{4} & -\frac{1}{4} & \;\;\frac{1}{4} & \;\;\frac{1}{4}
\end{pmatrix}.
\end{equation}
The unit effect of this GPT is given by the sum of the rows of $A$, which is $u=\begin{pmatrix}
    2 & 0 & 0
\end{pmatrix}$.
The Tucker rank of $C^{2D}$ is $(r_\e,r_{\TT^1},r_\pp)=(3,3,3)$. Clearly, $\rank(A)=\rank(B)=3$. To confirm that the matrix $Y^1$, whose rows are row-flattenings of the $T^{1j}$ matrices, satisfies $\rank(Y^1)=3$, note that \begin{equation}
T^{11}-T^{12}+T^{13}-T^{14}=0.
\end{equation}
It follows that each condition in  Eq.~\eqref{eq:proceduresquotiented} is satisfied, confirming that all representations of procedures are quotiented.
\end{example}
\end{exbox}

\subsection{Ontological Models of $\pp\TT^1\MM$ Scenario}

\subsubsection{General Ontological Models}
An ontological model provides a realistic interpretation of the operational theory through its probabilistic structure. It is described by an underlying ontic variable space $\Lambda$ which we assume to be finite-dimensional, and spanned by a finite set of discrete ontic points $\{\lambda_i\}$. Here, each $\lambda_i$ represents a definite state of the system, in which all of its properties are predetermined.

In this model, a preparation $\PP_i$ prepares the system in a random ontic point according to some probability distribution $\mu^i(\lambda_j)$. As the function $\mu^i$ represents an agent's limited knowledge of the true state of the system, it is referred to as an \textit{epistemic state}.
Being a probability distribution, $\mu^i(\lambda_j)$ satisfies $\sum_j \mu^i(\lambda_j)=1$.
Each transformation $\TT^1_i$ in the ontological model is represented by a stochastic map $\Gamma^{1i}:\Lambda\rightarrow \Lambda$.
In the linear-algebraic formalism, these are represented by column-stochastic matrices encoding transition probabilities between ontic points, $\Gamma^{1i}(\lambda_k|\lambda_j)$, with $\sum_k \Gamma^{1i}(\lambda_k|\lambda_j)=1$.
Finally, if the system is in the state $\lambda_i$ just before the measurement $\MM_j$ with possible outcomes $\{\e_k^j\}$ is performed, the outcome $k$ occurs according to a probability distribution over outcomes $\xi^{jk} (\lambda_i)$.
The function $\xi^{jk}$ is known as a \textit{response function} and satisfies $\sum_k \xi^{jk}(\lambda_i)=1$ for all $\lambda_i$. 

It follows that an ontological model of the COPE 3-tensor $C$ specifies a nonnegative factorization
\begin{equation}
\label{eq:PTM_ont_decomp}
    C_{ijl}=R_{i:}\Gamma^{1j}E_{:l}.
\end{equation}
We now show that every such model amounts to a nonnegative preGPT.

\begin{lemma} \label{lemma:ont_nonnegpreGPT}
    An ontological model and a nonnegative preGPT model  are equivalent up to a nonnegative diagonal rescaling.
\end{lemma}
\begin{proof}
    We can represent each ontic point $\lambda_i$ as the $ith$ cartesian unit vector of a $|\Lambda|$-dimensional vector space. Therefore, we can represent the probability distribution $\mu^i$ as a nonnegative stochastic vector.
    Each $\Gamma^{1i}$ is a stochastic matrix, such that $\Gamma^{1i}_{lm}$ is the transition probability from $\lambda_m$ to $\lambda_l$. Finally, we can represent $\xi^{jk}$ as a nonnegative vector. 
    This has the same structure as a preGPT wherein the unit effect is the vector of all ones $\boldsymbol{1}$. The latter follows from the condition $\sum_{k}\xi^{jk}(\lambda_i)=1$ for each measurement outcome $j$ and all $\lambda_i$ in the ontological model. 

    For the reverse direction, given that the nonnegative preGPT has a nonnegative unit effect $u$ which is preserved by all transformations, we have that
\begin{equation}
\begin{split}
    \sum_{l_k\in O_l} A_{(l,l_k):} &= u^\top,\\
    u^\top B &= \boldsymbol{1}^\top,\\
    u^\top T^{1j} &= u^\top,
\end{split}
\end{equation}
for every measurement $\MM_l$, in which we denote its set of outcomes by $O_l$, and every transformation $T^{1j}$. 

Defining the matrix $D=\operatorname{diag}(u)$, we rescale the matrices $A,T^{1j}$ and $B$ to form
\begin{equation}
\begin{split}
    R &:= AD^{\dagger},\\
    \Gamma^{1j} &:= DT^{1j}D^{\dagger},\\
    E &:= DB,
\end{split}
\end{equation}
where $\dagger$ denotes the pseudoinverse.
Then
\begin{equation}
    R\Gamma^{1j}E=AT^{1j}B=C_{:j:}.
\end{equation}
Now, $R$ and $\Gamma^{1j}$ may have some zero columns. To obtain an ontological model, for any $u_i=0$ and $u_{i'}>0$, we overwrite the zero column $R_{:i}$ with $R_{:i'}$, and overwrite the zero column $\Gamma^{1j}_{:i}$ with $\Gamma^{1j}_{:i'}$ for all $j$. We then repeat this process for all $u_i=0$.
Thus $R$ is a normalized response function matrix, $E$ is a matrix of epistemic states, and each $\Gamma^{1j}$ is column-stochastic. Hence the nonnegative preGPT is equivalent to an ontological model up to the nonnegative diagonal rescaling above.
\end{proof}

\begin{exbox}
\begin{example}
\textbf{2D toy theory: Ontological Model}          
\label{ex:rebit_ont}    
A possible ontological model for $C^{2D}$ is
\begin{equation}
\begin{aligned}
R &=\mathds{1}_{4\times4},\\
\Gamma^{1j}_{ik} &= \delta^{[4]}_{\,i-k,\;j-1}\;,\\
E &=C^{2D}_{:1:},
\end{aligned}
\end{equation}
where $\delta^{[4]}_{a,b}=1$ if $a\equiv b \pmod 4$, and 0 otherwise. 
The unit response function is the sum of rows of $R$, which is the vector of all ones.
\end{example}
\end{exbox}

\subsubsection{NCOMs of $\pp\TT^1\MM$ Scenario} \label{sec:NCOM}

We now demonstrate how our formalism captures the notion of generalized noncontextuality~\cite{Spekkens_2005} in $\pp\TT^1\MM$ scenarios. For clarity, as in the previous sections, we label outcome events $\e_j$ and the corresponding response functions $\xi^j$ with a single joint index $j$, where $j$ runs over all measurements and their outcomes. 
By definition, an ontological model is noncontextual precisely when it assigns identical representations to all operationally equivalent procedures, i.e.,
\begin{equation} \label{eq:PTMNCOM}
\begin{split}
    & \pp_i\simeq\pp_j \Leftrightarrow \mu^i = \mu^j,\\
    & \e_i \simeq \e_j \Leftrightarrow \xi^i = \xi^j,\\
    &\TT^{1}_i\simeq \TT^{1}_j \Leftrightarrow \Gamma^{1i}=\Gamma^{1j}.
\end{split}
\end{equation}
It may be useful to note that, in the above definition of ontological noncontextuality, the sufficiency direction is an assumption motivated by Leibniz's principle of ontological identity of empirical indiscernibles~\cite{spekkens2019ontologicalidentityempiricalindiscernibles} as presented in Spekkens' seminal work~\cite{Spekkens_2005}.
The converse direction, however, is a mathematical necessity baked into the very framework of ontological models\footnote{To violate the necessity direction, the measurement response functions would need to depend directly on the preparation recipes, thereby violating the $\lambda$-mediation assumption. However, one can trivially restore $\lambda$-mediation by enlarging the ontic space to include the space of recipes, in which case the epistemic states are no longer identical.}. Although including the necessity direction is a tautology, it allows us to draw a parallel between noncontextual ontological models and GPTs in Eq.~\eqref{eq:quot_def} as follows.
\begin{lemma}
    \label{lemma:NCOM_nonnegGPT}
    An NCOM and a nonnegative GPT are equivalent up to a nonnegative rescaling.
\end{lemma}
\begin{proof}
    The equivalence relation $\simeq$ certifies that an ontological model is an NCOM in Eq.~\eqref{eq:PTMNCOM}, and that a preGPT is a GPT in Eq.~\eqref{eq:quot_def}. Combining this with Lemma~\ref{lemma:ont_nonnegpreGPT} produces the result.
\end{proof}
We formulate the noncontextuality conditions of Eq.~\eqref{eq:PTMNCOM} as rank constraints on the COPE tensor as follows.
\begin{theorem}
\label{thm:PTM-NC-rank-cnstrnts}
Suppose a COPE 3-tensor $C$ is described by an ontological model as in Eq.~\eqref{eq:PTM_ont_decomp}. Let $r_\e=\rank(C_{[\e]})$, $r_\pp=\rank(C_{[\pp]})$, and $r_{\TT^1}=\rank(C_{[\TT^1]})$. We denote by $G^1$ the matrix such that its $j$th row $G^1_{j:}$ is the row-flattening of $\Gamma^{1j}$.
Then an NCOM of a $\pp\TT^1\MM$ scenario defined by the matrices of epistemic states $E$, response functions $R$, and row-flattened transformations $G^1$, must satisfy
\begin{equation}\label{eq:NCOM_cond}
\begin{split}
    \rank(R)&=r_\e,\\
    \rank(E)&=r_\pp,\\
    \rank(G^1)&=r_{\TT^1}.
\end{split}
\end{equation}
\end{theorem}
\begin{proof}
    Lemmas~\ref{lemma:ont_nonnegpreGPT} and~\ref{lemma:NCOM_nonnegGPT} established an equivalence between ontological models and nonnegative preGPTs, and between NCOMs and nonnegative GPTs respectively. The result follows from the application of Theorem~\ref{thm:GPT model} to nonnegative preGPTs. 
\end{proof}
The conditions in Eq.~\eqref{eq:NCOM_cond} are the generalization of the equirank conditions introduced in Ref.~\cite{shahandeh2025unifiedlinearalgebraicframework} for PM scenarios. The violation of the first, second and third line of Eq.~\eqref{eq:NCOM_cond} imply measurement, preparation and transformation contextuality, respectively. 

We note that these conditions certify that operational equivalences 
at each stage are respected at the ontological level; equivalences between composite procedures may not be respected. For example, consider two sentences $\Pi=\pp_i\TT^1_j$ and $\Pi'=\pp_k\TT^1_l$ in the operational theory. Each corresponds to a particular preparation followed by a particular transformation. If every measurement following $\Pi$ or $\Pi'$ yields the same statistics in either case, then these composite procedures are experimentally indistinguishable, $\Pi\simeq \Pi'$. However, Eq.~\eqref{eq:NCOM_cond} does not entail that their representations, $\Gamma^{1j}\mu^i$ and $\Gamma^{1l}\mu^k$, are identical. We argue that this is natural, given that composite equivalences of this kind fall outside the purview of the noncontextuality conditions in Eq.~\eqref{eq:PTMNCOM}. 

For completeness, we nevertheless note that composite equivalences can be treated with rank constraints in a similar fashion. For instance, suppose we require all equivalences between transformed preparations, such as $\pp_i\TT^1_j\simeq \pp_k\TT^1_l$, to be respected. We note that $r_\e=\rank(C_{[\e]})$ is the rank of the mode-$\e$ unfolding of $C$, whose columns are labelled by all possible pairs of transformations and preparations. The matrix factor associated with transformed preparations is $E'$, defined as $E'_{:(ij)}=\Gamma^{1j}\mu^i$. Therefore, $\rank(E')=r_{\e}$ amounts to the requirement that all composite equivalences $\pp_i\TT^1_j\simeq \pp_k\TT^1_l$ are respected, i.e. transformed preparations are represented noncontextually. We do not apply these extra constraints to our factorizations, but an example showcasing single-stage vs composite equivalences is presented in Sec.~\ref{sec:temporal}.

Finally, in Appendix~\ref{sec:appA} we display factorizations of each flattening $C_{[\e]},C_{[\TT^1]}$ and $C_{[\pp]}$ of the COPE 3-tensor in terms of $R,G^1$ and $E$, and explain their interpretations. 

So far, we have characterized NCOMs as nonnegative factorizations of the COPE tensor satisfying the rank constraints in Theorem~\ref{thm:PTM-NC-rank-cnstrnts}, with $E$, every transformation matrix $\Gamma^{1j}$, and each measurement block of $R$ column-stochastic. However, when searching for an NCOM, normalization need not be imposed separately on each measurement block of $R$. Under the rank condition $\rank(R)=\rank(C_{[\e]})$, requiring each column of $R$ to sum to the number of measurements, $n_\MM$, already guarantees that each measurement block is column-stochastic. We establish this in the following lemma.

\begin{lemma}
\label{lemma:blocknormalization}
Suppose a COPE $3$-tensor for a $\pp\TT^1\MM$ scenario with $n_\MM$ measurements admits a nonnegative factorization
\begin{equation}
    C_{ijk}=R_{i:}\Gamma^{1j}E_{:k},
\end{equation}
where the epistemic state matrix $E$ and every transformation matrix $\Gamma^{1j}$ are column-stochastic. Suppose further that $R$ has $d$ columns and satisfies
\begin{equation}
    \begin{split}
        &\rank(R)=\rank(C_{[\e]}),\\
        &\sum_i R_{it}=n_\MM,
    \end{split}
\end{equation}
for each column $R_{:t}$.

Let $O_l$ denote the set of outcomes of measurement $\MM_l$. Then each measurement block of $R$ is column-stochastic:
\begin{equation}
    \sum_{l_k\in O_l}R_{(l,l_k)t}=1 \quad \forall\,l,t.
\end{equation}
Thus $R$ is a normalized response function matrix. Additionally, this holds for any number of transformation stages.
\end{lemma}

\begin{proof}
Consider the mode-$\e$ unfolding, $C_{[\e]}=RE'$, where the columns of $E'$ represent all transformed preparations:
\begin{equation}
    E'_{:(ij)}=\Gamma^{1j}E_{:i} \quad \forall\,i,j.
\end{equation}
For each measurement $\MM_l$, normalization of its outcome probabilities implies
\begin{equation}
    \sum_{l_k\in O_l}\bigl(C_{[\e]}\bigr)_{(l,l_k)t}=1 \quad \forall\,l,t.
\end{equation}
since the observed probabilities are normalized separately for each measurement. Consequently for any two measurements $\MM_l$ and $\MM_{l'}$,
\begin{equation}
\label{eq:homogeneousrelation}
    \sum_{l_k\in O_l}\bigl(C_{[\e]}\bigr)_{(l,l_k):}-\sum_{l'_k\in O_{l'}}\bigl(C_{[\e]}\bigr)_{(l',l'_k):}=0.
\end{equation}

The factorization $C_{[\e]}=RE'$ implies
\begin{equation}
    \text{cspan}(C_{[\e]})\subseteq\text{cspan}(R).
\end{equation}
By the rank condition $\rank(R)=\rank(C_{[\e]})$, these column spaces have the same dimension and hence are equal. Therefore, every linear relation among the rows of $C_{[\e]}$ also holds among the corresponding rows of $R$. Therefore, Eq.~\eqref{eq:homogeneousrelation} implies,
\begin{equation}
    \sum_{l_k\in O_l}R_{(l,l_k):}-\sum_{l'_k\in O_{l'}}R_{(l',l'_k):}=0 \quad \forall\,l,l'.
\end{equation}
Thus, for any fixed column of $R$, the sum of its entries within a measurement block is the same for every measurement. Since
\begin{equation}
    \sum_i R_{it}=n_\MM
\end{equation}
for every column $R_{:t}$, and there are $n_\MM$ measurement blocks, it follows that
\begin{equation}
    \sum_{l_k\in O_l}R_{(l,l_k)t}=1 \quad \forall\,l,t.
\end{equation}
Hence each measurement block of $R$ is column-stochastic.
\end{proof}

\begin{exbox}
\begin{example}
\textbf{2D toy theory: NCOM.}          
\label{ex:noncont}    
an NCOM for $C^{2D}$ is
\begin{equation}
R = E =
\frac{1}{2}\begin{pmatrix}
1 & 1 & 0 & 0\\
0 & 1 & 1 & 0\\
0 & 0 & 1 & 1\\
1 & 0 & 0 & 1
\end{pmatrix},
\end{equation}
\begin{equation}
\Gamma^{11} =
\frac{1}{2}\begin{pmatrix}
1 & 1 & 0 & 0\\
0 & 1 & 1 & 0\\
0 & 0 & 1 & 1\\
1 & 0 & 0 & 1
\end{pmatrix},
\end{equation}
\begin{equation}
\Gamma^{12} =
\frac{1}{2}\begin{pmatrix}
1 & 0 & 0 & 1\\
1 & 1 & 0 & 0\\
0 & 1 & 1 & 0\\
0 & 0 & 1 & 1
\end{pmatrix},
\end{equation}
\begin{equation}
\Gamma^{13} =
\frac{1}{2}\begin{pmatrix}
0 & 0 & 1 & 1\\
1 & 0 & 0 & 1\\
1 & 1 & 0 & 0\\
0 & 1 & 1 & 0
\end{pmatrix},
\end{equation}
\begin{equation}
\Gamma^{14} =
\frac{1}{2}\begin{pmatrix}
0 & 1 & 1 & 0\\
0 & 0 & 1 & 1\\
1 & 0 & 0 & 1\\
1 & 1 & 0 & 0
\end{pmatrix}.
\end{equation}
The unit response function is the vector of all ones as required.
To see that this model is noncontextual, first observe that
\begin{equation*}
    \rank(R)=\rank(E)=3.
\end{equation*}
The rows of the matrix $G^1$ are the row-flattened $\Gamma^{1i}$ matrices. Noting that 
\begin{equation*}
    \Gamma^{11}-\Gamma^{12}+\Gamma^{13}-\Gamma^{14}=0
\end{equation*}
we have 
\begin{equation*}
    \rank(G^1)=\dim \text{span}(\{\Gamma^{11},\Gamma^{12},\Gamma^{13},\Gamma^{14}\})=3.
\end{equation*}
The Tucker rank of $C^{2D}$ then satisfies
\begin{equation*}
    \begin{split}
        (r_\e,r_{\TT^1},r_\pp)&=(\rank(R),\rank(G^1),\rank(E))\\&=(3,3,3)
    \end{split}
\end{equation*}
as in Example~\ref{ex:rebit_GPT}.
\end{example}
\end{exbox}

\subsubsection{Complexity of computing NCOMs in $\pp\TT^1\MM$ Scenario} \label{sec:computingNCOMs}
In the previous section we showed that an NCOM of a PTM scenario amounts to a nonnegative factorization of the COPE 3-tensor satisfying particular rank constraints.
In this section, we tackle the problem of determining whether such a factorization exists, and computing one if it does. We also offer a complexity analysis of this problem.
To do so, 
it is useful to develop a geometric interpretation of the conditions in Eq.~\eqref{eq:NCOM_cond}. 

These conditions specify the rank of nonnegative matrices appearing in a decomposition of the COPE tensor. The study of nonnegative \textit{matrix} factorizations with rank restrictions was introduced in Ref.~\cite{GILLIS20122685} and given the name of \textit{restricted nonnegative matrix factorization} (RNMF). For any nonnegative matrix $C$, an RNMF is a nonnegative factorization $C=RE$ such that $\rank(R)=\rank(C)$.
Note that in any real factorization $C=AB$, nonnegative or otherwise, the columns of $B$ can be interpreted as a complete set of points, consisting of vertices and possibly redundant points, whose convex hull is a polytope $\mathcal{B}$. Likewise, the rows of $A$ can be interpreted as a complete set of inequalities, consisting of facet inequalities and possibly redundant inequalities, whose intersection is a polytope $\mathcal{A}$.
This leads to the following geometric interpretation of RNMF~\cite{GILLIS20122685}.
\begin{lemma} \label{lemma:nestedpolytopes}
    Let $C$ be an $m\times n$ matrix of rank $r$ and column sum $l$ with a real factorization $C=AB$ of inner dimension $r$. Define the polyhedron $\mathcal{A}=\{x\in\operatorname{rspan}(A)|Ax\geq0, \sum_{i}(Ax)_i=l\}$ and the polytope $\mathcal{B}:=\operatorname{conv}\{B_{:1},\ldots,B_{:n}\}$. 
    Then, there exists an RNMF $C=RE$ of inner dimension $k$ with $\rank(R)=r$ if and only if there exists a nested polytope $\mathcal{G}\subset\operatorname{rspan}(A)$ with at most $k$ vertices such that $\mathcal{B}\subseteq\mathcal{G}\subseteq\mathcal{A}$.
\end{lemma}
\begin{proof}
    Firstly, for a column $B_{:j}$ of $B$, $AB_{:j}=C_{:j}\geq0$ and $\sum_i(AB_{:j})_i=\sum_iC_{ij}=l$. Hence every column of $B$ lies in $\mathcal A$, and therefore $\mathcal B\subseteq\mathcal A$.
    
    Next, assume that such a $\mathcal G$ exists, and let $G$ be a matrix with $k$ columns whose convex hull is $\mathcal G$. Clearly, $\mathcal G\subseteq\mathcal A$ implies $AG\geq0$. Since $\mathcal{B}\subseteq\mathcal{G}$, each column of $B$ is in the convex hull of the columns of $G$. Therefore, there exists a nonnegative matrix $E$ such that $GE=B$. Then, defining $R:=AG$, we have $RE=AGE=AB=C$. We know $\rank(R)=r$ since $\rank(A)=r$, and $R$ has $k$ columns, therefore, $C=RE$ is an RNMF of inner dimension $k$.

    For the reverse direction, let $C=RE$ be an RNMF of inner dimension $k$, where, without loss of generality, $R$ has column-sum $l$, and $E$ is column-stochastic. Then, since $\rank(R)=\rank(A)$, there exists a matrix $G$ such that $R=AG$ and $\text{cspan}(G)=\text{rspan}(A)$. Now, we have $RE=A(GE)$. Since $\text{cspan}(GE)=\text{cspan}(G)=\text{rspan}(A)$, we can identify $GE=B$. Therefore, $B$ is in the convex hull of the columns of $G$.
\end{proof}
Note that an RNMF is always possible when $k$ equals the number of vertices of $\mathcal{A}$.
In this case, the intermediate polytope $\mathcal{G}$ is chosen to be $\mathcal{A}$ itself, which guarantees that $\mathcal{B}\subseteq\mathcal{G}$. 
We call this decomposition the
\textit{extremal} RNMF and the associated restricted extremal nonnegative matrix factor $\bar{R}$.

\begin{lemma} \label{lemma:extremalfactor}

    For an $m\times n$ nonnegative matrix $C$ of rank $r$, there exists a rank-$r$ nonnegative matrix $\bar{R}$ such that for any RNMF $C=RE$, $R$ can be written as $R=\bar{R}M$, for some nonnegative matrix $M$.
    $\bar{R}$ has $m^{O(r)}$ columns, and the complexity of computing it is at most $m^{O(r)}+\operatorname{poly}(m,n)$.
\end{lemma}
\begin{proof}
    We will give a constructive proof of the existence of such a matrix $\bar{R}$. We assume w.l.g. that $C$ has a column sum of some integer number $l$. 
    Given a rank factorization $C=AB$, define the polytopes $\mathcal{A}$ and $\mathcal{B}$ as in Lemma~\ref{lemma:nestedpolytopes}, as well as $r:=\operatorname{rank}(C)$.
    We first show that $\mathcal{A}$ is bounded. Assume for contradiction that it is unbounded: then there exists a set of points $x+ty$ in $\mathcal{A}$, for some fixed nonzero vectors $x,y$, and a scalar $t\geq0$. Since $\sum_i (A(x+ty))_i=l$ must hold for $t\ge 0$, we must have $\sum_i (Ay)_i=0$. However, $y$ is nonzero and $A$ is full rank, so $Ay\neq0$. Therefore $Ay$ contains a negative element. This means that $A(x+ty)$ contains a negative element for large enough $t$, which is a contradiction.

    Given that $\mathcal{A}$ is bounded, any nested polytope $\mathcal{G}$ is in the convex hull of the vertices of $\mathcal{A}$. Let $\bar{G}$ be the matrix whose columns are vertices of $\mathcal{A}$. By Lemma~\ref{lemma:nestedpolytopes} we have that if $C=RE$ is an RNMF, $R=AG$ where the columns of $G$ are vertices of the nested polytope $\mathcal{G}\subseteq \mathcal{A}$. The matrix $G$ can therefore be written as $G=\bar{G}M$, for some nonnegative matrix $M$. Therefore, $R=AG=A\bar{G}M$ for some $M$. The desired matrix is given by $\bar{R}=A\bar{G}$.

    It is known that vertex enumeration of a polytope with $m$ facets in $r$ dimensions has complexity $m^{O(r)}$ and the number of vertices is $m^{O(r)}$~\cite{Khachiyan2008}. The complexity of computing the rank factorization $C=AB$ is $\operatorname{poly}(m,n)$, hence the complexity result.
\end{proof}
\noindent A similar argument can be applied to $C^\top$ to construct an extremal matrix, $\bar{E}$.

\begin{lemma}
\label{lemma:extremalfactor2}
    For an $m\times n$ nonnegative matrix $C$ of rank $r$, there exists a rank-$r$ nonnegative matrix $\bar{E}$ such that for any RNMF of $C^\top$, $C^\top=E^\top R^\top$, $E$ can be written as $E=M'\bar{E}$, for some nonnegative $M'$.
    This $\bar{E}$ has $n^{O(r)}$ rows, and the complexity of computing it is $n^{O(r)}+\operatorname{poly}(m,n)$.
\end{lemma}

Now, we can use Lemmas~\ref{lemma:extremalfactor} and~\ref{lemma:extremalfactor2} to decide if an NCOM of an operational theory exists, and compute one if it does. For COPE matrices the rank $r$ and the minimum GPT dimension coincide; for COPE tensors we will see that the latter quantity continues to play a key role in the complexity of computing an NCOM. Recall that the minimum GPT dimension $r$ for a COPE 3-tensor with Tucker rank $(r_\e, r_{\TT^1},r_\pp)$ is given by $r=\max\{r_\e,r_\pp\}$.

\begin{lemma}
\label{lemma:extremalnormalform}
Suppose a $\pp\TT^1\MM$ scenario with COPE tensor $C$ admits an NCOM
\begin{equation}
C_{:j:}=R\Gamma^{1j}E .
\end{equation}
Let $\bar R$ be the extremal factor obtained from  $C_{[\e]}$ according to Lemma~\ref{lemma:extremalfactor}. By Lemma~\ref{lemma:blocknormalization}, each represented measurement block in $\bar R$ is column-stochastic.
Moreover, let $\bar E$ be the extremal state factor obtained from  $C_{[\pp]}$ according to Lemma~\ref{lemma:extremalfactor2}.
Then, $C$ admits another NCOM
\begin{equation}
C_{:j:}=\bar R\tilde\Gamma^{1j}\tilde E ,
\end{equation}
where $\tilde E$ is obtained from $\bar E$ by applying a nonnegative diagonal rescaling. In particular, $\bar R$ and $\tilde E$ are normalized, and every $\tilde\Gamma^{1j}$ is column-stochastic, after duplicating columns and appending rows of zeros if required.
\end{lemma}

\begin{proof}
By Lemmas \ref{lemma:extremalfactor} and \ref{lemma:extremalfactor2}, there are nonnegative matrices $M$ and $M'$ such that $R=\bar R M$ and $E=M'\bar E$.
Using a positive diagonal rescaling, the factor $\bar R$ can be assumed without loss of generality to be composed of column-stochastic measurement blocks. Hence both $R$ and $\bar R$ are composed of disjoint column-stochastic sets of rows. Therefore, $M$ is column-stochastic.
Defining $\widehat{\Gamma}^{1j}:=M\Gamma^{1j}M'$ and $v^{\top}:=\mathbf{1}^{\top}M'$, we obtain a new factorization $C_{:j:} = \bar R\widehat{\Gamma}^{1j}\bar E$. Using the normalization of the original NCOM factors,
\begin{equation}
\begin{split}
&\mathbf{1}^{\top} E=(\mathbf{1}^{\top} M') \bar E \implies v^{\top}\bar E = \mathbf{1}^{\top},\\
&\mathbf{1}^{\top}\widehat{\Gamma}^{1j} =(\mathbf{1}^{\top}M\Gamma^{1j})M'=\mathbf{1}^{\top}M'=v^{\top}.
\end{split}
\end{equation}
Now we consider deleting unused rows and columns of the factors. For any $v_a=0$, the $a$th column of every $\widehat{\Gamma}^{1j}$ is identically zero. Therefore, we may delete the $a$th column of $\widehat{\Gamma}^{1j}$ and the $a$th row of $\bar E$ without changing the COPE tensor. After repeating for all $v_a=0$, assume $v>0$.

In order to normalize the remaining unnormalized factors, we introduce the diagonal matrix $D^v:=\operatorname{diag}(v)$. Using this matrix, we rescale according to
\begin{equation}
\begin{split}
\tilde E &:= D^v\bar E,\\
\tilde\Gamma^{1j} &:= \widehat{\Gamma}^{1j}(D^v)^{-1}.
\end{split}
\end{equation}
Finally, to arrive at a consistent ontic dimension, rows of zeros can be appended to the $\tilde\Gamma^{1j}$ matrices and $\tilde E$, and columns of $\bar R$ and the $\tilde\Gamma^{1j}$ matrices can be duplicated.

This leads to the factorization $C_{:j:}=\bar R\tilde\Gamma^{1j}\tilde E$, where all factors satisfy the appropriate normalization conditions, and nonnegativity is preserved throughout.

It remains to note that none of the above operations increase the ranks relevant to the noncontextuality conditions in Eq.~\eqref{eq:NCOM_cond}. Deleting rows and columns, and applying a positive diagonal rescaling are linear operations. Likewise, the map $\Gamma^{1j}\mapsto M\Gamma^{1j}M'$ is linear. Indeed, any equality between convex mixtures of the transformation matrices can be written as
    $\sum_i \alpha_i\Gamma^{1i}=0$, where the coefficients $\alpha_i$
    are the difference of two probability distributions. By linearity,
    $\sum_i \alpha_i M\Gamma^{1i}M'=M\left(\sum_i\alpha_i\Gamma^{1i}\right)M'=0$.
In short, the transformation $\Gamma^{1i}\rightarrow M\Gamma^{1i}M'$ does not destroy any equivalences.

Therefore the new normalized ontological model is an NCOM.
\end{proof}

This above normal form allows us to set up the following linear program to decide contextuality in the $\pp\TT^1\MM$ scenario.

\begin{theorem}\label{thm:PTMLinprog}
    It can be decided whether a $\pp\TT^1\MM$ scenario in an operational theory with $n_\pp$ preparations, $n_\e$ measurement events, $n_{\TT^1}$ transformations, and minimum GPT dimension $r$ of its COPE tensor admits an NCOM in time $\poly((n_\e n_\pp)^{O(r)},n_{\TT^1})$.
\end{theorem}
    \begin{proof}
    Suppose an ontological model of the scenario is given by response function matrix $R$, epistemic state matrix $E$, and stochastic transformation matrices $\{\Gamma^{1i}\}_{i=1}^{n_{\TT^1}}$. Then, by Lemma~\ref{lemma:extremalnormalform}, it can be assumed without loss of generality that $R=\bar{R}$, and that $E=\bar{E}$ up to a positive diagonal rescaling.

    Now, to decide if the COPE admits an NCOM, we first find the extremal factors $\bar{R}$ and $\bar{E}$, which have at most $n_\e^{O(r)}$ columns and at most $n_\pp^{O(r)}$ rows, respectively. The complexity of this step is $n_\e^{O(r)}+n_\pp^{O(r)}$. Note that here, $\bar{E}$ is subject to a positive diagonal rescaling. From these factors, we define the matrix of epistemic state-response function transitions 
\begin{equation}
\label{eq:J1}
    J^1=\bar R^\top\otimes \bar E
\end{equation}
which is of size $(n_\e n_\pp)^{O(r)}\times n_\e n_\pp$.
Now, the problem of finding an NCOM can be reduced to finding a matrix $G^1$, which has the row-flattened transformation matrices as its rows, together with a vector $v$ encoding the diagonal rescaling of $\bar E$. In particular, $G^1$ must satisfy $G^1J^1=C_{[\TT^1]}$ and $\rank(G^1)=\rank(C_{[\TT^1]})$.

To encode the remaining normalization constraints, We choose the size of the ontic space $\Lambda$ to be the maximum of the number of columns of $\bar R$ and the number of rows of $\bar E$. Then, denote by $\Sigma$ the fixed matrix which maps the row-flattening of a transformation matrix to its vector of column sums. Explicitly, given that the row flattening of $\Gamma^{1i}$ is $G^1_{i:}$, we have that $(G^1_{i:} \Sigma)_j$ is the sum of entries in the column $\Gamma^{1i}_{:j}$.
Therefore, the remaining normalization constraints are given by,
\begin{equation}
\begin{split}
    G^1\Sigma &= \boldsymbol{1}^{n_{\TT^1}}v^\top,\\
    v^\top\bar E &= \boldsymbol{1}^{n_\pp\top} ,
\end{split}
\end{equation}
for some nonnegative vector $v$, where $\boldsymbol{1}^{n_{\TT^1}}$ is the vector of all ones of length $n_{\TT^1}$, and $\boldsymbol{1}^{n_\pp}$ is the vector of all ones of length $n_\pp$.

To search for $G^1$, we start with the ansatz $G^{1\prime}=C_{[\TT^1]}J^{1\dagger}$, where $J^{1\dagger}$ is the pseudoinverse of $J^1$. Obviously, $\rank(G^{1'})=\rank(C_{[\TT^1]})$. It follows that finding the matrix $G^1$ from this ansatz is equivalent to finding a shear matrix $Q$ of the form~\cite{yianni2025complexitycontextuality},
\begin{equation}
\label{eq:shear}
    Q=\mathds{1}_{k\times k}+\sum_{i=1}^\rho\sum_{j=1}^{k-\rho}D_{ij} \bm{a}^{i}\bm{b}^{j\top},
\end{equation}
such that, for some vector $v$,
\begin{equation}
\begin{split}
    G^{1\prime}Q &\geq 0,\\
    (G^{1\prime}Q)\Sigma &= \boldsymbol{1}^{n_{\TT^1}}v^\top,\\
    v^\top\bar E &= \boldsymbol{1}^{n_\pp\top},\\
    v &\geq 0.
\end{split}
\end{equation}
Here, $\{\bm{a}^{1},...,\bm{a}^{\rho}\}$ is an orthonormal basis for the columns of $J^1$, $\{\bm{b}^{1\top},...,\bm{b}^{(k-\rho)\top}\}$ is an orthonormal basis for the left kernel of $J^1$, $\rho$ is the rank of $J^1$, $k$ is the number of rows of $J^1$, and
$D$ is a real matrix to be found via optimization.

It is immediate given the forms of $Q$ and $G^{1\prime}$ that
\begin{equation}
\label{eq:T1flattening}
\begin{split}
    G^{1\prime}QJ^1&=G^{1\prime}J^1,\\
    \rank(G^{1\prime}Q)&=\rank(C_{[\TT^1]}),
\end{split}
\end{equation}
therefore, 
\begin{equation}
    G^1=G^{1\prime}Q
\end{equation}
is a valid solution for $G^1$ whenever the above linear constraints are feasible.

With this choice of $G^1$, let $\Gamma^{1i}$ be the matrix obtained by reshaping the $i$th row of $G^1$. For any $v_a=0$, the $a$th column of every $\Gamma^{1i}$ is zero by nonnegativity. Therefore, the $a$th row of $E$ can be deleted without changing the COPE tensor. After repeating for all $v_a=0$, define
\begin{equation}
    D^v=\operatorname{diag}(v).
\end{equation}
Finally, we rescale the factors $\bar E$ and $\{\Gamma^{1i}\}$ as,
\begin{equation}
\begin{split}
    \bar E&\mapsto D^v\bar E,\\
    \Gamma^{1i}&\mapsto \Gamma^{1i}(D^v)^{-1},
\end{split}
\end{equation}
ensuring they are now column-stochastic without changing the COPE tensor.

Thus a feasible solution of the linear program defines an NCOM.

With this choice of $G^1$, the conditions for matrix $D$ and vector $v$ can be written as linear inequalities and equalities, so that they can be found via a linear program. In particular, each entry of $G^1$ corresponds to a linear inequality and is a linear function of $G^{1\prime}$, $D$ and the $\boldsymbol a^i$ and $\boldsymbol b^j$ vectors. Likewise, the constraints $G^1\Sigma=\boldsymbol{1}^{n_{\TT^1}}v^\top$ and $v^\top\bar E=\boldsymbol{1}^\top$ are linear equalities. There are at most $n_{\TT^1}(n_\e n_\pp)^{O(r)}$ nonnegativity inequalities, at most $n_{\TT^1}n_\pp^{O(r)}+n_\pp$ normalization equalities, and $\poly((n_\e n_\pp)^{O(r)},n_{\TT^1})$ many variables. Therefore, the total complexity of solving this linear program is $\poly((n_\e n_\pp)^{O(r)},n_{\TT^1})$.
\end{proof}

In fact, we can argue for a tighter complexity bound in the above proof. For procedure $\mathsf{O}$, let $r_\mathsf{O}$ be the rank of the mode-$\mathsf{O}$ matricization of $C$. Using Lemma~\ref{lemma:extremalfactor}, $\bar{E}$ has $n_\pp^{O(r_\pp)}$ rows, and $\bar{R}$ has $n_\e^{O(r_\e)}$ columns.
Constructing $J^{1}$ from these extremal factors yields a matrix with $\operatorname{poly}(n_\pp,n_\e)^{r_\e+r_\pp}$ rows. Then, the matrix $G^1$ would have dimensions $n_{\TT^1}\times \operatorname{poly}(n_\pp,n_\e)^{r_\e+r_\pp}$.  Therefore, finding a nonnegative solution for $G^1$ would have complexity $\operatorname{poly}(n_{\TT^1})\operatorname{poly}(n_\pp,n_\e)^{r_\e+r_\pp}$.

In the following sections, we will discuss complexity simply in terms of the minimal GPT dimension $r$, rather than the various ranks of matricizations of the COPE tensor. Nevertheless, it is worth noting that these ranks can be used to tighten the complexity bounds via similar arguments as above. This is useful in cases where many of these ranks are small.

\section{Models of COPE 4-tensors}\label{sec:sequential}
The natural extension of the $\pp\TT^1\MM$ scenario is the $\pp\TT^1\TT^2\MM$ in which there are two sequential stages of transformations.
In this case, the probabilistic structure of the operational theory is given by a COPE $4$-tensor, where one of the dimensions encodes the recipes for the second transformation stage. 
As in the $\pp\TT^1\MM$ scenario, the procedures in each stage appearing in the COPE tensor are atomic.

The operational equivalence between two transformations in the first stage is denoted by
\begin{equation} \label{eq:transformation2layereqv}
\begin{split}
\TT^1_i \simeq \TT^1_j \; \text{iff} \; p(\e| \pp, \TT^1_i,\TT^2_k, \MM) = &p(\e|\pp, \TT^1_j,\TT^2_k, \MM) \\ &\forall \ \pp, \TT^2_k, \e, \MM, \\
\end{split}
\end{equation}
and for two transformations in the second stage by
\begin{equation}
    \begin{split}
    \label{eq:T2_opequiv}
        \TT^2_i \simeq \TT^2_j \; \text{iff} \; p(\e| \pp, \TT^1_k,\TT^2_i, \MM) = &p(\e|\pp, \TT^1_k,\TT^2_j, \MM) \\ &\forall \ \pp, \TT^1_k, \e, \MM.
    \end{split}
\end{equation}

In this section we characterize preGPT, GPT and noncontextual ontological models as factorizations of COPE 4-tensors. 

\subsection{preGPT and GPT Models of $\pp\TT^1\TT^2\MM$ Scenario}
\label{sec:preGPT_GPT_seq}

For a $\pp\TT^1\TT^2\MM$ scenario in an operational theory with COPE tensor $C$, a preGPT model is given by a factorization such that 
\begin{equation}
    \label{eq:cijkl_decomp}C_{ijkl}=A_{i:}T^{2j}T^{1k}B_{:l}
\end{equation}
in a vector space of dimension $r$.
Here, the set of effects $\{e^i\}$ are the rows of the matrix $A$, the set of states $\{s^l\}$ are the columns of $B$, and the transformations are given by the $r\times r$ matrices $\{T^{1k}\}$ and $\{T^{2j}\}$.
Note also that in a $\pp\TT^1\TT^2\MM$ scenario, each transformation stage of the preGPT must preserve the unit effect.

The decomposition of the COPE 4-tensor in Eq.~\eqref{eq:cijkl_decomp} is an instance of a \textit{tensor-train} decomposition, typically written as
\begin{equation}
C_{ijkl}=\sum_{\alpha=1}^{d_1}\sum_{\beta=1}^{d_2}\sum_{\gamma=1}^{d_3}A_{i\alpha}T^{2j}_{\alpha\beta}T^{1k}_{\beta\gamma}B_{\gamma l}.
\label{eq:TT_decompCijkl}
\end{equation}
The decomposition in Eq.~\eqref{eq:cijkl_decomp} is defined over a vector space of dimension $r$, which amounts to the choice $d_i=r$ for $i=1,2,3$ in Eq.~\eqref{eq:TT_decompCijkl}.
Furthermore, the best-known algorithm for tensor-train decomposition is called the \textit{tensor-train singular-value-decomposition}~\cite{tensor-train}.

The converse of the above is also true, i.e., 
any tensor-train decomposition can be modified to give a valid preGPT model for the $\pp\TT^1\TT^2\MM$ scenario as shown below.

\begin{lemma}
\label{lemma:preGPT_pt1t2}
Consider a COPE tensor $C$ with a tensor-train decomposition of the form given in Eq.~\eqref{eq:TT_decompCijkl}. This decomposition can always be converted into a preGPT model for the $\pp\TT^1\TT^2\MM$ scenario.
\end{lemma}
\begin{proof}
We begin by replacing the outcome index $i$ with $(m,m_n)$, where $m$ denotes a measurement and $m_n$ its $n$-th outcome. The set of outcomes of measurement $m$ is denoted $O_m$. Following Eq.~\eqref{eq:cijkl_decomp}, the COPE tensor's tensor-train decomposition is written as $C_{(m,m_n)jkl}=A_{(m,m_n):}T^{2j}T^{1k}B_{:l}$. 
Here, $A$ has $d_1$ columns, $B$ has $d_3$ rows, and the transformation matrices $T^{1k}$ and $T^{2j}$ are of sizes $d_2\times d_3$ and $d_1\times d_2$ respectively.

We define the matrix $B''$ whose columns represent all possible combinations of states with transformations from the first and second layer, $B''_{:(jkl)}=T^{2j}T^{1k}B_{:l}$.
The sum of effects for measurement $m$ is given as
\begin{equation*}
u^{m\top}:=\sum_{m_n\in O_m}A_{(m,m_n):}.
\end{equation*}
For any preparation and pair of transformations, the sum of all measurement outcomes associated with measurement $m$ produces a row of ones in the COPE,
$u^{m\top} B''=\mathbf 1^{\mathsf T}$ for all $m$. We now choose a reference
measurement $m_0$ and set $u:=u^{m_0}$, and translate the rows of the $m$th block of $A$ by
\begin{equation*}
A_{(m,m_n):}\mapsto A_{(m,m_n):}+\frac{1}{|O_m|}(u-u^m)^\top.
\end{equation*}
Since $(u-u^m)^\top B''=0$ for all $m$, the product $AB''$ and consequently the COPE tensor are unchanged. The vector $u^\top$ is now the common unit vector for all measurements.

We now define the matrix $B'$ consisting of all possible preparations followed by one transformation in the first stage, $B'_{:kl}=T^{1k}B_{:l}$, and select a reference transformation in the second stage $T^{2j_0}$. The vector $v$, defined as $v^\top:=u^\top T^{2j_0}$, satisfies $v^\top B'=\mathbf 1^{\mathsf T}$. We select a column vector $y$ such that $u^\top y=1$, and shift the transformation matrices according to
$T^{2j}\mapsto \widetilde T^{2j}:=
T^{2j}+y\left(v^\top -u^\top T^{2j}\right)$.
The COPE tensor is unchanged because $(v^\top-u^\top T^{2j})B'=0$, and the new transformation matrices $\widetilde T^{2j}$ transport the unit vector $u^\top$ to the same output vector $v^\top$.

This process is repeated for the first transformation layer. We define a vector $w$ from a reference transformation $T^{1k_0}$, as 
$w^\top:=v^\top T^{1k_0}$. Similarly, we choose a column vector $z$ with $v^\top z=1$, and shift the set of transformation matrices in the first stage according to
$\widetilde T^{1k}
:=
T^{1k}+z\left(w^\top-v^\top T^{1k}\right)$.
The COPE tensor is unchanged since
$(w^\top-v^\top T^{1k})B=0$, and every transformation in the first stage now maps the vector $v^\top$ to a unique output vector $w^\top$.
The unit vector $u^\top$ is therefore
transported through each stage of transformations as
\begin{equation*}
u^\top
\xrightarrow{\;\widetilde T^{2j}\;}
v^\top
\xrightarrow{\;\widetilde T^{1k}\;}
w^\top,
\end{equation*}
independently of the transformation labels $j$ and $k$.

To ensure that the transformations fix the unit effect rather than transporting it uniquely to $v^\top$ and then $w^\top$, we append zero rows and duplicate nonzero columns to make the matrices square. Define the maximal latent dimension
$D_\ast:=\max\{d_1,d_2,d_3\}$ and duplicate columns of $A$ and the sets of matrices $T^{1k}$ and $T^{2j}$, and append zero rows to $B$ and the matrices $T^{1k}$ and $T^{2j}$ to form 
$\widehat A,\widehat B,\widehat T^{1k}$ and $\widehat T^{2j}$, so that $\widehat A$ has $D_\ast$ columns, $\widehat B$ has $D_\ast$ rows, and every $\widehat T^{1k}$ or $\widehat T^{2j}$ matrix is of size $D_\ast \times D_\ast$. This corresponds to duplicating entries in the vectors $u,v$ and $w$ to form vectors $\widehat u,\widehat v$ and $\widehat w$ of length $D_\ast$.
Now we have
\begin{equation*}
\widehat u^\top\,\widehat T^{2j}=\widehat v^\top,
\quad
\widehat v^\top\,\widehat T^{1k}=\widehat w^\top, \quad \widehat{A}\widehat{T}^{2j}\widehat{T}^{1k}\widehat{B}=A\tilde{T}^{2j}\tilde{T}^{1k}B.
\end{equation*}

Finally, we choose matrices $X_1,X_2,X_3\in GL(D_\ast)$ such that
\begin{equation*}
\widehat u^\top X_1=
\widehat v^\top X_2=
\widehat w^\top X_3=:u^\top_\ast,
\end{equation*}
and redefine all of our factors according to
\begin{equation*}
A^\ast=\widehat AX_1,\quad
T^{\ast2j}=X_1^{-1}\widehat T^{2j}X_2,
\end{equation*}
\begin{equation*}
T^{\ast1k}=X_2^{-1}\widehat T^{1k}X_3,
\quad
B^\ast=X_3^{-1}\widehat B.
\end{equation*}
These result in the decomposition\begin{equation*}
C_{(m,m_n)jkl}=A^\ast_{(m,m_n):}T^{\ast2j}T^{\ast 1k}B^\ast_{:l},
\end{equation*}
such that
\begin{equation}
    \sum_{m_n\in O_m}A^\ast_{(m,m_n):}=u^\top_\ast \ \forall \  \MM_m,
\end{equation}
and
\begin{equation*}
u^\top_\ast T^{\ast 2j}=u^\top_\ast,\quad
u^\top_\ast T^{\ast 1k}=u^\top_\ast \quad\forall \  j,k.
\end{equation*}
This decomposition is therefore a valid preGPT model for the $\pp\TT^1\TT^2\MM$ scenario.
\end{proof}
\noindent We note that the construction in the above Lemma can easily be extended to scenarios with an arbitrary number of  transformation stages.

As in the $\pp\TT^1\MM$ scenario, we isolate each stage of procedures in the above factorization by considering the corresponding flattenings of the COPE tensor. Each flattening admits a decomposition with a particular interpretation, as follows.
\begin{enumerate}
    \item $C_{[\e]} = AB''$, where each column of $C_{[\e]}$ is labelled by a preparation composed with one transformation from the first and the second stage. In the factorization this sequence of procedures is represented by a column of $B''$.
    \item $C_{[\pp]}=A''B$, where each row of $C_{[\pp]}$ is labelled by a measurement outcome composed with one transformation from the first and the second stage. In the factorization this sequence of procedures is represented by a row of $A''$. 
    \item $C_{[\TT^1]}=Y^1\Theta^1$, where each row of $C_{[\TT^1]}$ is labelled by a single transformation in the first stage, and represented in the factorization by a row of $Y^1$. Each  column of $C_{[\TT^1]}$ is labelled by a choice of preparation, second-stage transformation and measurement outcome, and represented by a column of $\Theta^1$.
    \item $C_{[\TT^2]}=Y^2\Theta^2$, where each row of $C_{[\TT^2]}$ is labelled by a single transformation in the second stage, and represented in the factorization by a row of $Y^2$. Each column of $C_{[\TT^2]}$ is labelled by a choice of preparation, first-stage transformation and measurement outcome, and represented by a column of $\Theta^2$.
\end{enumerate}

In each of the above factorizations, one factor contains the representatives of the set of atomic procedures singled out by the flattening, and the other contains the representatives of the composite procedures which make up the rest of the sentence.
When a preGPT as in Eq.~\eqref{eq:cijkl_decomp} for the scenario is provided, the factorizations of flattenings can be chosen such that
\begin{equation*}
        B''_{:(jkl)}=T^{2j}T^{1k}B_{:l} \ \text{ and } \ A''_{(ijk):}=A_{i:}T^{2j}T^{1k}.
\end{equation*}
This is simply the statement that composite procedures consisting of a preparation followed by two transformations, and of two transformations followed by a measurement, are represented in the preGPT in a way that is consistent with the causal structure encoded in Eq.~\eqref{eq:cijkl_decomp}.
Similarly, we must have
\begin{equation}
\label{eq:ythetaSeqDefinition}
    \begin{split}
        &\Theta^2 :=A^\top\otimes B',\\
        &\Theta^1 :=A^{\prime \top}\otimes B,\\
        &Y^1_{j,(r(m-1)+n)} :=T^{1j}_{m,n},\\
        &Y^2_{j,(r(m-1)+n)} :=T^{2j}_{m,n}.
    \end{split}
\end{equation}

For a GPT model of this scenario, the states, effects and first stage of transformations must be quotiented as in Eq.~\eqref{eq:quot_def}; the second stage of transformations $\{T^{2k}\}$ must also be quotiented according to the last line of Eq.~\eqref{eq:quot_def}.
Then, the conditions for the decomposition in Eq.~\eqref{eq:cijkl_decomp} to correspond to a GPT can be expressed in terms of the tensor flattenings as follows.
\begin{theorem}
Suppose the COPE 4-tensor $C$ admits a factorization of the form in Eq.~\eqref{eq:cijkl_decomp}, and define the matrices $Y^1$ and $Y^2$ from the sets of matrices $\{T^{1i},T^{2j}\}$ as in Eq.~\eqref{eq:ythetaSeqDefinition}. Then, the matrices $A$ and $B$ and sets of matrices $\{T^{1i}\}$ and $\{T^{2j}\}$ constitute a GPT model of the COPE if and only if the conditions
\label{thm:2stageGPTconditions}
\begin{equation}\label{eq:2stageGPTconditions}
\begin{split}
    &\rank(A)=\rank(C_{[\e]}),\\
    &\rank(B)=\rank(C_{[\pp]}),\\
    &\rank(Y^1)=\rank(C_{[\TT^1]}),\\
    &\rank(Y^2)=\rank(C_{[\TT^2]}),
\end{split}
\end{equation}
are satisfied.
\end{theorem}
\begin{proof}
    The argument proceeds identically to Theorem~\ref{thm:GPT model}. Each flattening in this set of equations isolates a particular stage of procedures, and rank equality certifies that the GPT representations of that stage are completely separated by the other stages. Therefore, satisfaction of all four equations guarantees that every procedure is represented uniquely as required by a GPT. The converse also holds: if any of the equations were not satisfied, there would necessarily exist distinct elements which failed to be separated, and the decomposition would represent a preGPT  but not a GPT. Hence the conditions of Eq.~\eqref{eq:2stageGPTconditions} are necessary and sufficient for the decomposition to be a GPT model.
\end{proof}

In the tensor-train decomposition of Eq.~\eqref{eq:TT_decompCijkl}, the tensor-train rank $(\chi_1,\chi_2,\chi_3)$ is defined as the minimum inner dimensions $(d_1,d_2,d_3)$ in which a decomposition is possible.
We call such a decomposition the \textit{minimal} tensor-train decomposition.
We claim that the preGPT constructed from the minimal tensor-train decomposition via the procedure in Lemma~\ref{lemma:preGPT_pt1t2} is a minimal GPT. To see this, note that the 
dimensions $\chi_1$ and $\chi_3$ in the decomposition are equal to $r_\e=\rank(C_{[\e]})$ and $r_\pp=\rank(C_{[\pp]})$ which guarantee $\rank(A)=r_\e$ and $\rank(B)=r_\pp$, respectively.
The dimension $\chi_2$, on the other hand, equals the rank of the middle unfolding of the COPE tensor, $C_{[\text{mid}]}:=C_{(ij),(kl)}$.
The latter is understood as the PM scenario involving $\TT^1$-transformed preparations and $\TT^2$-transformed measurements.

To show that the matrix $Y^2$ defined in Theorem~\ref{thm:2stageGPTconditions} has rank $r_{\TT^2}=\rank(C_{[\TT^2]})$, consider $C_{[\text{mid}]}=A'B'$ in which $(B')_{\beta,(k,l)}:=\sum_\gamma T^{1k}_{\beta \gamma}B_{\gamma l}$ and  $(A')_{(i,j),\beta}:=\sum_\alpha A_{i\alpha }T^{2j}_{\alpha \beta}$.
The inner dimension of this factorization is $\chi_2$, which is assumed to be minimal.
Such a factorization is only possible if $\rank (B')=\chi_2$, so that $B'$ has full row-rank. Now, the mode-$\TT^2$ unfolding factors as $C_{[\TT^2]}=Y^2(A^{\top}\otimes B')$.
Since $\rank(A^\top\otimes B') =\chi_1\chi_2$, the matrix $A^\top\otimes B'$ is full row-rank.
Right multiplication by this matrix is therefore injective and consequently
the row space of $Y^2$ has dimension equal to the row space of $C_{[\TT^2]}$, so that
\begin{equation*}
\operatorname{rank}(Y^2)=\operatorname{rank}(C_{[\TT^2]})=r_{\TT^2}.
\end{equation*}
An analogous argument can be made to show that $\rank(Y^1)=\rank(C_{[\TT^1]})=r_{\TT^1}$. 

Finally, the duplication and zero-padding of columns, and the invertible transformations used in Lemma~\ref{lemma:preGPT_pt1t2} to construct the preGPT do not affect the ranks of any factors, so that the resulting preGPT satisfies all the conditions of Theorem~\ref{thm:2stageGPTconditions} and thus it is in fact a GPT.
It is also guaranteed to be the smallest possible GPT for the scenario, since otherwise there would exist a tensor-train decomposition smaller than the minimal one.

The fact that the minimal tensor-train decomposition provides a GPT may tempt us to consider $\rank(A')=\rank(B')=\rank({C_{\rm{mid}}})$ as an alternative rank criterion in place of those using $\rank(C_{[\TT^1]})$ and $\rank(C_{[\TT^2]})$ in Eq.~\eqref{eq:2stageGPTconditions}.
However, minimality of the tensor-train decomposition is essential for obtaining a GPT.
In other words, a \textit{non-minimal} decomposition satisfying the conditions
\begin{equation}
\label{eq:2stageNOTGPT}
\begin{split}
&\rank(A)=\rank(C_{[\e]}),\\ &\rank(B)=\rank(C_{[\pp]}),\\  &\rank(B')=\rank(A')=\rank({C_{\rm{mid}}})
\end{split}
\end{equation}
may not represent a GPT by failing
to satisfy the last two lines of Eq.~\eqref{eq:2stageGPTconditions}.
This can be shown through an explicit counterexample.
Consider a
$6\times 6\times 6\times 6$ COPE 4-tensor 
with a tensor-train decomposition with the following factors:
\begin{align}\label{eq:NOTGPT_counter}
    &A_{i:}=
        \begin{pmatrix}
        1 & \cos(\theta_i) & \sin(\theta_i) & 1
        \end{pmatrix}, \quad &i=1,2,\dots ,6\\ &B_{:l}=1/2\begin{pmatrix}
            1 & \cos(\theta_l) & \sin(\theta_l) & 0
        \end{pmatrix}^\top, \quad &l=1,2,\dots ,6 \\[3pt] 
        &T^{11}=\begin{pmatrix}
            1 & 0 & 0 & 0 \\
            0 & \cos(\theta_1) & \sin(\theta_1) & 0 \\ 0 & -\sin(\theta_1) & \cos(\theta_1) & 0 \\ 0 & 0 & 0 & 1
        \end{pmatrix}, \\[3pt]
        &T^{1k}=\begin{pmatrix}
            1 & 0 & 0 & 1 \\
            0 & \cos(\theta_k) & \sin(\theta_k) & 0 \\ 0 & -\sin(\theta_k) & \cos(\theta_k) & 0 \\ 0 & 0 & 0 & 0
        \end{pmatrix}, \quad & k=2,3,\dots 6 \\[3pt]
        \text{and } \ &T^{2j}=\begin{pmatrix}
            1 & 0 & 0 & 1 \\
            0 & \cos(\theta_j) & \sin(\theta_j) & 0 \\ 0 & -\sin(\theta_j) & \cos(\theta_j) & 0 \\ 0 & 0 & 0 & 0
        \end{pmatrix}, \quad & j=1,2,\dots 6,
\end{align}
where $\theta_j = 2\pi j/6$. Note that the above decomposition is not minimal, because $\rank(C_\e)=\rank(C_\pp)=\rank({C_{\rm{mid}}})=3$, but $d_i=4$ for $i=1,2,3$.
Furthermore, it satisfies Eq.~\eqref{eq:2stageNOTGPT} as $\rank(A)=\rank(B)=\rank(A')=\rank(B')=3$.
Nevertheless,  $\rank(Y^1)>\rank(C_{[\TT^1]})$, so that the decomposition fails to represent a GPT.
The latter follows from the fact that the mode-$\TT^1$ slices of the COPE tensor imply that there is an operational equivalence between transformations in the first stage given by $\frac{1}{3}(\TT^{11}+\TT^{13}+\TT^{15})\cong \frac{1}{3}(\TT^{12}+\TT^{14}+\TT^{16})$. 
However, by inspecting the matrix factors in Eq.~\eqref{eq:NOTGPT_counter}, we see that $\frac{1}{3}(T^{11}+T^{13}+T^{15})\ne \frac{1}{3}(T^{12}+T^{14}+T^{16})$ because $T^{11}$ has a one in its final row, while all other $T^{1k}$ matrices have only zeroes.
The phenomenon above can be intuitively explained by the fact that redundancies in the representation of transformations may be obscured in the representations of transformed preparations and transformed measurements by the projections encoded in the representations of preparations and measurements.

\subsection{Ontological Models of $\pp\TT^1\TT^2\MM$ Scenario}

\label{sec:pT1T2M_ont}

\subsubsection{General Ontological Models}

Similarly to the previous section, an ontological model of this scenario is given by a response function matrix $R$, an epistemic state matrix $E$, and nonnegative stochastic transformation matrices $\{\Gamma^{1k}\}$ and $\{\Gamma^{2j}\}$ representing $\T^1$ and $\T^2$, respectively. These provide an ontological model of the COPE 4-tensor via the decomposition
\begin{equation}
\label{eq:PT1T2MCOPEdecomp}
    C_{ijkl}=R_{i:}\Gamma^{2j}\Gamma^{1k}E_{:l}
\end{equation}
analogous to Eq.~\eqref{eq:cijkl_decomp}.

Lemma~\ref{lemma:ont_nonnegpreGPT}, which established the equivalence between nonnegative preGPTs and ontological models in the $\pp\TT^1\MM$ scenario, can be trivially extended to the sequential $\pp\TT^1\TT^2\MM$ scenario.
Indeed, the proof of Lemma~\ref{lemma:ont_nonnegpreGPT} does not rely on having only one transformation stage.
Starting from an ontological model, one represents each ontic point by a Cartesian basis vector, so that response functions, epistemic states, and stochastic transition maps are represented by nonnegative matrices. Then Eq.~\eqref{eq:PT1T2MCOPEdecomp} is exactly a nonnegative preGPT factorization of the form in Eq.~\eqref{eq:cijkl_decomp}, with $A=R$, $B=E$, $T^{1k}=\Gamma^{1k}$, and $T^{2j}=\Gamma^{2j}$. 
Conversely, assume that a given preGPT has a nonnegative unit effect preserved by every transformation in both stages. Then the nonnegative diagonal rescaling from Lemma~\ref{lemma:ont_nonnegpreGPT} is employed to achieve
\begin{equation}
\begin{split}
&A T^{2j}T^{1k}B\\
&=(AD^\dagger)\left(D T^{2j}D^\dagger\right)\left(D T^{1k}D^\dagger\right)\left(D B\right)\\
&=R\Gamma^{2j}\Gamma^{1k}E,
\end{split}
\end{equation}
where $D=\operatorname{diag}(u)$. 
Since both sets of matrices $T^{2j}$ and $T^{1k}$ preserve the unit effect, the rescaled transformation matrices are column-stochastic on their nonzero columns, while the rescaled effect and state factors are, respectively, composed of disjoint column-stochastic sets of rows, and a column-stochastic matrix, on their nonzero columns.
Hence, the second transformation stage merely introduces an additional stochastic transition matrix in the ontological factorization; it does not alter the structural equivalence between ontological models and nonnegative preGPT models established in Lemma~\ref{lemma:ont_nonnegpreGPT}.

\subsubsection{NCOMs of $\pp\TT^1\TT^2\MM$ Scenario}

Following the characterization of
ontological models for the $\pp\TT^1\TT^2\MM$ scenario, we use the GPT conditions of Theorem~\ref{thm:2stageGPTconditions} to express the conditions that must be satisfied by an NCOM of the scenario.  This is justified because, following the argument in Lemma~\ref{lemma:NCOM_nonnegGPT}, an NCOM of the $\pp\TT^1\TT^2\MM$ scenario is simply a nonnegative GPT model.
Indeed, the argument is unchanged by the presence of a second transformation stage. By definition, an ontological model is noncontextual if and only if it assigns identical ontological representations to operationally equivalent procedures. Thus, in the two-transformation-stage scenario, in addition to the usual conditions for preparations and measurement outcomes, one must have
\begin{equation}
\label{eq:2transformationstage_quot}
\begin{split}
    \TT^1_i\simeq \TT^1_j \Leftrightarrow \Gamma^{1i}&=\Gamma^{1j},\\
    \TT^2_i\simeq \TT^2_j \Leftrightarrow \Gamma^{2i}&=\Gamma^{2j}.
\end{split}
\end{equation}
These are precisely the quotienting conditions which turn a nonnegative preGPT representation of the form in Eq.~\eqref{eq:PT1T2MCOPEdecomp} into a GPT representation, now applied separately to the two transformation stages. Therefore, noncontextuality is exactly the requirement that the nonnegative preGPT be quotiented over all atomic procedure types.

It follows that we must consider NMFs of the four flattenings of the COPE tensor given by $C_{[\e]} = RE''$, $C_{[\pp]} = R''E$, $C_{[\TT^1]}=G^1 F^1$, and $C_{[\TT^2]}=G^2 F^2$.
Using the conditions in Theorem~\ref{thm:2stageGPTconditions}, we now state the following result.

\begin{theorem}
\label{thm:2stageNCOMconditions}
    Suppose the nonnegative matrices $R$ and $E$ and sets of nonnegative matrices $\{\Gamma^{2j}\}, \{\Gamma^{1k}\}$ induce a decomposition of a COPE 4-tensor $C$ as in Eq.~\eqref{eq:PT1T2MCOPEdecomp}. Define the matrices $G^2$ and $G^1$ such that the $j$th row of $G^2$ is the row-flattened form of $\Gamma^{2k}$,
    and analogously for $G^1$. The decomposition constitutes an NCOM if and only if the following conditions are satisfied:
    \begin{equation}
    \label{eq:2stageNCOMconditions}
        \begin{split}
            &\rank(R)=\rank(C_{[\e]})\\
             &\rank(E)=\rank(C_{[\pp]}) \\
             &\rank(G^1)=\rank(C_{[\TT^1]}) \\
             &\rank(G^2)=\rank(C_{[\TT^2]})
        \end{split}
    \end{equation}
\end{theorem}
\begin{proof}
    Since an NCOM is a nonnegative GPT, the proof is simply the application of Theorem~\ref{thm:2stageGPTconditions} to nonnegative preGPTs.
\end{proof}

Now, we present a general normal form for NCOMs of this sequential scenario that allows us to set out our linear program for deciding whether an NCOM exists in this setting and computes one if it does.
\begin{lemma}
\label{lemma:sequentialnormalform}
Suppose a $\pp\TT^1\TT^2\MM$ scenario with COPE tensor $C$ admits an NCOM.
Denote by $\bar R$, $\bar E$, $\bar G^1$, and $\bar G^2$ the extremal factors obtained by applying Lemma~\ref{lemma:extremalfactor} to $C_{[\e]}$, Lemma~\ref{lemma:extremalfactor2} to $C_{[\pp]}$, and Lemma~\ref{lemma:extremalfactor} to $C_{[\TT^1]}$ and $C_{[\TT^2]}$, respectively.
Then, $C$ admits an equivalent NCOM
\begin{equation}
C_{ijkl}=\sum_{wxy}\bar R_{iw}\tilde\Gamma^{2j}_{wy}\tilde\Gamma^{1k}_{yx}\tilde E_{xl},
\end{equation}
in which $\tilde E$ is obtained by rescaling the rows of $\bar E$, and $\tilde \Gamma^{2j}$ and $\tilde \Gamma^{1k}$ are obtained by the application of nonnegative linear maps to the factors $\bar G^1$ and $\bar G^2$ respectively.
\end{lemma}

\begin{proof}
By assumption, $C$ admits an NCOM,
\begin{equation}
\label{eq:NCOMbyassumption}
C_{ijkl}=\sum_{wxy}R_{iw}\Gamma^{2j}_{wy}\Gamma^{1k}_{yx}E_{xl}.
\end{equation}

We denote by $G^1$ and $G^2$ the matrices formed by placing the row-flattenings of $\Gamma^{1k}$ and $\Gamma^{2j}$ in their respective rows. By applying Lemma~\ref{lemma:extremalfactor} to the flattenings $C_{[\e]}, C_{[\TT^1]}$ and $C_{[\TT^2]}$, and Lemma~\ref{lemma:extremalfactor2} to $C_{[\pp]}$, we compute the extremal factors $\bar{R}$, $\bar G^1, \bar G^2$ and $\bar{E}$ respectively. Following the argument in Lemmas~\ref{lemma:extremalfactor} and~\ref{lemma:extremalfactor2}, the matrices $R,E,G^1$ and $G^2$ associated with the original NCOM in Eq.~\eqref{eq:NCOMbyassumption} must satisfy the conditions
\begin{equation}
\label{eq:newfactors}
\begin{split}
&R=\bar RM^1,\\
&E=M^2\bar E,\\
&G^1=\bar G^1M^3,\\
&G^2=\bar G^2M^4,
\end{split}
\end{equation}
where the matrices $M^1$, $M^2$, $M^3$, and $M^4$ are nonnegative.
As in Lemma~\ref{lemma:extremalnormalform}, $\bar R$ may be chosen to be normalized, that is, each measurement block represented in $\bar R$ is column-stochastic, in which case $M^1$ is column-stochastic.

The matrices $M^1$ and $M^2$ may be absorbed into the adjacent transformation stages by writing
\begin{equation}
\begin{split} \label{eq:absorption}
\widehat\Gamma^{1k}&:=\Gamma^{1k}M^2,\\
\widehat\Gamma^{2j}&:=M^1\Gamma^{2j},\\
\end{split}
\end{equation}
The extremal factors $\bar R$ and $\bar E$, alongside the new transformation matrices $\widehat\Gamma^{1k}$ and $\widehat\Gamma^{2j}$, give rise to a new decomposition of the COPE tensor, given by
\begin{equation}
C_{ijkl}=\sum_{wxy}\bar R_{iw}\widehat\Gamma^{2j}_{wy}\widehat\Gamma^{1k}_{yx}\bar E_{xl}.
\end{equation}

Since in an NCOM the matrix $E$ of epistemic states is column-stochastic, and given that $E=M^2\bar{E}$, we see that 
\begin{equation}
    \mathbf 1^\top E=(\mathbf 1^\top M^2)\bar{E}=\mathbf 1 ^\top .
\end{equation}
We define the vector $v^\top:=\mathbf 1^\top M^2$ so that $v^\top\bar E=\mathbf1^\top$. The matrices $\Gamma^{1k}$ are also column-stochastic, so that the first line of Eq.~\eqref{eq:absorption} implies
\begin{equation}
    \mathbf{1}^\top\widehat\Gamma^{1k} =(\mathbf{1}^\top\Gamma^{1k})M^2\ =\mathbf{1}^\top M^2 = v^\top
\end{equation}
Finally, recalling that in the first line of Eq.~\eqref{eq:newfactors} we choose the matrix $\bar{R}$  such that $M^1$ is column-stochastic, the second line of Eq.~\eqref{eq:absorption} implies
\begin{equation}
\mathbf{1}^\top\widehat\Gamma^{2j} = \mathbf 1^\top M^1\Gamma^{2j}=\mathbf{1}^\top
\end{equation}
since the $\Gamma^{2j}$ matrices are also column-stochastic.

The vector $v$ is then used to rescale the new matrix factors, via $D:=\operatorname{diag}(v)$, and its pseudoinverse, $D^\dagger$. We define the rescaled epistemic state and transformation matrices 
\begin{equation}
\begin{split} \label{eq:diagonalscaling}
&\tilde E:=D\bar E,\\
&\tilde\Gamma^{1k}:=\widehat\Gamma^{1k}D^\dagger,\\
&\tilde\Gamma^{2j}:=\widehat\Gamma^{2j}.
\end{split}
\end{equation}
Note that the rescalings preserve the COPE tensor. To ensure column-stochasticity of $\tilde\Gamma^{1k}$, one may simply overwrite each zero column $\tilde\Gamma^{1k}_{:x}$ for $v_x=0$ with a nonzero column $\tilde\Gamma^{1k}_{:x_0}$ for some fixed $x_0$, where $v_{x_0}>0$. Thus, all factors are now normalized.

The absorption of $M^1$ and $M^2$ in Eq.~\eqref{eq:absorption}, restriction to the retained coordinates, and positive diagonal rescaling in Eq.~\eqref{eq:diagonalscaling}, induce fixed nonnegative linear maps on the row-flattened transformation matrices. Consequently, there exist nonnegative matrices $L^1$ and $L^2$ such that
$\tilde G^1=G^1L^1$ and $\tilde G^2=G^2L^2$.
Combining these relations with the extremal factors given in Eq.~\eqref{eq:newfactors} results in $\tilde G^1=\bar G^1H^1$ and $\tilde G^2=\bar G^2H^2$, where $H^1:=M^3L^1$ and $H^2:=M^4L^2$ are nonnegative.

The new matrices $\bar R, \tilde G^1, \tilde G^2$ and $\tilde E$ reproduce the COPE tensor $C$, so that their ranks are bounded below by $\rank(C_{[\e]})$, $\rank(C_{[\TT^1]})$, $\rank(C_{[\TT^2]})$, and $\rank(C_{[\pp]})$, respectively. At the same time, they can all be expressed in terms of matrix factorizations which include the corresponding extremal factors $\bar R$, $\bar G ^1$, $\bar G ^2$, and  $\bar E$, whose ranks are $\rank(C_{[\e]})$, $\rank(C_{[\TT^1]})$, $\rank(C_{[\TT^2]})$, and $\rank(C_{[\pp]})$, respectively, as shown in Lemmas~\ref{lemma:extremalfactor} and~\ref{lemma:extremalfactor2}. Hence, $\bar R$, $\tilde G^1$, $\tilde G^2$, and $\tilde E$ also have ranks bounded above by $\rank(C_{[\e]})$, $\rank(C_{[\TT^1]})$, $\rank(C_{[\TT^2]})$, and $\rank(C_{[\pp]})$, respectively. Therefore, all four rank conditions of Theorem~\ref{thm:2stageNCOMconditions} are preserved. 

Finally, if necessary, the columns of $\bar R$ and the columns of the final transformation matrices may be duplicated, and zero rows may be appended to the final transformation matrices and $\tilde E$, to arrive at a common dimension without altering the COPE tensor, normalization, or rank conditions.
\end{proof}
We now formalize the linear program as follows.
\begin{theorem}
\label{thm:sequential}
It can be decided whether a $\pp\TT^1\TT^2\MM$ scenario in an operational theory with $n_\pp$ preparations, $n_\e$ measurement events, $n_{\TT^1}$ first-stage transformations, $n_{\TT^2}$ second-stage transformations, and maximum Tucker rank component $r$ of its COPE tensor admits an NCOM in time
$\poly((n_\e n_\pp n_{\TT^1}n_{\TT^2})^{O(r)})$.
\end{theorem}

\begin{proof}
By Lemma~\ref{lemma:sequentialnormalform}, if an NCOM exists, then one exists using the extremal factors $\bar R$, $\bar E$, $\bar G^1$, and $\bar G^2$, where $\tilde E=D^v\bar E$, $\tilde G^1=\bar G^1H^1$, and $\tilde G^2=\bar G^2H^2$, where $D^v=\operatorname{diag}(v)$ is a nonnegative diagonal matrix, and matrices $H^1$ and $H^2$ are also nonnegative. Therefore, it suffices to search for an NCOM in this normal form.

We denote the number of columns of $\bar R$ by $d_R$, the number of rows of $\bar E$ by $d_E$, and the numbers of columns of $\bar G^1$ and $\bar G^2$ by $m_1$ and $m_2$ respectively. To encode the nonnegative maps on $\bar G^1$ and $\bar G^2$, we introduce the following compositions of nonnegative maps,
\begin{equation}
\begin{split}
&K_{(b,a),(w,x)}:=v_x\sum_y H^2_{b,(w,y)} H^1_{a,(y,x)},\\
&u_{(a,x)}:=v_x\sum_y H^1_{a,(y,x)}.
\end{split}
\label{eq:sequential-lift}
\end{equation}

Here, $H^1$ and $H^2$ represent the nonnegative maps on $\bar G^1$ and $\bar G^2$, respectively, up to a diagonal scaling, and $v$ and $u$ represent the diagonal scalings of the factors. The matrix $K$ is the nonnegative rescaling of the tensor product of $H^1$ and $H^2$, which is treated as a single matrix variable in the following linear program. In particular, the linear program need not explicitly solve for the individual maps $H^1$ and $H^2$.

Now, the linear program to decide the existence of a NCOM of the sequential $\pp\TT^1\TT^2\MM$ scenario is as follows in terms of the variables $K,u,v \geq 0$,
\begin{align}
&\left[(\bar G^2\otimes\bar G^1)K(\bar R^\top\otimes\bar E)\right]_{(j,k),(i,l)}:=\\
&\sum_{a,b,w,x} \bar G^2_{jb}\bar G^1_{ka}K_{(b,a),(w,x)}\bar R_{iw}\bar E_{xl}=C_{ijkl}\quad \forall \  i,j,k,l,\\
&\sum_x v_x\bar E_{xl}=1 \quad\forall \  l \label{eq:sequentialE} \\
&\sum_a\bar G^1_{ka}u_{(a,x)}=v_x\quad\forall \  k,x, \label{eq:sequentialG1}\\
&\sum_{b,w}\bar G^2_{jb}K_{(b,a),(w,x)}=u_{(a,x)}\quad\forall \  j,a,x. \label{eq:sequentialG2}
\end{align}
Then, given a solution to the above linear program, we obtain the following matrices. Using the notation $D^v:=\operatorname{diag}(v)$ and $D^u:=\operatorname{diag}(u)$, and using $\dagger$ to denote the pseudoinverse, we set
\begin{align}
&\tilde E_{xl} :=v_x\bar E_{xl}, \label{eq:sequentialrecoveredE}\\
&\tilde\Gamma^{1k}_{(a,x'),x} :=\delta_{x,x'}\bar G^1_{ka}u_{(a,x)}D^{v\dagger}_{xx},\label{eq:sequentialrecoveredG1}\\
&\tilde\Gamma^{2j}_{w,(a,x)} :=\sum_b\bar G^2_{jb}K_{(b,a),(w,x)}D^{u\dagger}_{(a,x),(a,x)}. \label{eq:sequentialrecoveredG2}
\end{align}

We see that Eq.~\eqref{eq:sequentialG2} combined with Eq.~\eqref{eq:sequentialrecoveredG2} ensures that each nonzero column of each $\tilde \Gamma^{2j}$ matrix is stochastic. Similarly, by Eq.~\eqref{eq:sequentialG1} and Eq.~\eqref{eq:sequentialrecoveredG1}, each nonzero column of each $\tilde \Gamma^{1k}$ is stochastic, and by Eq.~\eqref{eq:sequentialE} and Eq.~\eqref{eq:sequentialrecoveredE}, $\tilde E$ is column-stochastic.

To ensure full column-stochasticity of the transformation matrices, one may choose one index $x_0$ with $v_{x_0}>0$ and one index $(a_0,x_1)$ with $u_{(a_0,x_1)}>0$. For every $x$ with $v_x=0$, overwrite the $x$th column of every $\tilde\Gamma^{1k}$ by a copy of its $x_0$th column. For every $(a,x)$ with $u_{(a,x)}=0$, overwrite the corresponding column of every $\tilde\Gamma^{2j}$ by a copy of its $(a_0,x_1)$th column. Since $\tilde E$ and $\tilde \Gamma^{1k}$ have zero-rows in the corresponding positions of $v_x=0$ and $u_{(a,x)}=0$, respectively, this overwriting procedure preserves the COPE tensor.

The final step is to pad the dimensions of the factor matrices to a common ontic dimension if necessary. The epistemic state matrix $\tilde E$ is padded by appending zero rows, while the response function matrix $\tilde R$ is padded by duplicating an existing column. Each transformation matrix is enlarged compatibly by appending zero rows and duplicating an existing normalized column. These additions preserve the reconstructed COPE tensor, as the new ontic states are never populated.

Since we obtained the final factors by linear transformations of the extremal factors, their ranks equal the ranks of the corresponding COPE unfoldings. Therefore, the final factors satisfy the rank conditions of Theorem~\ref{thm:2stageNCOMconditions} and constitute an NCOM.

Now we consider the complexity. By Lemma~\ref{lemma:extremalfactor}, $\bar R$ has $n_\e^{O(r)}$ columns, $\bar E$ has $n_\pp^{O(r)}$ rows, $\bar G^1$ has $n_{\TT^1}^{O(r)}$ columns, and $\bar G^2$ has $n_{\TT^2}^{O(r)}$ columns. Hence $K$ has $(n_\e n_\pp n_{\TT^1}n_{\TT^2})^{O(r)}$ entries, and the number of variables and constraints in the above linear program is polynomial in this quantity. The decision procedure consequently runs in time
\begin{equation}
\poly((n_\e n_\pp n_{\TT^1}n_{\TT^2})^{O(r)}).
\end{equation}
\end{proof}

\section{Extension to COPE $k$-Tensors}

The analysis applied to PTM scenarios with two transformation stages in Sec.~\ref{sec:sequential} can be easily extended to scenarios with arbitrarily many transformation stages, i.e., the $\pp\TT^1\TT^2...\TT^{k-2}\MM$ scenario for arbitrary but finite $k$.

In this case, the probabilistic structure of the operational theory is given by a COPE $k$-tensor $C_{i_1i_2...i_k}$. Here, the last index $i_k$ labels preparations and the first index $i_1$ labels measurement outcomes. Furthermore, the index $i_{k-1}$ labels first-stage transformations, $i_{k-2}$ labels second-stage transformations, and so on, down to $i_2$, which labels $(k-2)$th stage transformations. 
As in $\pp\TT^1\MM$ and $\pp\TT^1\TT^2\MM$ scenarios, the procedures in each stage appearing in the COPE tensor are atomic. Now, the operational equivalence between two transformations in the $l$th stage is denoted by
\begin{equation} \label{eq:transformationk-2layereqv}
\begin{split}
\TT^l_i \simeq \TT^l_j \; &\text{iff} \; p(\e| \pp, \TT^1_{i_{k-1}},...,\TT^l_i,...,\TT^{k-2}_{i_2}, \MM) \\
&=p(\e| \pp, \TT^1_{i_{k-1}},...,\TT^l_j,...,\TT^{k-2}_{i_2}, \MM) \\ &\forall \  \pp, \e, \MM, \TT^q_{i_{k-q}}
\text{ with }q\neq l. \\
\end{split}
\end{equation}
In parallel with the previous section, we characterize preGPT, GPT and noncontextual ontological models as factorizations of COPE $k$-tensors. 

\subsection{PreGPT and GPT Models}

For a sequential-transformation scenario in an operational theory with a COPE $k$-tensor $C$, a preGPT model is given by a factorization of the form,
\begin{equation}
    C_{i_1...i_k}=A_{i_1:}T^{(k-2)i_2}...T^{2i_{k-2}}T^{1i_{k-1}}B_{:i_k}
\end{equation}
in a vector space of dimension $r$.

It is clear that such a preGPT model can be found by first applying the tensor-train singular-value decomposition step of Sec.~\ref{sec:sequential}. Then, to arrive at a single consistent unit effect resulting from the matrix $A$, we may apply the row-translation step shown in the proof of Lemma~\ref{lemma:preGPT_pt1t2}. Finally, in order to preserve the unit effect after each transformation stage, we iterate the transformation-translation step used in the proof of Lemma~\ref{lemma:preGPT_pt1t2}.

In order to isolate each stage of procedures in the above factorization, we consider the corresponding flattenings of the COPE tensor.

To begin, we denote by $A^{l}$ the matrix whose rows each represent the composition of a choice of outcome, a $(k-2)$th-stage transformation, a $(k-3)$th-stage transformation, and so on, down to an $l$th-stage transformation. Note that, as a result, $A^l$ has $n_\e^\MM n_\MM n_{\TT^{k-2}}\cdots n_{\TT^l}$ rows.
Similarly, we denote by $B^{l}$ the matrix whose columns each represent the composition of a choice of preparation, a first-stage transformation, a second-stage transformation, and so on, up to an $l$th-stage transformation.
This matrix has $ n_{\TT^{l}} \cdots n_{\TT^1} n_\pp$ columns.
Furthermore, we denote by $\Theta^l$ the matrix whose columns each represent a choice of each stage of procedures except the $l$th transformation stage, and denote by $Y^l$ the matrix whose rows each represent an $l$th-stage transformation. Specifically, with the convention $A^{(k-1)}:=A$, and $B^0:=B$,
\begin{equation}
\begin{split}
    &\Theta^l:=A^{(l+1)\top}\otimes B^{(l-1)},\\
    &Y^l_{j,(r(m-1)+n)}:=T^{lj}_{m,n},
\end{split}
\end{equation}
where we recall that $r$ is the preGPT dimension.

Now we are ready to construct the relevant flattenings of $C$:
\begin{equation}
    \begin{split}
        C_{[\e]}&=AB^{(k-2)},\\
        C_{[\pp]}&=A^{1}B,\\
        C_{[\TT^l]}&=Y^l\Theta^l\quad\forall \  l.
    \end{split}
\end{equation}
It is clear that the proof of Theorem~\ref{thm:2stageGPTconditions} can be extended to show that a preGPT model of the COPE $k$-tensor is a GPT if and only if the following conditions are simultaneously satisfied:
\begin{equation} \label{eq:kGPTconditions}
    \begin{split}
    &\rank(A)=\rank(C_{[\e]}),\\
    &\rank(B)=\rank(C_{[\pp]}),\\
    &\rank(Y^l)=\rank(C_{[\TT^l]})\quad\forall \  l.
\end{split}
\end{equation}
As in Sec.~\ref{sec:sequential}, a preGPT constructed from the minimal tensor-train decomposition of $C$ is a GPT.

\subsection{NCOMs}

In an ontological model of the $\pp\TT^1...\TT^{k-2}\MM$ scenario, response functions are represented by rows of the matrix $R$, and epistemic states are represented by columns of the matrix $E$. Transformations in the $l$th stage of transformations, consisting of $n_{\TT^l}$ many transformations, are represented by transformation matrices $\Gamma^{li_{k-l}}$ for $i_{k-l}=1,...,n_{\TT^l}$. Thus, an ontological model of the COPE $k$-tensor is expressed by the decomposition
\begin{equation}
    \label{eq:kCOPEdecomp}
    C_{i_1...i_k}=R_{i_1:}\Gamma^{(k-2)i_2}...\Gamma^{1i_{k-1}}E_{:i_k}.
\end{equation}
Furthermore, it is immediate from Lemma~\ref{lemma:NCOM_nonnegGPT} and Theorem~\ref{thm:2stageNCOMconditions} that the conditions for a NCOM of the $\pp\TT^1...\TT^{k-2}\MM$ scenario can be expressed in the same form as the conditions for the GPT model of the $\pp\TT^1...\TT^{k-2}\MM$ scenario.

\begin{theorem}
    Suppose the nonnegative matrices $R$ and $E$ and sets of nonnegative matrices $\{\Gamma^{l(i_{k-l})}\}$ for $l=1,...,k-2$, induce a decomposition of a COPE $k$-tensor $C$ as in Eq.~\eqref{eq:kCOPEdecomp}. Define the matrices $G^l$ such that the $j$th row of $G^l$ is the row-flattened form of $\Gamma^{lj}$. The decomposition constitutes an NCOM if and only if the following conditions are satisfied:
    \begin{equation}
    \label{eq:kstageNCOMconditions}
        \begin{split}
            &\rank(R)=\rank(C_{[\e]})\\
             &\rank(E)=\rank(C_{[\pp]}) \\
             &\rank(G^l)=\rank(C_{[\TT^l]})\quad\forall \  l.
        \end{split}
    \end{equation}
\end{theorem}
\begin{proof}
    Since a NCOM is a nonnegative GPT up to a nonnegative rescaling, the proof is simply an application of the GPT conditions in Eq.~\eqref{eq:kGPTconditions} applied to nonnegative preGPTs.
\end{proof}
In the following Theorem, we extend the linear program set out in Theorem~\ref{thm:sequential} to decide the contextuality of $\pp\TT^1\TT^2\MM$ scenarios, to scenarios with $(k-2)$ stages of transformations. 

\begin{theorem}
\label{thm:ksequential}
It can be decided whether a $\pp\TT^1...\TT^{k-2}\MM$ scenario in an operational theory with $k\geq4$ total stages, $n_\pp$ preparations, $n_\e$ measurement events, $n_{\TT^l}$ $l$-stage transformations for all $l=1,...,k-2$, and maximum Tucker rank component $r$ of its COPE tensor admits an NCOM in time
$\poly((n_\e n_\pp n_{\TT^1}...n_{\TT^{k-2}})^{O(r)})$.
\end{theorem}

\begin{proof}
We give an outline of the proof for the $\pp\TT^1...\TT^{k-2}\MM$ scenario with $k\geq4$.

We now apply the same logic used to restrict the candidate factors in the $\pp\TT^1\MM$ and $\pp\TT^1\TT^2\MM$ scenarios to each stage of procedures. Suppose that an NCOM exists, with response function matrix $R$, epistemic-state matrix $E$, and transformation factors $G^l$, where the $i_{k-l}$th row of $G^l$ is the row-flattening of $\Gamma^{li_{k-l}}$, for $l=1,\ldots,k-2$. By Lemmas~\ref{lemma:extremalfactor} and~\ref{lemma:extremalfactor2}, there exist nonnegative matrices $M^1,M^2,M^3,\ldots,M^k$ such that
\begin{equation}
\begin{split}
&R=\bar R M^1,\\
&E=M^2\bar E,\\
&G^l=\bar G^lM^{l+2},\qquad l=1,\ldots,k-2,
\end{split}
\end{equation}
where $\bar R,\bar E,\bar G^1,\ldots,\bar G^{k-2}$ are the corresponding extremal factors.

As in Lemma~\ref{lemma:sequentialnormalform}, the maps on $R$ and $E$ may be absorbed into the last and first transformation stages, respectively. After applying the required diagonal normalization, there is therefore an equivalent NCOM of the form
\begin{equation}
\begin{split}
&\tilde R =\bar R,\\
&\tilde E =D^v\bar E,\\
&\tilde G^l =\bar G^lH^l,\qquad l=1,\ldots,k-2,
\end{split}
\end{equation}
where $D^v:=\operatorname{diag}(v)$, its pseudoinverse $D^{v\dagger}$, and the $H^l$ matrices are nonnegative. Thus, it suffices to search for an NCOM in this normal form.

Solving for the nonnegative maps $D^v,H^1,\ldots,H^{k-2}$ directly would be computationally inefficient. Therefore, we consider their compositions. The columns of $\bar G^l$ are indexed by $a_l$, the rows of $\bar E$ by $x$, and the columns of $\bar R$ by $w$. For an NCOM in the above normal form, we construct the successive marginals of the nonnegative maps $D^v,H^1,\ldots,H^{k-2}$ as,
\begin{equation} \label{eq:kmarginals}
\begin{split}
&u^l_{(a_l,\ldots,a_1,x)}:=v_x\sum_{\lambda_1,\ldots,\lambda_l}\prod_{q=1}^lH^q_{a_q,(\lambda_q,\lambda_{q-1})}
\end{split}
\end{equation}
for $l=1,\ldots,k-3$, and their compositions,
\begin{equation}
K_{(a_{k-2},\ldots,a_1),(w,x)}:=v_x\sum_{\lambda_1,\ldots,\lambda_{k-3}}\prod_{q=1}^{k-2}H^q_{a_q,(\lambda_q,\lambda_{q-1})},
\end{equation}
where $\lambda_0=x$ and $\lambda_{k-2}=w$. The marginals in Eq.~\eqref{eq:kmarginals} will be used to enforce normalizability on the final transformation factors.

Consequently, the existence of an NCOM can be expressed as feasibility of the following linear constraints:
\begin{align}
&\sum_{a_1,\ldots,a_{k-2},w,x}(\bar R_{i_1w}K_{(a_{k-2},\ldots,a_1),(w,x)}\bar E_{xi_k}\prod_{l=1}^{k-2}\bar G^l_{i_{k-l}a_l})\\
&=C_{i_1i_2\cdots i_k}\quad \forall \  i_1,i_2,\ldots,i_k, \label{eq:ksequential-reconstruction}\\
&\sum_xv_x\bar E_{xi_k}=1\quad \forall \  i_k \label{eq:ksequential-E}\\
&\sum_{a_1}\bar G^1_{i_{k-1}a_1}u^1_{(a_1,x)}=v_x\quad \forall \  i_{k-1},x, \label{eq:ksequential-first}\\
&\sum_{a_l}\bar G^l_{i_{k-l}a_l}u^l_{(a_l,\ldots,a_1,x)}=u^{l-1}_{(a_{l-1},\ldots,a_1,x)}\\
&\forall \  i_{k-l},a_{l-1},\ldots,a_1,x,\quad l=2,\ldots,k-3, \label{eq:ksequential-middle}\\
&\sum_{a_{k-2},w}\bar G^{k-2}_{i_2a_{k-2}}K_{(a_{k-2},\ldots,a_1),(w,x)}=u^{k-3}_{(a_{k-3},\ldots,a_1,x)}\\
&\forall \  i_2,a_{k-3},\ldots,a_1,x, \label{eq:ksequential-final}\\
&K,u^1,\ldots,u^{k-3},v\geq0. \label{eq:ksequential-nonnegative}
\end{align}
Then, given a solution to the above linear program, we obtain the following matrices. Using the notation $D^v:=\operatorname{diag}(v)$ and $D^{u^l}:=\operatorname{diag}(u^l)$, we set
\begin{align}
    &\tilde E_{xi_k}:=v_x\bar E_{xi_k},\label{eq:krecoveredE} \\
    &\tilde\Gamma^{1i_{k-1}}_{(a_1,x),x}:=\bar G^1_{i_{k-1}a_1}u^1_{(a_1,x)}D^{v\dagger}_{x,x},\label{eq:ksequential-recovered-first} \\
    &\tilde\Gamma^{li_{k-l}}_{(a_l,\ldots,a_1,x),(a_{l-1},\ldots,a_1,x)}:=\\
    &\bar G^l_{i_{k-l}a_l}u^l_{(a_l,\ldots,a_1,x)}D^{u^{l-1}\dagger}_{(a_{l-1},\ldots,a_1,x),(a_{l-1},\ldots,a_1,x)}\label{eq:ksequential-recovered-middle}\\
    &\forall \  l=2,...,k-3,\\
    &\tilde\Gamma^{(k-2)i_2}_{w,(a_{k-3},\ldots,a_1,x)}:=\\
    &\sum_{a_{k-2}}\bar G^{k-2}_{i_2a_{k-2}}K_{(a_{k-2},\ldots,a_1),(w,x)}D^{u^{k-3}\dagger}_{(a_{k-3},\ldots,a_1,x),(a_{k-3},\ldots,a_1,x)}. \label{eq:ksequential-recovered-final}
\end{align}

We see that equation~\eqref{eq:ksequential-reconstruction} ensures the COPE tensor is reconstructed. Eq.~\eqref{eq:ksequential-E} combined with Eq.~\eqref{eq:krecoveredE} ensures that $\tilde E$ is column-stochastic. Similarly, equations~\eqref{eq:ksequential-first}, \eqref{eq:ksequential-middle}, and~\eqref{eq:ksequential-final} combined with equations~\eqref{eq:ksequential-recovered-first}, \eqref{eq:ksequential-recovered-middle}, and \eqref{eq:ksequential-recovered-final}, respectively, ensure that every nonzero column of the reconstructed transformation matrices is stochastic. Columns corresponding to zero entries of $v,u^1,\ldots,u^{k-3}$ may be overwritten by copies of nonzero columns. Such overwritten columns correspond to ontic points that are never populated, so this replacement does not change the COPE tensor.

Substituting Eqs.~\eqref{eq:krecoveredE}, \eqref{eq:ksequential-recovered-first}, \eqref{eq:ksequential-recovered-middle}, and~\eqref{eq:ksequential-recovered-final} into the sequential tensor contraction of Eq.~\eqref{eq:kCOPEdecomp} causes the successive rescaling factors, $D^v,D^{u^l}$, and their pseudoinverses to cancel on their supports. In particular, observe that Eq.~\eqref{eq:ksequential-recovered-middle} has a recursive structure: for each $l$ in the allowed range, the rescaling factor is $D^{u^l}D^{u^{l-1}\dagger}$. Then, the resulting contraction is precisely the left-hand side of Eq.~\eqref{eq:ksequential-reconstruction}, and therefore reproduces $C$.

The reconstructed row-flattened transformation factors are obtained by nonnegative linear maps from the extremal factors $\bar G^l$, and $\tilde E$ is obtained by a nonnegative diagonal rescaling of $\bar E$. Hence, the final factors were obtained by nonnegative linear transformations of the extremal factors, thus their ranks cannot increase after such transformations. Therefore, all rank conditions required for an NCOM are satisfied. If necessary, the different intermediate ontic spaces may be padded to a common dimension by appending zero rows and duplicating normalized columns, without changing the COPE tensor or the rank conditions.

It remains to bound the complexity. By Lemmas~\ref{lemma:extremalfactor} and~\ref{lemma:extremalfactor2}, $\bar R$ has at most $n_\e^{O(r)}$ columns, $\bar E$ has at most $n_\pp^{O(r)}$ rows, and $\bar G^l$ has at most $n_{\TT^l}^{O(r)}$ columns for every $l=1,\ldots,k-2$. Therefore, $K$ has at most
\begin{equation}\label{eq:ksize}
n_\e^{O(r)}n_\pp^{O(r)}\prod_{l=1}^{k-2}n_{\TT^l}^{O(r)}=(n_\e n_\pp n_{\TT^1}...n_{\TT^{k-2}})^{O(r)}
\end{equation}
entries. Every marginal variable $u^l$, and every family of constraints in the above linear program, has a size bounded by a polynomial in the quantity in Eq.~\eqref{eq:ksize}. Hence, feasibility can be decided by a linear program in time
\begin{equation}
\poly((n_\e n_\pp n_{\TT^1}...n_{\TT^{k-2}})^{O(r)}).
\end{equation}
\end{proof}

\section{Application to Known Toy Models}\label{sec:application}

\subsection{Spekkens' Toy Theory}

To show that the criteria of Theorem~\ref{thm:PTM-NC-rank-cnstrnts} certify noncontextuality, we apply them to several well-known examples, starting with Spekkens' toy theory.
This is an operational theory that reproduces many qualitative features of quantum mechanics, including entanglement, non-commutativity, and interference, but it lacks contextuality~\cite{SpekkensToy,SpekkensToyReview}.
The PTM scenario in this theory admits an NCOM defined over the ontic space $\Lambda=\{\lambda_1,\lambda_2,\lambda_3,\lambda_4\}$. There are six epistemic states, each with equal support on a pair of ontic states, which we denote
\begin{equation}\label{eq:STT_ES}
\mu^{i\lor j}(\lambda)=\begin{cases}
    1/2, \lambda \in\{\lambda_i,\lambda_j\} \\ 0 \ \ \ \ \text{ otherwise.}
\end{cases}
\end{equation}
The response functions are defined similarly as 
\begin{equation}\label{eq:STT_RF}
\xi^{i\lor j}(\lambda)=\begin{cases}
    1, \lambda \in\{\lambda_i,\lambda_j\} \\ 0 \ \ \ \ \text{ otherwise.}
\end{cases}
\end{equation}
Since $\Lambda$ can be partitioned into two equal subsets in only three distinct ways, the observer has three binary measurements available.

The allowed transformations correspond to relabellings of the ontic points. These are represented by the 24 distinct $4\times 4 $ permutation matrices. As every row and column must sum to 1, there are 6 constraints on each matrix; the final row/column will be determined by the others. Considering the matrices as vectors in $\mathbb{R}^{16}$ placed under 6 constraints, we see that only $10$ are linearly independent. 

The COPE tensor of Spekkens' toy theory can be written out explicitly using its six pure states, twenty four permutations, and three binary measurements (six events).
Using the matrices of epistemic states in Eq.~\eqref{eq:STT_ES}, matrices of response functions in Eq.~\eqref{eq:STT_RF}, and the transformation matrices (permutations), we find
\begin{equation}
\begin{split}
&\rank(C_{[\e]})=\rank(R)= 4, \\
&\rank(C_{[\pp]})=\rank(E)= 4, \\
&\rank(C_{[\TT]})=\rank(G)= 10.
\end{split}
\end{equation}
Therefore, the model indeed provides a noncontextual representation of the statistics. We note that the rank equalities above extend to any number of transformation stages, certifying that the model remains universally noncontextual as more stages are added.

\subsection{8-state Theory for Single Qubit stabilizer}

\label{sec:8state}

We now demonstrate our criteria on a theory which is known to exhibit transformation contextuality: The single-qubit stabilizer subtheory~\cite{Lillystone2019,Schmid_2024}. Note that we will use the operator algebra of quantum theory to describe a GPT model for this theory, as is common practice in the literature.

A GPT model of this operational theory is given by the following operator representations: possible
measurements are the Paulis $\mathcal{M}_{\text{1QS}}=\{X,Y,Z\}$, possible preparations are their $\pm 1$ eigenprojector $\mathcal P_{\text{1QS}}=\{X^\pm,Y^\pm,Z^\pm\}$, and
the transformations are $\mathcal T_{\text{1QS}}=\{\tau_{\mathds{1}}, \tau_Z, \tau_S, \tau_{S^{-1}}\}$, where $S=\sqrt Z$ is the phase gate. 
In particular, $\tau_G$ is the channel $\tau^{(G)}(\cdot):=G(\cdot)G^\dagger$. Indeed, we see that this model is quotiented under operational equivalences, since the representations of the preparations, the preparations composed with the transformations, the transformations composed with the measurements, and the measurements, all span the full Bloch sphere. Thus, no stage of procedures is given an underspecified representation.

For an ontological model of the theory, we consider an
ontic space $\Lambda$ with size $s:=|\Lambda|$. The epistemic states representing the eigenprojector of operator $O$ will be $\mu^{O^\pm}(\lambda)$, and the corresponding response functions representing measurement outcomes will be $\xi^{O^\pm}(\lambda)$.
The transformations $\tau_G$ are represented by stochastic matrices $\Gamma^{(G)}$. 

Any pair of $\pm 1$ eigenprojector of a fixed Pauli operator is \textit{single-shot distinguishable} (SSD), meaning that there exists a measurement which perfectly discriminates them.  As noted in Ref.~\cite{Spekkens_2005}, any two SSD states must be represented by orthogonal epistemic state vectors. To see this, recall that for any complete measurement $\MM_j$ with outcomes labelled by $k$, the response functions satisfy
\begin{equation}
    \sum_k \xi^{jk}(\lambda)=1 \qquad \forall \ \lambda , j.
\end{equation}
For an $O \in \{X,Y,Z\}$ measurement this implies
\begin{equation}
    \xi^{O^+}(\lambda)+\xi^{O^-}(\lambda)=1 \qquad \forall \ \lambda .
\end{equation}
If the epistemic states $\mu^{O^+}$ and $\mu^{O^-}$ shared any support, then the ontic state of the system could lie within the overlap; no measurement could conclude which epistemic state it belonged to with certainty. Hence, their supports must be disjoint.

This observation allows us to partition the ontic space $\Lambda$. Define the ontic subspaces which are the supports of each epistemic state, $\Delta_{\pm \sigma}=\text{supp}(\mu^{\sigma^\pm})$ for $\sigma\in\ \mathcal{M}_{\text{1QS}}$.  As every pair of eigenprojectors of a fixed Pauli observable is SSD, the eight ontic subspaces
\begin{equation}
\Lambda_{x,y,z}
= \Delta_{xX}
\cap
\Delta_{yY}
\cap
\Delta_{zZ},
\label{eq:disjoint}
\end{equation}
where $(x,y,z)\in\{\pm\}^3$, must also be disjoint. Each region corresponds to the shared support of a triplet of eigenprojectors of $X$, $Y$, and $Z$.

Given that there are eight disjoint subspaces, we define a model with $s=8$ so that
\begin{flalign}
\label{eq:ontic_points_8state}
    &\lambda_1 = (+,+,+), \lambda_2 = (-,+,+), \lambda_3 = (+,-,+)\notag \\ &\lambda_4 = (+,+,-),
    \lambda_5 = (+,-,-), \lambda_6 = (-,+,-), \notag \\ &\lambda_7 = (-,-,+), \lambda_8 = (-,-,-).
\end{flalign}
Each ontic state $\lambda$ labels whether or not that ontic point is in the support of the corresponding eigenprojectors. For instance, $\lambda_1$ is the unique shared support of the epistemic states $\mu^{X^+},\mu^{Y^+}$ and $\mu^{Z^+}$. Note that, any larger model will essentially inherit the same structure.

We first define response functions and transformation matrices.
Outcome determinism $\langle\xi^{\sigma^\pm},\mu^{\sigma^\pm}\rangle=1$ fixes the response functions' form as
\begin{flalign}
    \xi^{\sigma^\pm}(\lambda)= 
    \begin{cases}
        1   \ \ \ \lambda\in \text{supp}(
        \mu^{\sigma^\pm})\\ 
        0 \ \ \ \text{otherwise}.
    \end{cases}
    \label{eq:8-state-RF}
\end{flalign}
The key observation is that the above specification of ontic states and response functions
completely fixes the representation of transformations $\tau^{(G)}$ as permutation matrices.
For instance, we know that the channel $\tau^{(X)}$ leaves $X$ eigenprojectors unchanged, but flips the sign of $Y$ and $Z$ eigenprojectors due to anti-commutation. The matrix uniquely representing this action is precisely the permutation $\Gamma^{(X)}:(x,y,z)\mapsto(x,-y,-z)$, and similarly for all other channels~\cite{Lillystone2019}.

Now, the epistemic states are invariant under the action of permutation matrices only if they are uniform over their support, that is,
\begin{equation}
    \mu^{X^\pm}(\lambda)=\begin{cases}
        1/4,  \ \ \  \lambda=(x,y,z), x=\pm \\
        0 \qquad \text{ otherwise.}
    \end{cases}
    \label{eq:8-state-ES}
\end{equation}
All extremal epistemic states $\mu^{X^\pm},\mu^{Y^\pm}$ and $\mu^{Z^\pm}$ are thus specified analogously to Eq.~\eqref{eq:8-state-ES}. 

With the ontological model in place, we consider a $\pp\TT^1\MM$ scenario with preparations, transformations and measurements admitting the GPT descriptions given by the sets $\mathcal P_{\text{1QS}}, \mathcal T_{\text{1QS}}$ and $\mathcal M_{\text{1QS}}$ respectively.

The matrices of epistemic states and response functions are $8\times6$ epistemic and  $6\times8$ respectively.  We also flatten the permutation matrices $\Gamma^{(Z)}$, $\Gamma^{(\mathds{1})}$, $\Gamma^{(S)}$ and $\Gamma^{(S^{-1})}$ into length $64$ vectors, which become the rows of the  matrix $G^1$, as defined in Theorem~\ref{thm:PTM-NC-rank-cnstrnts}. We find that
\begin{equation}
    \begin{split}
        &\rank(C_{[\e]})=\rank(R)= 4, \\
        &\rank(C_{[\pp]})=\rank(E)= 4, \\
        &\rank(C_{[\TT^1]})=3,\\ 
        &\rank(G^1)= 4,
    \end{split}
\end{equation}
which, in view of the criteria set out in Theorem~\ref{thm:PTM-NC-rank-cnstrnts}, imply that the ontological model for the $\pp\TT^1\MM$ scenario in this theory is transformation contextual, in agreement with Ref.~\cite{Schmid_2024}.

Indeed, the set of transformations obeys the operational identity
\begin{equation}
    \frac{1}{2}(\tau_Z(\cdot)+\tau_{\mathds{1}}(\cdot))=\frac{1}{2}(\tau_S(\cdot)+\tau_{S^{-1}}(\cdot)),
    \label{eq:stabilizer_opequiv}
\end{equation}
so that transformation noncontextuality demands
\begin{equation}
\label{eq:ZIPhaseNC}
    \frac{1}{2}(\Gamma^{(Z)}+\Gamma^{(\mathds{1})})=\frac{1}{2}(\Gamma^{(S)}+\Gamma^{(S^{-1})}).
\end{equation}
A closer examination of the permutation matrices,
\begin{equation} \label{eq:8-state-trans}
\begin{split}
  \Gamma^{(\mathds{1})}:(x,y,z)&\mapsto(x,y,z), \\
  \Gamma^{(Z)}:(x,y,z)&\mapsto(-x,-y,z), \\ 
  \Gamma^{(S)}:(x,y,z)&\mapsto(y,-x,z), \\
  \Gamma^{(S^{-1})}:(x,y,z)&\mapsto(-y,x,z),
\end{split}
\end{equation}
reveals the source of transformation contextuality. The identity and Z gates are represented by parity-preserving stochastic matrices, whereas the phase gates $S$ and $S^{-1}$ are represented by parity-reversing ones. The channel defined in Eq.~\eqref{eq:stabilizer_opequiv} therefore has two distinct representations as parity-preserving and parity-reversing matrices. These representations are the LHS and RHS of the noncontextuality requirement in Eq.~\eqref{eq:ZIPhaseNC}, respectively, leading to the contradiction.

\subsection{Necessity of Sequential Scenarios}

\label{subsec:necessity}

The example of 8-state theory in the previous section demonstrated that considering transformations in operational theories can introduce constraints beyond the PM scenario that render an NCOM impossible. Given that the set of transformations in the operational theory is assumed to be atomic, we ask whether the contextuality or noncontextuality of the theory can always be determined by analyzing the $\pp\TT^1\MM$ scenario.
We answer this question in the negative by providing a counterexample: an operational theory whose $\pp\TT^1\TT^2\MM$ scenario is contextual, but whose $\pp\TT^1\MM$ scenario is noncontextual. In fact, this conclusion holds true regardless of how one chooses to coarse-grain the $\pp\TT^1\TT^2\MM$ scenario to arrive at the $\pp\TT^1\MM$ scenario.  We thereby prove that contextuality is inseparable from the causal structure of the scenario.

This example is an extension of the $\mathsf{PT}^1\mathsf{M}$ scenario given by Example~\ref{ex:Spekkens_COPE}, which we may refer to as the $\mathsf{S}_1$ scenario. 
The mode-$\TT^1$ slices of the COPE tensor are given by,
\begin{align}
    C^{2D}_{:1:}=\frac{1}{8}
\begin{pmatrix}
1 & 3 & 3 & 1\\
1 & 1 & 3 & 3\\
3 & 1 & 1 & 3\\
3 & 3 & 1 & 1
\end{pmatrix} \label{eq:C_1},\\
C^{2D}_{:2:}=\frac{1}{8}
\begin{pmatrix}
3 & 3 & 1 & 1\\
1 & 3 & 3 & 1\\
1 & 1 & 3 & 3\\
3 & 1 & 1 & 3
\end{pmatrix},\\
C^{2D}_{:3:}=\frac{1}{8}
\begin{pmatrix}
3 & 1 & 1 & 3\\
3 & 3 & 1 & 1\\
1 & 3 & 3 & 1\\
1 & 1 & 3 & 3
\end{pmatrix},\\
C^{2D}_{:4:}=\frac{1}{8}
\begin{pmatrix}
1 & 1 & 3 & 3\\
3 & 1 & 1 & 3\\
3 & 3 & 1 & 1\\
1 & 3 & 3 & 1
\end{pmatrix}.
\end{align}
We first prove that the only possible NCOM, up to a relabelling of ontic points, is given by
\begin{align}
    &R=E=\frac{1}{2}\begin{pmatrix}
        1&1&0&0\\
        0&1&1&0\\
        0&0&1&1\\
        1&0&0&1
    \end{pmatrix} \label{eq:RandE},
\end{align}
and
\begin{align}
    &\Gamma^{11}=\frac{1}{2}\begin{pmatrix} \label{eq:gamma1}
        1&1&0&0\\
        0&1&1&0\\
        0&0&1&1\\
        1&0&0&1
    \end{pmatrix},\\
    &\Gamma^{12}=\frac{1}{2}\begin{pmatrix}
        1&0&0&1\\
        1&1&0&0\\
        0&1&1&0\\
        0&0&1&1
    \end{pmatrix},\\
    &\Gamma^{13}=\frac{1}{2}\begin{pmatrix}
        0&0&1&1\\
        1&0&0&1\\
        1&1&0&0\\
        0&1&1&0
    \end{pmatrix},\\
    &\Gamma^{14}=\frac{1}{2}\begin{pmatrix}
        0&1&1&0\\
        0&0&1&1\\
        1&0&0&1\\
        1&1&0&0
    \end{pmatrix}.
\end{align}
A rank factorization of $C^{2D}_{:1:}$ is given by,
\begin{equation}
\begin{split}
    C^{2D}_{:1:}&=\begin{pmatrix}
        1&0&1\\
        0&-1&1\\
        -1&0&1\\
        0&1&1
    \end{pmatrix}
    \left(\frac{1}{8}\right)\begin{pmatrix}
1&1&-1&-1\\
1&-1&-1&1\\
2&2&2&2
\end{pmatrix}\\
&=AB.
\end{split}
\end{equation}
Defining the polytope $\mathcal{A}$ from the rows of $A$ as $\mathcal{A}=\{x|Ax\geq0, \sum_i (Ax)_i=1\}$, as in Lemma~\ref{lemma:nestedpolytopes}, we find that the columns of $R$ are given by the images of the vertices of $\mathcal{A}$ under $A$. Using Lemma~\ref{lemma:extremalfactor}, this implies that, if a noncontextual model exists, it can be built on $R$ in Eq.~\eqref{eq:RandE}.
Similarly, Lemma~\ref{lemma:extremalfactor2} guarantees that a noncontextual model exists built on $E$ in Eq.~\eqref{eq:RandE}.
With $R$ and $E$ fixed, it is clear from the sparsity patterns of $C$, $R$, and $E$ that noncontextuality fixes the $\Gamma^{1i}$ matrices to the form given above.

Now we consider the $\mathsf{PT}^1\mathsf{T}^2\mathsf{M}$ scenario, referred to as $\mathsf{S}_2$ in this section, in which selecting the $i$th and $j$th transformations from the first and second stage, respectively, gives rise to the transformation slices given by,
\begin{align}
    C^{S_2}_{:ij:}=\frac{1}{8}
\begin{pmatrix}
1 & 3 & 3 & 1\\
1 & 1 & 3 & 3\\
3 & 1 & 1 & 3\\
3 & 3 & 1 & 1
\end{pmatrix},\: i+j\equiv1\pmod{4}, \label{eq:C'_1}\\
C^{S_2}_{:ij:}=\frac{1}{8}
\begin{pmatrix}
3 & 3 & 1 & 1\\
1 & 3 & 3 & 1\\
1 & 1 & 3 & 3\\
3 & 1 & 1 & 3
\end{pmatrix},\: i+j\equiv2\pmod{4},\\
C^{S_2}_{:ij:}=\frac{1}{8}
\begin{pmatrix}
3 & 1 & 1 & 3\\
3 & 3 & 1 & 1\\
1 & 3 & 3 & 1\\
1 & 1 & 3 & 3
\end{pmatrix},\: i+j\equiv3\pmod{4},\\
C^{S_2}_{:ij:}=\frac{1}{8}
\begin{pmatrix}
1 & 1 & 3 & 3\\
3 & 1 & 1 & 3\\
3 & 3 & 1 & 1\\
1 & 3 & 3 & 1
\end{pmatrix},\: i+j=0\pmod{4}.
\end{align}
We see that upon combining the two transformation stages, the statistics coincide with those of scenario $\mathsf{S}_1$, except that each cyclic shift of the transformations is repeated four times.
Therefore, we can use the same $R$ and $E$ as in Eq.~\eqref{eq:RandE}
and this is still the most general choice.

Furthermore, upon composing $\mathsf{T}^2$  with measurements, the cyclic shifts of measurement outcomes are repeated four times each.
The $\Gamma^{1i}$ matrices in Eq.~\eqref{eq:gamma1} are therefore the only noncontextual representation of $\mathsf{T}^1$ for $\mathsf{S}_2$.

Finally, applying this argument to the composition of preparations and $\mathsf{T}^1$ transformations,
we also see that the $\Gamma^{1i}$ matrices of Eq.~\eqref{eq:gamma1} are the only noncontextual representation of $\mathsf{T}^2$.
However, if both $\mathsf{T}^1$ and $\mathsf{T}^2$ are represented with the $\Gamma^{1i}$ matrices, they no longer reproduce the statistics.
Note that here, we did not assume a priori that the symmetry of scenario $\mathsf{S}_2$ requires both transformation stages to have identical representation.

\subsection{Restrictions on Procedures}

\label{sec:temporal}

The previous examples highlighted the interplay between the causal structure of an operational theory and the presence of contextuality. Here, we show that modifying the set of atomic transformation procedures at each stage also plays a nontrivial role. To do so, we consider two operational theories similar to the 8-state theory defined in Sec.~\ref{sec:8state}.

The first new operational theory admits a GPT model which uses the same quantum operator representations used by the 8-state model for measurements and preparations, namely the Pauli measurements $\mathcal{M}_{\text{1QS}}=\{X,Y,Z\}$ and their $\pm1$ eigenprojectors $\mathcal{P}_{\text{1QS}}=\{X^\pm,Y^\pm,Z^\pm\}$. The set of atomic transformations in the theory is $\mathcal{T}=\{\TT_{\mathds{1}}, \TT_{X}, \TT_{Y}, \TT_{Z}, \TT_{H}\}$. These admit a GPT model given by the quantum channels $\{\tau_{\mathds{1}},\tau_X,\tau_Y,\tau_Z,\tau_H\}$, where $\tau_G=G(\cdot)G^\dagger$.

We consider the $\pp\TT^1\TT^2\MM$ scenario in which the two transformation stages are given by $\mathcal{T}_1=\mathcal{T}_2=\mathcal{T}$. An ontological model for this scenario can be constructed by extending the model described in Sec.~\ref{sec:8state}, with 8 ontic states labelled according to the convention in Eq.~\eqref{eq:ontic_points_8state}. In particular, outcome determinism allows us to partition the ontic space into the same 8 subspaces defined in Eq.~\eqref{eq:disjoint}, and forces the response functions to have the form in Eq.~\eqref{eq:8-state-RF}. By the same argument given in Sec.~\ref{sec:8state}, the transformation matrices representing $\mathcal{T}$ must be permutation matrices, and subsequently the epistemic states must be given by Eq.~\eqref{eq:8-state-ES}. 

The matrices representing the identity and $Z$ gates are characterized in the first two lines of Eq.~\eqref{eq:8-state-trans}, and the remaining matrices are such that
\begin{equation}
\label{eq:XYH_8state}
    \begin{split}
        \Gamma^{(X)}&:(x,y,z)\mapsto(x,-y,-z) , \\
        \Gamma^{(Y)}&:(x,y,z)\mapsto(-x,y,-z) ,\\
        \Gamma^{(H)}&:(x,y,z)\mapsto(z,-y,x) .
    \end{split}
\end{equation}

By forming the COPE 4-tensor of the $\pp\TT^1\TT^2\MM$ scenario in this operational theory, we find:
\begin{equation}
\begin{split}
    &\rank(C_{[\e]})=\rank(R)= 4, \\
    &\rank(C_{[\pp]})=\rank(E)= 4, \\
    &\rank(C_{[\TT^1]})=\rank(G^1)=5, \\ &\rank(C_{[\TT^2]})=\rank(G^2)= 5.
    \label{eq:ranks_scenario2}
\end{split}
\end{equation}
In other words, our rank criteria certify this model as an NCOM of the statistics. 
Noting that the theory includes five transformations, the last two equalities certify that the model is transformation noncontextual because the flattenings $C_{[\TT^1]}$ and $C_{[\TT^2]}$, and the corresponding matrix factors $G^1$ and $G^2$, are all full rank. This means that there are no nontrivial linear dependencies among the transformations, either at the level of their statistics or at the level of their representations.

Recall that the noncontextuality conditions in Eqs.~\eqref{eq:2stageNCOMconditions} only capture operational equivalences between atomic procedures, or convex combinations thereof, but they do not capture operational equivalences between composite procedures. The operational theory defined above has no nontrivial linear dependencies between atomic procedures, but there are nontrivial dependencies between composite procedures.

In particular, consider the completely depolarizing channel given by $\tau_D(\rho)=\mathds{1}/2 $ for all $\rho$, which can be written as the convex combination
\begin{equation}
    \begin{split}
        \tau_D(\cdot) &= \frac{1}{4}\left[\mathds{1}(\cdot)\mathds{1}+X(\cdot)X^\dagger+Y(\cdot)Y^\dagger+Z(\cdot)Z^\dagger\right], \\
        & = \frac{1}{4}\left[\tau_\mathds{1}(\cdot)+\tau_X(\cdot)+\tau_Y(\cdot)+\tau_Z(\cdot)\right].
    \end{split}
\end{equation}
The ontological representation must be the corresponding convex mixture of Paulis, that is,
\begin{equation}
\Gamma^{(D)}=\frac{1}{4}(\Gamma^{(\mathds{1})}+\Gamma^{(X)}+\Gamma^{(Y)}+\Gamma^{(Z)}).
\end{equation}
Since ontological representations of Paulis preserve parity, so does $\Gamma^{(D)}$.
Alternatively, we also have,
\begin{equation}
\label{eq:depolaris_opidentity}
    \tau_{D'}(\cdot):= \tau_H\circ\tau_D(\cdot) = \tau_D(\cdot)
\end{equation}
which implies
\begin{equation}
    \TT_{D'}(\cdot)\simeq\TT_{D}(\cdot).
    \label{eq:2stagedepolaris}
\end{equation}
However, the corresponding ontological representation of $\TT_{D'}(\cdot)$, $\Gamma^{(D')}$, is parity-reversing because the representation of the Hadamard given in Eq.~\eqref{eq:XYH_8state} is parity-reversing. As a result, $\Gamma^{(D')}\neq  \Gamma^{(D)}$~\cite{Lillystone2019}. 

Guided by this observation, we consider the $\pp\TT^1\TT^2\MM$ scenario in an adjacent operational theory, in which the transformations in each stage are denoted $\mathcal{T}_1'=\mathcal{T}_2'=\set{T_\mathds{1},T_X,T_Y,T_Z,T_H,T_{HX},T_{HY},T_{HZ}}$. The GPT model provided by quantum operators still applies, and the new transformations are simply represented as channel compositions. For instance, $\TT_{HX}$ is represented in the GPT model by the channel $\tau_{HX}(\cdot):=\tau_H(\tau_X(\cdot))$. This theory therefore contains the same procedures as the previous theory, but also includes atomic transformation procedures which the previous theory considered composite.  

Following the same argument we used to construct the ontological model of the previous theory, the transformation matrices are fixed as permutation matrices. The new transformation matrices can be expressed in terms of those in Eqs.~\eqref{eq:8-state-trans} and~\eqref{eq:XYH_8state}, as $\Gamma^{(HX)}=\Gamma^{(H)}\Gamma^{(X)}$, $\Gamma^{(HY)}=\Gamma^{(H)}\Gamma^{(Y)}$ and $\Gamma^{(HZ)}=\Gamma^{(H)}\Gamma^{(Z)}$. Placing the full set of row-flattened transformation matrices into matrices $G'^{1}$ and $G'^{2}$, we find that
\begin{equation}
\label{eq:lasttheory_ranksep}
    \begin{split}
        &\rank(C_{[\e]})=\rank(R)= 4, \\
        &\rank(C_{[\pp]})=\rank(E)= 4, \\
        &\rank(C_{[\TT'^1]})=\rank(C_{[\TT'^2]})=6, \\
        &\rank(G'^1)=\rank(G'^2) =8.
    \end{split}
\end{equation}
The last two equalities imply that, in each stage of transformations, there are operational equivalences which are not respected at the ontological level. One such equivalence is given in Eq.~\eqref{eq:2stagedepolaris}, which is implied by the operational identity in Eq.~\eqref{eq:depolaris_opidentity}. We identify a second identity,
\begin{equation*}
    \frac{1}{2}(HX\rho X^\dagger H^\dagger+HZ\rho Z^\dagger H^\dagger ) = \frac{1}{2}(\mathds{1}\rho \mathds{1} +Y\rho Y^\dagger),
\end{equation*}
    which implies the equivalence
\begin{equation}
    \frac{1}{2}[\TT_{HX}+\TT_{HZ} ]\simeq\frac{1}{2}[\TT_{\mathds{1}}+\TT_Y].
    \label{eq:HXHZisIY}
\end{equation} 
This equivalence is also not respected by the ontological model, as $\Gamma^{(HX)}$ and $\Gamma^{(HZ)}$ are parity-reversing, where $\Gamma^{(I)}$ and $\Gamma^{(Y)}$ are parity-preserving.

The previous example demonstrates that the contextuality/noncontextuality of a given model depends nontrivially on which procedures are assumed to be atomic.
In the first operational theory where $\T_1=\T_2=\{T_{\mathds{1}},T_X,T_Y,T_Z,T_H\}$, transformations such as $T_{HX}$ can only be implemented over two stages, so that they are secondary propositional schema. Hence, the operational equivalences~\eqref{eq:2stagedepolaris} and \eqref{eq:HXHZisIY} are not expressible in terms of atomic procedures alone, and so there is no requirement that an NCOM should respect them. Conversely, in the second theory, the transformations $T_{HX},T_{HY}$ and $T_{HZ}$ are considered atomic and are available in each stage by definition. This introduces nontrivial convex dependencies between the statistics of transformations in each stage, and the failure of the model to represent these equivalences noncontextually manifests in the last two lines of Eq.~\eqref{eq:lasttheory_ranksep}.

\section{Conclusion}

\label{sec:conclusion}

In this work, we have extended the COPE-matrix formalism for PM scenarios of operational theories introduced in Ref.~\cite{shahandeh2025unifiedlinearalgebraicframework}
to a COPE tensor formalism applicable to scenarios involving transformations. The COPE tensor contains the statistics associated with all possible sequences of procedures in a scenario specified in a given operational theory.
The goal was to devise mathematical models for the statistical structure that comply with the causal structure, namely, the scenario, and possibly satisfy further constraints such as unique representations and nonnegativity.
We then discussed three such models, namely, preGPT, GPT, and ontological models.
We discussed how each model arises as a particular factorization of the COPE tensor.
In particular, preGPTs are factorizations that do not follow any restrictions besides the temporal structure of the scenario. We then showed that a GPT is a preGPT which is required to assign a unique representation to each equivalence class of statistically indistinguishable operational procedure. 
We formulated this requirement as a set of rank constraints on the COPE tensor decomposition.
Then, by establishing the equivalence of ontological models with nonnegative preGPTs, we showed that NCOMs correspond to nonnegative factorizations satisfying a set of rank equalities similar to GPTs. Conversely, the impossibility of such a factorization certifies the contextuality of the operational theory given the required causal structure.

In this work, we applied the above modelling to temporal PTM scenarios. For such scenarios, we explicitly constructed the smallest possible GPT.
Here, we considered an arbitrary number of transformations and derived necessary and sufficient equirank criteria for their (non)contextuality within a given operational theory.

Taking a geometric perspective on rank-restricted factorizations led us to an algorithm for deciding whether a noncontextual model exists in single-stage transformation scenarios, and for constructing one when it does. 
The algorithm is polynomial time in the number of procedures, but its complexity is linearly exponential in the minimum GPT dimension required to model the COPE tensor.
A suitable iteration of this algorithm was shown to apply to scenarios with sequential transformations, and the accompanying complexity analysis suggests that the method is almost optimal. Specifically, the complexity of this algorithm for arbitrarily many stages of transformations still has complexity linearly exponential in the GPT dimension, and polynomial in the number of procedures. Therefore, this complexity scaling is similar to the best known complexity scaling for PM scenarios~\cite{Selby_2024,yianni2025complexitycontextuality}.

We demonstrated our equirank criteria on several known toy theories. 
We employed a 2D toy theory as a running example to showcase our approach and model concepts.
In particular, the 2D toy theory has the property that multi-stage scenarios exhibit contextuality, whereas single-stage scenarios are noncontextual.
It thus demonstrates that sequential structure is a genuine source of contextual behavior.
We showed the noncontextuality of Spekkens' toy theory with an arbitrary number of transformation stages and analyzed the transformation contextuality of the 8-state single-qubit stabilizer subtheory.
Then, we examined the effect of restricting the allowed procedures. By modifying the 8-state single-qubit stabilizer subtheory, we devised two operational theories whose only point of difference was that the composite transformations in the first theory were designated as atomic in the second theory: the former scenario admitted an NCOM, but the same model, applied to the latter scenario, failed to provide a noncontextual representation.

This work lays down the groundwork for characterizing the role of multi-stage contextuality as a resource for information processing, since many models of computation are intrinsically sequential. To this end, it is reasonable to study measures of contextuality, their computational complexity, and behaviour in various scenarios using the COPE framework.
Of greater practical significance is the development of robust versions of our criteria that tolerate approximate operational equivalences, allowing one to certify (non)contextuality under experimentally realistic noise and finite statistics.
Finally, an interesting direction would be to characterize the contextuality of multi-stage scenarios with intermediate measurements. This would lead to broader contextuality conditions, which are relevant to models of computation and multi-stage communication protocols.

\acknowledgements
TY acknowledges the support through the Quantum Computing Studentship funded by Royal Holloway, University of London.
FS gratefully acknowledges the financial support from the Engineering and Physical Sciences Research Council (EPSRC) through the Hub in Quantum Computing and Simulation grant (EP/T001062/1). NR is grateful to be supported by the EPSRC Quantum Technologies Doctoral Training Partnership [grant number EP/W524311/1].
\bibliography{reffinal.bib}

@misc{spekkens2019ontologicalidentityempiricalindiscernibles,
      title={The ontological identity of empirical indiscernibles: Leibniz's methodological principle and its significance in the work of Einstein}, 
      author={Robert W. Spekkens},
      year={2019},
      eprint={1909.04628},
      archivePrefix={arXiv},
      primaryClass={physics.hist-ph},
      url={https://arxiv.org/abs/1909.04628}, 
}

@Article{Khachiyan2008,
author={Khachiyan, Leonid
and Boros, Endre
and Borys, Konrad
and Elbassioni, Khaled
and Gurvich, Vladimir},
title={Generating All Vertices of a Polyhedron Is Hard},
journal={Discrete {\&} Computational Geometry},
year={2008},
month={Mar},
day={01},
volume={39},
number={1},
pages={174-190},
issn={1432-0444},
url={https://doi.org/10.1007/s00454-008-9050-5}
}

@article{Spekkens_2005,
   title={Contextuality for preparations, transformations, and unsharp measurements},
   volume={71},
   ISSN={1094-1622},
   url={http://dx.doi.org/10.1103/PhysRevA.71.052108},
   number={5},
   journal={Physical Review A},
   publisher={American Physical Society (APS)},
   author={Spekkens, R. W.},
   year={2005},
   month=may }

@misc{ourpaper,
      title={Characterizing Contextuality via Rank Separation with Applications to Cloning}, 
      author={Farid Shahandeh and Theodoros Yianni and Mina Doosti},
      year={2024},
      eprint={2406.19382},
      archivePrefix={arXiv},
      primaryClass={quant-ph},
      url={https://arxiv.org/abs/2406.19382}, 
}

@article{Selby_2024,
   title={Linear Program for Testing Nonclassicality and an Open-Source Implementation},
   volume={132},
   ISSN={1079-7114},
   url={http://dx.doi.org/10.1103/PhysRevLett.132.050202},
   number={5},
   journal={Physical Review Letters},
   publisher={American Physical Society (APS)},
   author={Selby, John H. and Wolfe, Elie and Schmid, David and Sainz, Ana Belén and Rossi, Vinicius P.},
   year={2024},
   month=jan }

@ARTICLE{Koc,
    author = "Simon Kochen, E. Specker",
     title = "The Problem of Hidden Variables in Quantum Mechanics",
   journal = "Indiana Univ. Math. J.",
  fjournal = "Indiana University Mathematics Journal",
    volume = 17,
      year = 1968,
     issue = 1,
     pages = "59--87",
      issn = "0022-2518",
     coden = "IUMJAB",
   mrclass = "",
}

@article{computation,
  title = {Contextuality as a Resource for Models of Quantum Computation with Qubits},
  author = {Bermejo-Vega, Juan and Delfosse, Nicolas and Browne, Dan E. and Okay, Cihan and Raussendorf, Robert},
  journal = {Phys. Rev. Lett.},
  volume = {119},
  issue = {12},
  pages = {120505},
  numpages = {5},
  year = {2017},
  month = {Sep},
  publisher = {American Physical Society},
  url = {https://link.aps.org/doi/10.1103/PhysRevLett.119.120505}
}

@misc{shahandeh2021advantage,
      title={Quantum computational advantage implies contextuality}, 
      author={Farid Shahandeh},
      year={2021},
      eprint={2112.00024},
      archivePrefix={arXiv},
      primaryClass={quant-ph},
      url={https://arxiv.org/abs/2112.00024}, 
}

@article{Schmid2022Stabilizer,
  title = {Uniqueness of Noncontextual Models for Stabilizer Subtheories},
  author = {Schmid, David and Du, Haoxing and Selby, John H. and Pusey, Matthew F.},
  journal = {Phys. Rev. Lett.},
  volume = {129},
  issue = {12},
  pages = {120403},
  numpages = {6},
  year = {2022},
  month = {Sep},
  publisher = {American Physical Society},
  url = {https://link.aps.org/doi/10.1103/PhysRevLett.129.120403}
}

@article{Gitton_2022,
   title={Solvable Criterion for the Contextuality of any Prepare-and-Measure Scenario},
   volume={6},
   ISSN={2521-327X},
   url={http://dx.doi.org/10.22331/q-2022-06-07-732},
   journal={Quantum},
   publisher={Verein zur Forderung des Open Access Publizierens in den Quantenwissenschaften},
   author={Gitton, Victor and Woods, Mischa P.},
   year={2022},
   month=jun, pages={732} }

@misc{Schmid_2024,
      title={Noncontextuality inequalities for prepare-transform-measure scenarios}, 
      author={David Schmid and Roberto D. Baldijão and John H. Selby and Ana Belén Sainz and Robert W. Spekkens},
      year={2024},
      eprint={2407.09624},
      archivePrefix={arXiv},
      primaryClass={quant-ph},
      url={https://arxiv.org/abs/2407.09624}, 
}

@article{NCOM-OPTs,
   author = {Sina Soltani and Marco Erba and David Schmid and John H. Selby},
   month = {2},
   title = {Noncontextual ontological models of operational probabilistic theories},
   url = {https://arxiv.org/pdf/2502.11842},
   year = {2025}
}

@article{Schmid2024,
   author = {David Schmid and John H. Selby and Matthew F. Pusey and Robert W. Spekkens},
   doi = {10.22331/q-2024-03-14-1283},
   issn = {2521327X},
   journal = {Quantum},
   month = {3},
   pages = {1283},
   publisher = {Verein zur Förderung des Open Access Publizierens in den Quantenwissenschaften},
   title = {A structure theorem for generalized-noncontextual ontological models},
   volume = {8},
   url = {https://quantum-journal.org/papers/q-2024-03-14-1283/},
   year = {2024}
}

@article{GILLIS20122685,
title = {On the geometric interpretation of the nonnegative rank},
journal = {Linear Algebra and its Applications},
volume = {437},
number = {11},
pages = {2685-2712},
year = {2012},
issn = {0024-3795},
doi = {https://doi.org/10.1016/j.laa.2012.06.038},
url = {https://www.sciencedirect.com/science/article/pii/S0024379512005022},
author = {Nicolas Gillis and François Glineur},
}

@article{Lillystone2019,
   author = {Piers Lillystone and Joel J. Wallman and Joseph Emerson},
   doi = {10.1103/PhysRevLett.122.140405},
   issn = {10797114},
   issue = {14},
   journal = {Physical Review Letters},
   month = {4},
   pages = {140405},
   pmid = {31050485},
   publisher = {American Physical Society},
   title = {Contextuality and the Single-Qubit Stabilizer Subtheory},
   volume = {122},
   url = {https://journals.aps.org/prl/abstract/10.1103/PhysRevLett.122.140405},
   year = {2019}
}

@article{SpekkensToy,
   author = {Robert W. Spekkens},
   doi = {10.1103/PhysRevA.75.032110},
   issue = {3},
   journal = {Physical Review A - Atomic, Molecular, and Optical Physics},
   month = {10},
   title = {In defense of the epistemic view of quantum states: a toy theory},
   volume = {75},
   url = {http://arxiv.org/abs/quant-ph/0401052 http://dx.doi.org/10.1103/PhysRevA.75.032110},
   year = {2005}
}

@article{SpekkensToyReview,
   author = {Ladina Hausmann and Nuriya Nurgalieva and Lídia del Rio},
   month = {5},
   title = {A consolidating review of Spekkens' toy theory},
   url = {https://arxiv.org/pdf/2105.03277},
   year = {2021}
}

@article{Bell1964,
  title = {On the Einstein Podolsky Rosen paradox},
  author = {Bell, J. S.},
  journal = {Physics Physique Fizika},
  volume = {1},
  issue = {3},
  pages = {195--200},
  numpages = {6},
  year = {1964},
  month = {Nov},
  publisher = {American Physical Society},
  doi = {10.1103/PhysicsPhysiqueFizika.1.195},
  url = {https://link.aps.org/doi/10.1103/PhysicsPhysiqueFizika.1.195}
}

@misc{shahandeh2025unifiedlinearalgebraicframework,
      title={A Unified Linear Algebraic Framework for Physical Models and Generalized Contextuality}, 
      author={Farid Shahandeh and Theodoros Yianni and Mina Doosti},
      year={2025},
      eprint={2512.10000},
      archivePrefix={arXiv},
      primaryClass={quant-ph},
      url={https://arxiv.org/abs/2512.10000}, 
}

@misc{yianni2025complexitycontextuality,
      title={Complexity of Contextuality}, 
      author={Theodoros Yianni and Farid Shahandeh},
      year={2025},
      eprint={2506.09133},
      archivePrefix={arXiv},
      primaryClass={quant-ph},
      url={https://arxiv.org/abs/2506.09133}, 
}

@article{Howard_2014,
   title={Contextuality supplies the ‘magic’ for quantum computation},
   volume={510},
   ISSN={1476-4687},
   url={http://dx.doi.org/10.1038/nature13460},
   DOI={10.1038/nature13460},
   number={7505},
   journal={Nature},
   publisher={Springer Science and Business Media LLC},
   author={Howard, Mark and Wallman, Joel and Veitch, Victor and Emerson, Joseph},
   year={2014},
   month=jun, pages={351–355} }

@article{PhysRevLett.121.230401,
  title = {Quantum Advantage from Sequential-Transformation Contextuality},
  author = {Mansfield, Shane and Kashefi, Elham},
  journal = {Phys. Rev. Lett.},
  volume = {121},
  issue = {23},
  pages = {230401},
  numpages = {6},
  year = {2018},
  month = {Dec},
  publisher = {American Physical Society},
  doi = {10.1103/PhysRevLett.121.230401},
  url = {https://link.aps.org/doi/10.1103/PhysRevLett.121.230401}
}

@article{Kunjwal2015,
  title = {From the Kochen-Specker Theorem to Noncontextuality Inequalities without Assuming Determinism},
  author = {Kunjwal, Ravi and Spekkens, Robert W.},
  journal = {Phys. Rev. Lett.},
  volume = {115},
  issue = {11},
  pages = {110403},
  numpages = {5},
  year = {2015},
  month = {Sep},
  publisher = {American Physical Society},
  doi = {10.1103/PhysRevLett.115.110403},
  url = {https://link.aps.org/doi/10.1103/PhysRevLett.115.110403}
}

@article{Abramsky_2011,
doi = {10.1088/1367-2630/13/11/113036},
url = {https://doi.org/10.1088/1367-2630/13/11/113036},
year = {2011},
month = {nov},
publisher = {IOP Publishing},
volume = {13},
number = {11},
pages = {113036},
author = {Abramsky, Samson and Brandenburger, Adam},
title = {The sheaf-theoretic structure of non-locality and contextuality},
journal = {New Journal of Physics}
}

@article{DzhafarovKujala2016,
  title   = {Context--content systems of random variables: The Contextuality-by-Default theory},
  author  = {Dzhafarov, Ehtibar N. and Kujala, Janne V.},
  journal = {Journal of Mathematical Psychology},
  volume  = {74},
  pages   = {11--33},
  year    = {2016},
  month   = oct,
  doi     = {10.1016/j.jmp.2016.04.010},
  url     = {https://doi.org/10.1016/j.jmp.2016.04.010}
}

@article{Budroni2022,
  title = {Kochen-Specker contextuality},
  author = {Budroni, Costantino and Cabello, Ad\'an and G\"uhne, Otfried and Kleinmann, Matthias and Larsson, Jan-\AA{}ke},
  journal = {Rev. Mod. Phys.},
  volume = {94},
  issue = {4},
  pages = {045007},
  numpages = {62},
  year = {2022},
  month = {Dec},
  publisher = {American Physical Society},
  doi = {10.1103/RevModPhys.94.045007},
  url = {https://link.aps.org/doi/10.1103/RevModPhys.94.045007}
}

@article{Bell1966,
  title   = {On the Problem of Hidden Variables in Quantum Mechanics},
  author  = {Bell, John S.},
  journal = {Reviews of Modern Physics},
  volume  = {38},
  number  = {3},
  pages   = {447--452},
  year    = {1966},
  month   = jul,
  doi     = {10.1103/RevModPhys.38.447},
  url     = {https://doi.org/10.1103/RevModPhys.38.447}
}

@article{CabelloSeveriniWinter2014,
  title = {Graph-Theoretic Approach to Quantum Correlations},
  author = {Cabello, Ad\'an and Severini, Simone and Winter, Andreas},
  journal = {Phys. Rev. Lett.},
  volume = {112},
  issue = {4},
  pages = {040401},
  numpages = {5},
  year = {2014},
  month = {Jan},
  publisher = {American Physical Society},
  doi = {10.1103/PhysRevLett.112.040401},
  url = {https://link.aps.org/doi/10.1103/PhysRevLett.112.040401}
}

@article{Cabello2008,
  title = {Experimentally Testable State-Independent Quantum Contextuality},
  author = {Cabello, Ad\'an},
  journal = {Phys. Rev. Lett.},
  volume = {101},
  issue = {21},
  pages = {210401},
  numpages = {4},
  year = {2008},
  month = {Nov},
  publisher = {American Physical Society},
  doi = {10.1103/PhysRevLett.101.210401},
  url = {https://link.aps.org/doi/10.1103/PhysRevLett.101.210401}
}

@article{Saha_2019,
  title = {Preparation contextuality as an essential feature underlying quantum communication advantage},
  author = {Saha, Debashis and Chaturvedi, Anubhav},
  journal = {Phys. Rev. A},
  volume = {100},
  issue = {2},
  pages = {022108},
  numpages = {13},
  year = {2019},
  month = {Aug},
  publisher = {American Physical Society},
  doi = {10.1103/PhysRevA.100.022108},
  url = {https://link.aps.org/doi/10.1103/PhysRevA.100.022108}
}

@article{Hameedi2017,
   author = {Alley Hameedi and Armin Tavakoli and Breno Marques and Mohamed Bourennane},
   doi = {10.1103/PhysRevLett.119.220402},
   issue = {22},
   journal = {Physical Review Letters},
   month = {11},
   publisher = {American Physical Society},
   title = {Communication games reveal preparation contextuality},
   volume = {119},
   url = {http://arxiv.org/abs/1704.08223 http://dx.doi.org/10.1103/PhysRevLett.119.220402},
   year = {2017}
}

@article{Roy2024,
   author = {Prabuddha Roy and A. K. Pan},
   doi = {10.1088/1751-8121/ad7108},
   issn = {17518121},
   issue = {37},
   journal = {Journal of Physics A: Mathematical and Theoretical},
   month = {8},
   pages = {375303},
   publisher = {IOP Publishing},
   title = {Generalized parity-oblivious communication games powered by quantum preparation contextuality},
   volume = {57},
   url = {https://iopscience.iop.org/article/10.1088/1751-8121/ad7108 https://iopscience.iop.org/article/10.1088/1751-8121/ad7108/meta},
   year = {2024}
}

@article{Ambainis2016,
   author = {Andris Ambainis and Manik Banik and Anubhav Chaturvedi and Dmitry Kravchenko and Ashutosh Rai},
   month = {7},
   title = {Parity Oblivious d-Level Random Access Codes and Class of Noncontextuality Inequalities},
   url = {https://arxiv.org/pdf/1607.05490},
   year = {2016}
}

@article{HOSVD1,
   author = {M. Alex O. Vasilescu and Demetri Terzopoulos},
   doi = {10.1007/3-540-47969-4_30/SAVE-RESEARCH},
   isbn = {3540437452},
   issn = {16113349},
   journal = {Lecture Notes in Computer Science (including subseries Lecture Notes in Artificial Intelligence and Lecture Notes in Bioinformatics)},
   pages = {447-460},
   publisher = {Springer Verlag},
   title = {Multilinear analysis of image ensembles: Tensorfaces},
   volume = {2350},
   url = {https://link.springer.com/chapter/10.1007/3-540-47969-4_30},
   year = {2002}
}

@article{HOSVD2,
   author = {M. Alex O. Vasilescu and Demetri Terzopoulos},
   doi = {10.1109/CVPR.2005.240},
   isbn = {0769523722},
   journal = {Proceedings - 2005 IEEE Computer Society Conference on Computer Vision and Pattern Recognition, CVPR 2005},
   pages = {547-553},
   publisher = {IEEE Computer Society},
   title = {Multilinear independent components analysis},
   volume = {I},
   url = {https://ieeexplore.ieee.org/document/1467315},
   year = {2005}
}

@article{tensor-train,
   author = {I. V. Oseledets},
   doi = {10.1137/090752286},
   issn = {10648275},
   issue = {5},
   journal = {https://doi.org/10.1137/090752286},
   month = {9},
   pages = {2295-2317},
   publisher = {Society for Industrial and Applied Mathematics},
   title = {Tensor-Train Decomposition},
   volume = {33},
   url = {/doi/pdf/10.1137/090752286?download=true},
   year = {2011}
}
\bibliographystyle{apsrev4-1}

\clearpage
\onecolumngrid
\appendix
\section{Flattenings and Factorizations}
\setlength\parindent{0pt}
\label{sec:appA}
In this Appendix we list some examples of COPE tensor flattenings as well as their factorizations, which can be interpreted as an ontological model. These pertain to the $\pp\TT^1\MM$ scenario. We use $\xi^i,\mu^i$ and $\Gamma^i$ to denote the $i$th response function, epistemic state and transformation matrix respectively. The sets of effects, preparations and transformations are given by $\mathcal{E},\mathcal{P}$ and $\mathcal{T}$.
\begin{flalign*}
    C_{[\e]} &= \begin{pmatrix}
        \xi^1(\lambda_1) & \xi^1(\lambda_2) & \dots & \xi^1(\lambda_s) \\
        \xi^2(\lambda_1) & \xi^2(\lambda_2) & \dots & \xi^2(\lambda_s) \\ \vdots & \vdots & \dots & \vdots \\
        \xi^{|\mathcal{E}|}(\lambda_1) & \xi^{|\mathcal{E}|}(\lambda_2) & \dots & \xi^{|\mathcal{E}|}(\lambda_s)
    \end{pmatrix}\times\\
    & \sum_i
        \begin{pmatrix}
        \Gamma^1(\lambda_1|\lambda_i)\mu^1(\lambda_i) & \Gamma^2(\lambda_1|\lambda_i)\mu^1(\lambda_i) & \cdots & \Gamma^{|\mathcal{T}|}(\lambda_1|\lambda_i)\mu^1(\lambda_i)& \Gamma^1(\lambda_1|\lambda_i)\mu^2(\lambda_i) & \cdots\\ \Gamma^1(\lambda_2|\lambda_i)\mu^1(\lambda_i) & \Gamma^2(\lambda_2|\lambda_i)\mu^1(\lambda_i) & \cdots & \Gamma^{|\mathcal{T}|}(\lambda_2|\lambda_i)\mu^1(\lambda_i)& \Gamma^1(\lambda_2|\lambda_i)\mu^2(\lambda_i) & \cdots
        \ \\
        \vdots  & \vdots & \cdots & \vdots & \vdots \\  \Gamma^1(\lambda_{s-1}|\lambda_i)\mu^1(\lambda_i) & \Gamma^2(\lambda_{s-1}|\lambda_i)\mu^1(\lambda_i) & \cdots & \Gamma^{|\mathcal{T}|}(\lambda_{s-1}|\lambda_i)\mu^1(\lambda_i)& \Gamma^1(\lambda_{s-1}|\lambda_i)\mu^2(\lambda_i) & \cdots\\ \Gamma^1(\lambda_s|\lambda_i)\mu^1(\lambda_i) & \Gamma^2(\lambda_s|\lambda_i)\mu^1(\lambda_i) & \cdots & \Gamma^{|\mathcal{T}|}(\lambda_s|\lambda_i)\mu^1(\lambda_i)& \Gamma^1(\lambda_s|\lambda_i)\mu^2(\lambda_i) & \cdots
    \end{pmatrix}
\end{flalign*}
That is,
\begin{equation*}
    C_{[\e]}=RE'
\end{equation*}
where $R$ contains the response functions, and $E'$ contains epistemic states describing the augmented set of preparations, produced by applying every transformation in $\mathcal{T}$ to every preparation in the original set $\mathcal{P}$. The entries of $C_{[\e]}$ are simply the probabilities of the events described by $R$ occurring in each of the preparations in the augmented set. Next, we have

\begin{equation*}
    C_{[\pp]} = 
     \sum_i
        \begin{pmatrix}
        \xi^1(\lambda_i)\Gamma^1(\lambda_i|\lambda_1) & \xi^1(\lambda_i)\Gamma^1(\lambda_i|\lambda_2) & \cdots & \xi^1(\lambda_i)\Gamma^1(\lambda_i|\lambda_s)\\ \xi^1(\lambda_i)\Gamma^2(\lambda_i|\lambda_1) & \xi^1(\lambda_i)\Gamma^2(\lambda_i|\lambda_2) & \cdots  &\xi^1(\lambda_i)\Gamma^2(\lambda_i|\lambda_s)
        \ \\
        \vdots  & \vdots & \cdots & \vdots & \vdots \\  \xi^1(\lambda_i)\Gamma^{|\mathcal{T}|}(\lambda_i|\lambda_1) & \xi^1(\lambda_i)\Gamma^{|\mathcal{T}|}(\lambda_i|\lambda_2) & \cdots & \xi^1(\lambda_i)\Gamma^{|\mathcal{T}|}(\lambda_i|\lambda_s)\\ \xi^2(\lambda_i)\Gamma^1(\lambda_i|\lambda_1) & \xi^2(\lambda_i)\Gamma^1(\lambda_i|\lambda_2) & \cdots & \xi^2(\lambda_i)\Gamma^1(\lambda_i|\lambda_s) \\
        \vdots & \vdots & \cdots & \vdots
    \end{pmatrix}
\end{equation*}
\begin{equation*}
     \hspace{2cm}\times\begin{pmatrix}
        \mu^1(\lambda_1) & \mu^2(\lambda_1) & \dots & \mu^{|\mathcal{P}|}(\lambda_1) \\
        \mu^1(\lambda_2) & \mu^2(\lambda_2) & \dots & \mu^{|\mathcal{P}|}(\lambda_2) \\ \vdots & \vdots & \dots & \vdots \\
        \mu^1(\lambda_s) & \mu^2(\lambda_s) & \dots & \mu^{|\mathcal{P}|}(\lambda_s)
    \end{pmatrix}
\end{equation*}
That is,
\begin{equation*}
    C_{[\pp]}=R'E
\end{equation*}
where $R'$ describes the augmented set of effects, calculated by applying every transformation \textit{before} every effect in the original set $\mathcal{E}$. The entries of $C_{[\pp]}$ are simply the probabilities of the events described in this augmented set occurring in any of the original preparations in $\mathcal P$. Finally, 

\begin{flalign*}
&C_{[\TT^1]}=
    \begin{pmatrix}
        \Gamma^1(\lambda_1|\lambda_1) & \cdots & \Gamma^1(\lambda_1|\lambda_s) &  \Gamma^1(\lambda_2|\lambda_1) &\cdots \\
        \Gamma^2(\lambda_1|\lambda_1) & \cdots & \Gamma^2(\lambda_1|\lambda_s) &  \Gamma^2(\lambda_2|\lambda_1) &\cdots \\
        \vdots & & \vdots & \vdots
    \end{pmatrix} \label{eq:T-matrix}
    \\
    &\times \begin{pmatrix}
        \xi^1(\lambda_1)\mu^1(\lambda_1) & \xi^1(\lambda_1)\mu^2(\lambda_1) & \cdots & \xi^1(\lambda_1)\mu^{|\mathcal{P}|}(\lambda_1)& \xi^2(\lambda_1)\mu^1(\lambda_1) & \cdots\\
        \vdots & \vdots & \cdots & \vdots & \vdots \\
        \xi^1(\lambda_1)\mu^1(\lambda_s) & \xi^1(\lambda_1)\mu^2(\lambda_s) & \cdots & \xi^1(\lambda_1)\mu^{|\mathcal{P}|}(\lambda_s)& \xi^2(\lambda_1)\mu^1(\lambda_s) & \cdots \\
        \xi^1(\lambda_2)\mu^1(\lambda_1)  & \xi^1(\lambda_2)\mu^2(\lambda_1) & \cdots & \xi^1(\lambda_2)\mu^{|\mathcal{P}|}(\lambda_1)& \xi^2(\lambda_2)\mu^1(\lambda_1) & \cdots  \\
        \vdots & \vdots & \cdots & \vdots & \vdots
    \end{pmatrix}
\end{flalign*}
That is,
\begin{equation*}
    C_{[\TT^1]}=G^1(R^\top\otimes E)
\end{equation*}
where $G^1$ is the matrix of row-flattened transformation matrices. Given a preparation and an event, an entry of $C_{[\TT^1]}$ can be interpreted as the probability of the measurement outcome conditioned on a given transformation.

\end{document}